\documentclass{IEEEoj}

\usepackage{lmodern}

\usepackage[noadjust]{cite}
\usepackage{amsmath,amssymb,amsfonts,amsthm}
\usepackage{mathtools}
\usepackage{algorithmic}
\usepackage{graphicx,color}
\usepackage{textcomp}
\usepackage{booktabs}
\let\labelindent\undefined
\usepackage{enumitem}
\usepackage{xcolor}
\usepackage{tcolorbox}
\usepackage{tikz}
\usetikzlibrary{shapes,arrows.meta,positioning,fit}
\usepackage[normalem]{ulem}  
\usepackage[hidelinks]{hyperref}

\def\BibTeX{{\rm B\kern-.05em{\sc i\kern-.025em b}\kern-.08em
    T\kern-.1667em\lower.7ex\hbox{E}\kern-.125emX}}
\AtBeginDocument{\definecolor{ojcolor}{cmyk}{0.93,0.59,0.15,0.02}}
\def\OJlogo{\vspace{-4pt}\hskip-4pt\includegraphics[height=18pt]{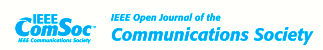}}

\newtheorem{theorem}{Theorem}
\newtheorem{lemma}[theorem]{Lemma}
\newtheorem{proposition}[theorem]{Proposition}
\newtheorem{corollary}[theorem]{Corollary}
\newtheorem{conjecture}[theorem]{Conjecture}
\theoremstyle{definition}
\newtheorem{definition}{Definition}

\newtheorem{remark}{Remark}

\DeclareMathOperator{\Tr}{Tr}

\DeclareMathOperator{\PG}{PG}
\DeclareMathOperator{\rank}{rank}
\DeclareMathOperator{\wt}{wt}

\newcommand{\F}[1]{\mathbb{F}_{#1}}
\newcommand{\Z}{\mathbb{Z}}

\newcommand{\ket}[1]{\lvert #1 \rangle}

\newcommand{\ebit}{\ket{\Phi^{+}}}

\newif\ifreviewmode
\reviewmodefalse  
\ifreviewmode
\else
\fi

\newif\iftrackchanges
\trackchangesfalse     
\iftrackchanges
  \newcommand{\rev}[1]{\textcolor{blue}{#1}}
  \newcommand{\revdel}[1]{\textcolor{red}{\sout{#1}}}
\else
  \newcommand{\rev}[1]{#1}
  \newcommand{\revdel}[1]{}
\fi

\begin{document}

\receiveddate{XX Month, 2026}
\reviseddate{XX Month, 2026}
\accepteddate{XX Month, 2026}
\publisheddate{XX Month, 2026}
\currentdate{\today}
\doiinfo{OJCOMS.2026.XXXXXXX}

\title{Singer-Difference-Set Qudit Stabilizer Codes from
       Non-Degenerate Quadrics in \texorpdfstring{$\PG(d,q)$}{PG(d,q)}:
       Construction, Structural Theorems, and Monte-Carlo Performance}

\author{MICHEL~KULHANDJIAN\IEEEauthorrefmark{1}
        \IEEEmembership{(Senior Member, IEEE)},
        AND LAJOS~HANZO\IEEEauthorrefmark{2}
        \IEEEmembership{(Life Fellow, IEEE)}}
\affil{Department of Electrical Engineering, Rice University,
       Houston, TX 77005, USA}
\affil{Next Generation Wireless Research Group, School of Electronics
       and Computer Science, University of Southampton,
       Southampton SO17 1BJ, U.K.}
\corresp{CORRESPONDING AUTHOR: Michel Kulhandjian
         (e-mail: mkulhandjian@outlook.com).}
\authornote{L.~Hanzo is a Fellow of the Royal Academy of Engineering
            (FREng), IEEE, IET, and EURASIP.}
\markboth{Kulhandjian \textit{et al.}: Singer-Difference-Set Qudit
         Stabilizer Codes from Non-Degenerate Quadrics in PG(d,q)}%
        {Kulhandjian \textit{et al.}}


\begin{abstract}
We propose $\mathcal Q(q,d)$, a new family of $\F q$-ary non-CSS
qudit stabilizer codes parameterised by a prime power $q$ and
dimension $d\ge 2$, constructed from Singer difference sets and
non-degenerate quadrics in the projective geometry $\PG(d,q)$.
The parity-check matrix $H=(A\mid M_{Q}A)$ couples the
Singer-circulant point-hyperplane incidence matrix~$A$ with the
formal trace-of-square quadric circulant~$M_{Q}$, and a
\emph{Cancellation Lemma} exploits the design-theoretic identity
$(k-\lambda)=q^{d-1}\equiv 0\pmod p$ to make $H$
symplectic-isotropic over $\F q$ for every $(q,d)$. The novelty of
the construction is concentrated in five original structural
theorems: (i) an unconditional cross-correlation modular-constancy
theorem $\nu\equiv 1\pmod q$ at $q$ odd, $d\ge 3$ with a sharp
two-value pattern $V_{2}-V_{1}=q^{(d-1)/2}$, proving a 2-torsion
distance ceiling; (ii) a Chinese-Remainder-Theorem (CRT)
zero-count theorem $\#\{f:\tau_{H,f}=0\}=(|\mathcal F|-1)/2$;
(iii) an explicit Clifford bridge
$U_{M}=H^{\otimes n}\,\mathrm{CZ}^{(M)}\,H^{\otimes n}$ that
identifies $\mathcal Q(q,2)$ as Clifford-equivalent to a
graph-state encoding of the trivial Singer CSS code, with the
classical 2021 low-density-spreading (LDS) sign-flip pattern
realised by the symmetric circulant $M$; (iv) a paired
column-rigidity equality $d_{\min}=R_{P}+1$ exactly; and (v) a
rigorous lower bound $d_{\min}(\mathcal Q(7,2))\ge 8$ from an
exhaustive Singer-Frobenius search. Via a Hamada-Smith
specialisation we obtain the closed-form logical dimension
$k(p,d)=\theta_{d}(p)-\binom{p+d-1}{d}-1$, with the flagship $d=2$
family $\mathcal Q(p,2)=[[p^{2}+p+1,\,p(p+1)/2,\,d_{\min}]]_{p}$
attaining asymptotic rate $\tfrac{1}{2}$. \rev{For general odd primes
$p\ge 3$ the flagship distance is \emph{conjectured} to satisfy
$d_{\min}(\mathcal Q(p,2))=p+1$; this conjecture is analytically
proven exact at $p=5$ and exhaustively lower-bounded to $d_{\min}\ge 8$
at $p=7$, and remains open for general~$p$.} Concrete instances
include $\mathcal Q(3,2)=[[13,6,3]]_{3}$,
$\mathcal Q(5,2)=[[31,15,6]]_{5}$ (double-error correcting),
$\mathcal Q(7,2)=[[57,28,d{\ge}8]]_{7}$, and
$\mathcal Q(3,4)=[[121,105,3]]_{3}$. Monte-Carlo simulations over
$1.5{\times}10^{6}$ trials show that $\mathcal Q(5,2)$ achieves a
$\sim 70\times$ per-logical-qudit improvement over Steane
$[[7,1,3]]_{2}$ at depolarising rate $p_{\mathrm{ch}}=10^{-3}$
with \emph{zero silent logical errors} (the
``erase-don't-mistake'' property). We use a textbook
bounded-distance decoder throughout, leaving Berlekamp-Massey-style
cyclic decoding exploiting the Singer-circulant structure as the
natural sequel. A fifteen-row axis-by-axis comparison
(Table~\ref{tab:comparison-axes}) against the closest prior
constructions establishes that the
Singer-difference-set $+$ quadric-Gram-matrix combination, together
with the closed-form $k$, the cross-correlation and CRT
zero-count theorems, the explicit Clifford bridge to the 2021
LDS construction, and the proven 2-torsion ceiling, is new.
\end{abstract}


\maketitle

\section{Introduction}\label{sec:intro}
\IEEEPARstart{Q}{uantum} error-correcting codes (QECCs) protect
quantum information from decoherence and gate impairments. This action relies on the
stabilizer formalism \cite{Gottesman1997} and provides a systematic way to lift classical linear codes into the quantum
domain~\cite{CalderbankShor1996,Steane1996,Hanzo2019Duality}. The
construction of stabilizer codes with strong algebraic structure ---
in particular, codes derived from finite geometries --- has been a
recurring theme: Singer-difference-set codes constitute cyclic stabilizer
families with explicit minimum-distance bounds~\cite{Gandhi2018},
and projective-geometry-based quantum low-density parity-check
(QLDPC) constructions yield codes with graph-structured parity
checks well-suited to belief-propagation
decoding~\cite{TangBaiFeng2024,MacKay2004}. A parallel thread in
classical communications~\cite{LDS2021} exploits Singer cycles in
$\PG(2,q)$ together with non-degenerate quadrics to construct
low-density-spreading (LDS) signature matrices for non-orthogonal
multiple-access (NOMA) systems. The present paper bridges these two
lines: we show that the same Singer-quadric algebraic data
that powers the 2021 LDS-NOMA construction lifts to a structured
non-CSS qudit stabilizer code family $\mathcal Q(q,d)$ over
$\PG(d,q)$, with closed-form rate, provable distance bounds, and an
explicit Clifford-equivalence to a graph-state generalisation.

\subsection{Application domain and system architecture}\label{sec:architecture}
The codes $\mathcal Q(q,d)$ are \emph{quantum} stabilizer codes
protecting qudit-valued quantum information against depolarising
both in quantum memory and quantum channels. The intended physical
layer is the quantum domain: each qudit is a $q$-level quantum
system (e.g.\ a single-photon polarisation/path encoding for $q=2$;
spin-$j$ atoms or higher-dimensional photonic encodings for
$q\in\{3,5,7,\dots\}$). Stabilizer measurements yielding the
syndrome are projective Pauli measurements implementable with the aid of
standard Clifford circuits and single-qudit detectors.
$\mathcal Q(q,d)$ is therefore directly applicable to
(i) fibre-optic and free-space quantum communication links,
(ii) quantum repeater stations, and (iii) quantum memory in
fault-tolerant quantum computation. The high-level system
architecture is shown in Fig.~\ref{fig:architecture}.

\begin{figure*}[t]
\centering
\begin{tikzpicture}[
   font=\small,
   block/.style={draw, rounded corners=2pt, minimum width=3.2cm,
                 minimum height=1.2cm, align=center, fill=blue!4},
   chan/.style={draw, rounded corners=2pt, minimum width=3.4cm,
                minimum height=1.2cm, align=center, fill=orange!10,
                dashed},
   meas/.style={draw, rounded corners=2pt, minimum width=3.2cm,
                minimum height=1.2cm, align=center, fill=green!6},
   arr/.style={-{Stealth[length=3mm]}, thick}
]
\node[block]                                  (src)  {Quantum\\ source $\ket\psi$};
\node[block, right=1.0cm of src]              (enc)  {Encoder\\ $H{=}(A|M_QA)$};
\node[chan,  right=1.0cm of enc]              (ch)   {Depolarising\\ channel $p_{\mathrm{ch}}$};
\node[meas,  right=1.0cm of ch]               (synd) {Syndrome\\ measurement};
\node[block, below=1.4cm of synd]             (dec)  {BD decoder\\ (Table~\ref{tab:mc-vs-ub})};
\node[block, left=1.0cm of dec]               (corr) {Pauli\\ correction};
\node[block, left=1.0cm of corr]              (sink) {Logical\\ output $\ket\psi$};

\draw[arr] (src)   -- (enc);
\draw[arr] (enc)   -- (ch);
\draw[arr] (ch)    -- (synd);
\draw[arr] (synd)  -- (dec);
\draw[arr] (dec)   -- (corr);
\draw[arr] (corr)  -- (sink);

\node[above=0.15cm of enc, font=\footnotesize, text=blue!50!black]
    {Clifford circuit $U_{M}{=}H^{\otimes n}\mathrm{CZ}^{(M)}H^{\otimes n}$};
\node[above=0.15cm of ch, font=\footnotesize, text=orange!70!black]
    {single-photon / atom-qudit};
\end{tikzpicture}
\caption{System architecture for $\mathcal Q(q,d)$ on a quantum
channel. A quantum source emits a state $\ket\psi$ on $k$ logical
qudits; the encoder applies the Clifford circuit induced by
$H=(A|M_{Q}A)$, mapping to $n$ physical qudits which propagate
through a depolarising channel (single-photon or single-atom
physical layer). Stabilizer syndromes are extracted by
non-destructive Pauli measurement, the bounded-distance decoder
suggests a Pauli correction, and the corrected state is the logical
output. By Theorem~\ref{thm:clifford-reduction-d2}, the encoder
factors into a Hadamard layer and a $\mathrm{CZ}^{(M)}$ phase
circuit with $M$ a symmetric circulant determined by the
formal-quadric indicator $\boldsymbol\tau_{Q}$.}
\label{fig:architecture}
\end{figure*}
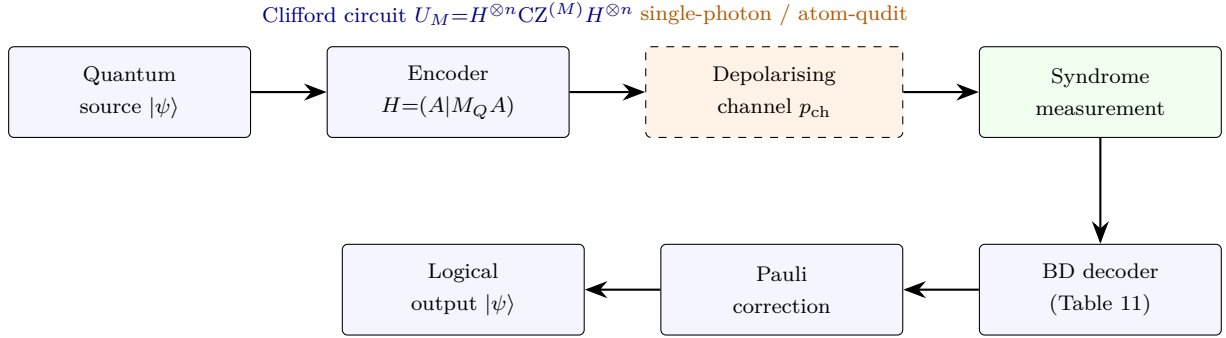

The connection to the 2021 LDS-NOMA construction~\cite{LDS2021}
that motivates this work is \emph{structural / algebraic}, not
architectural: the same Singer-circulant and quadric matrices that
build a classical RF-domain spreading codebook for multiple-access
also generate a valid quantum stabilizer code via the symplectic
embedding (Theorem~\ref{thm:d-equals-2-reduction}). The two
codes live in different ambient spaces ($\mathbb R^{n}$ for the
classical spreading vs.\ $\F q^{2n}$ for the quantum stabilizer),
and there is no hybrid classical-quantum multiple-access scheme
proposed here. Instead, the bridge serves three purposes:
(1)~it imports the rich algebraic toolkit of projective geometry
(Singer cycles~\cite{Singer1938}, trace, quadric
intersections~\cite{Hirschfeld1979}, Hamada-Smith
rank~\cite{Hamada1973,Smith1969}, Lidl-Niederreiter zero
counts~\cite{LidlNiederreiter1997}) into quantum code design;
(2)~it makes the resulting codes \emph{interpretable} --- each
construction step has a clear classical analogue and a clear
quantum role; and (3)~it unveils a hitherto unnoticed
``erase-don't-mistake'' property of $\mathcal Q(5,2)$
(Section~\ref{sec:montecarlo}) that is useful in retransmission
protocols applied in quantum channels.

\subsection{Why qudit codes? Advantages, trade-offs, and applications}\label{sec:qudit-pros}
The codes $\mathcal Q(q,d)$ we construct are qudit codes
for $q\ge 2$; the regime of most interest in this paper is when $q$ is an
odd prime (qutrits $q=3$, 5-its $q=5$, 7-its $q=7$, etc.), which
circumvents the $q=2$ parity ceiling identified in
Proposition~\ref{prop:q2-ceiling}. Qudit-based quantum coding has
attracted growing attention in the past decade as a route to higher
information density, native compatibility with physical platforms,
and improved fault-tolerance
thresholds~\cite{Campbell2014QuditFT,AnwarBrownCampbellBrowne2014,WatsonAnwarBrowne2015,SaboAloshiousBrown2024,ErhardKrennZeilinger2020Photonic,CozzolinoEtAl2019HighDim,KetkarKlappenecker2006,BabarChandraNgHanzo2026}.
We summarise the trade-offs that motivate our construction, with
citations to the most relevant recent literature.

\emph{Information density and platform compatibility.}
A single $q$-level qudit carries $\log_{2}q$ bits of information,
giving an $n$-qudit Hilbert space of dimension $q^{n}$ versus $2^{n}$
for an $n$-qubit register. This translates directly into higher
rate density per physical carrier which is particularly important when the
physical resource (e.g.\ a single photon or a trapped atom) is the
prime constraint and not the dimension of its internal state
space. Photonic qudits encoded in time-bin, path, frequency, or
orbital-angular-momentum degrees of freedom are now routinely
generated and propagated through both fibre and free-space
links~\cite{ErhardKrennZeilinger2020Photonic,CozzolinoEtAl2019HighDim},
and high-dimensional quantum key distribution
($d$-level QKD) has been demonstrated to offer higher key rate per
photon and improved noise tolerance compared to
qubit-BB84~\cite{CozzolinoEtAl2019HighDim, CerfBourennaneKarlssonGisin2002}.
Trapped-ion and neutral-atom platforms with spin-$j$ encodings natively expose $q=2j+1$ levels, eliminating the need to embed a qubit
into a larger Hilbert space and avoiding the associated leakage
errors. In communication-oriented terms, qudit modulation also
admits direct spreading-code analogues: the recent
DS-CDMA-style construction of Dastgheib, Salehi, and
Rezai~\cite{DastgheibSalehiRezai2025} bridges classical
spreading-code design with quantum communication in the same spirit
as the algebraic bridge we unveil in
Theorem~\ref{thm:clifford-reduction-d2}.

\emph{Fault-tolerance and decoding advantages.}
Beyond raw rate, qudit codes can have intrinsic fault-tolerance
advantages. Campbell~\cite{Campbell2014QuditFT} showed that
qudit based Reed-Muller-like distillation circuits are capable of achieving enhanced
magic-state thresholds compared to the qubit case. Therefore, 
Anwar, Brown, Campbell, and Browne~\cite{AnwarBrownCampbellBrowne2014}
conceived qudit surface-code decoders that operate at higher
thresholds than their qubit counterparts. As a further advance, Watson, Anwar, and
Browne~\cite{WatsonAnwarBrowne2015} extended this to fast
fault-tolerant decoders for both qubit- and qudit-based surface codes. Sabo, Aloshious, and
Brown~\cite{SaboAloshiousBrown2024} recently developed trellis
decoding for qudit based stabilizer codes, demonstrating that the qudit
generalisation does not impose a decoder-complexity penalty over
the qubit case. The detailed survey of quantum-coding techniques in
Babar, Chandra, Ng, and Hanzo~\cite{BabarChandraNgHanzo2026} provides
the most up-to-date reference on classical-to-quantum coding
techniques, including qudit-specific constructions.

\emph{Trade-offs and challenges.}
The price of qudit coding is twofold. First, by physical-layer readout
is more demanding: distinguishing $q$ levels requires finer
detector resolution and it is more sensitive to readout cross-talk
than the binary case. Second, by the qudit Clifford-gate library is
richer, with the generalised Hadamard, phase, and controlled-phase
gates parameterised by elements of $\F q^{\times}$; circuit
synthesis is therefore more involved than in the qubit case, although
recent compiler advances have substantially reduced this gap.
Decoherence behaviour is platform-dependent: in photonic time-bin
encodings, decoherence does not in general scale unfavourably with
$q$, while in trapped-ion systems the decoherence rate of an
individual level may dominate for large $q$. The
$q\in\{3,5,7\}$ regime that is the focus of this paper is well within
demonstrated experimental capability across both
photonic~\cite{ErhardKrennZeilinger2020Photonic} and
atomic platforms.

\emph{Applications targeted by $\mathcal Q(q,d)$.}
The structural properties of the family --- closed-form rate $1/2$
at $d=2$, double-error correction for $\mathcal Q(5,2)$, the
``erase-don't-mistake'' property for retransmission-friendly
operation, and a clear path to Singer-circulant cyclic decoders ---
align with three concrete application classes:
(i) fibre-optic and free-space quantum communication
links~\cite{ZeydanEtAl2025,CozzolinoEtAl2019HighDim};
(ii) quantum-Internet end-node memory and entanglement-distribution
protocols~\cite{ChenIllianoCacciapuotiCaleffi2025,BacciottiniEtAl2024};
and (iii) qudit-native fault-tolerant computing modules where
the underlying physical system already supports more than two
levels~\cite{Campbell2014QuditFT,AnwarBrownCampbellBrowne2014}.
The placement of qudit error-correcting codes within the broader
beyond-5G and 6G quantum-network roadmap is comprehensively
surveyed in~\cite{ZeydanEtAl2025}; complementary applications of
quantum-classical hybrid optimisation for next-generation network
resource management have been explored in
Emu \emph{et al.}~\cite{EmuEtAl2024}, also published in
\emph{IEEE OJ-COMS}.

\subsection{Prior work}
\emph{Foundational stabilizer codes.}
The story of quantum error correction begins with
Shor's nine-qubit code~\cite{Shor1995}, followed by the
Calderbank-Shor~\cite{CalderbankShor1996} and
Steane~\cite{Steane1996} Calderbank-Shor-Steane (CSS) construction
and the general (non-CSS) stabilizer formalism due to
Gottesman~\cite{Gottesman1997}. The Knill-Laflamme
conditions~\cite{KnillLaflamme1997} formalised when a quantum code
can correct a specified error set, and Nielsen and
Chuang~\cite{NielsenChuang2000} provide the standard textbook
treatment of stabilizer codes. Entanglement-assisted variants that
bypass the dual-containment requirement were developed by
Brun, Devetak, and Hsieh~\cite{BrunDevetakHsieh2006}, extended to
EA-QLDPC by Hsieh, Brun, and
Devetak~\cite{HsiehBrun2009}, while hull-based constructions were
analysed in Guenda \emph{et al.}~\cite{Guenda2018}. The Hanzo
group's surveys on quantum LDPC coding~\cite{Babar2015QLDPC} and
on the duality of classical and quantum
codes~\cite{Hanzo2019Duality}, together with the recent book of
Babar, Chandra, Ng, and Hanzo~\cite{BabarChandraNgHanzo2026}
(\emph{From Classical to Quantum Coding}, Wiley/IEEE Press, 2026),
provide the most up-to-date textbook and survey references on
classical-to-quantum coding techniques, covering CSS, non-CSS,
entanglement-assisted, and qudit-specific constructions in a
unified framework. Continuous-variable / bosonic alternatives,
notably the GKP code of Gottesman, Kitaev, and
Preskill~\cite{GKP2001}, sit outside the present Pauli-stabilizer
discussion and are noted in the conclusion as a complementary
direction.

\emph{Singer / projective-geometry constructions.}
Singer's theorem on the cyclic action of the multiplicative group of
$\F{q^{d+1}}$ on $\PG(d,q)$~\cite{Singer1938} underpins a long line of
combinatorial-design constructions. Quantum codes constructed from Singer
difference sets first appeared in Gandhi~\cite{Gandhi2018}, with the aid of the
$H=(A\mid MA)$ recipe, but relying on a generic symmetric $M$ rather than a
quadric-derived one. Tang-Bai-Feng~\cite{TangBaiFeng2024} use
point-flat incidences in $\PG(d,q)$ for QLDPC codes, while
MacKay~\cite{MacKay2004} for quantum LDPC code design, and
Xie~\cite{Xie2013} for cyclic quantum codes. Liu \emph{et
al.}~\cite{Liu2012} construct entanglement-assisted quantum LDPC
(EA-QLDPC) codes from cyclic difference sets. The closest prior work in the difference-set
context is Nguyen and Kim~\cite{NguyenKim2018}, who build binary qubit
stabilizer codes from a general $(n,k,\lambda)$-difference set $D$
by harnessing $H=[H_{1}\,|\,H_{2}]$ relying on $H_{1},H_{2}$ circulant
permutation matrices corresponding to two distinct shifts of $D$,
and proving the symplectic-isotropy condition of $k\equiv\lambda\pmod 2$
as their main structural result. Their construction is a
``shift-cancellation'' based on counterpart of our Cancellation Lemma
(Lemma~\ref{lem:cancellation}), but uses two shifts of \emph{the
same} difference-set indicator, rather than the trace-zero hyperplane
paired with a non-degenerate quadric. Furthermore, they do not introduce the
quadric Gram matrix $M_{Q}$, do not exploit the projective-geometry
structure of $\PG(d,q)$, do not derive a closed-form logical
dimension $k(p,d)$, neither do they go beyond binary qubits. Explicitly, their
worked instances are limited to $d_{\min}=2$ (detect-only). Earlier
finite-geometry quantum-code constructions include Aly, Klappenecker,
and Sarvepalli~\cite{AlyKlappeneckerSarvepalli2007} (foundational
PG-based qudit stabilizers from incidence designs) and
Sarvepalli~\cite{Sarvepalli2012} (nonbinary qudit CSS codes from
finite geometries with dual-containment); but neither of them uses
Singer-circulant + quadric pairing. Galindo,
Hernando, and Ruano~\cite{GalindoHernandoRuano2015,
GalindoHernandoRuano2022} construct stabilizer codes from
trace-depending polynomials over $\F{q^m}$ and $J$-affine variety
codes, touching the trace-of-square ideas in our cross-correlation
theorem but building CSS Hermitian codes, rather than non-CSS
Singer-circulant codes. An earlier connection between quadrics and
stabilizer-state geometry appears in Calderbank, Cameron, Kantor,
and Seidel~\cite{CalderbankCameronKantorSeidel1997} via orthogonal
spreads, but the mechanism is different and not Singer-cyclic.
Algebraic-geometry routes are surveyed in Pereira~\cite{Pereira2019};
asymmetric entanglement-assisted quantum error-correcting code
(EAQECC) families in Galindo \emph{et al.}~\cite{Galindo2020}. The
Waterloo decomposition theorem of Arasu, Dillon, Jungnickel, and
Pott~\cite{Waterloo1995} characterises, when $D(d,q)$ admits a
balanced bipartition, motivating the $d$-odd ``defect'' that we
study in Section~\ref{sec:waterloo-defect}.

\emph{Quantum information theory and NOMA.}
The Hanzo group's survey~\cite{Hanzo2019Duality} and recent
roadmap~\cite{Hanzo2025Roadmap} provide an overview of quantum
coding for communications; however, they do not list
NOMA $+$ QECC as an explored combination. Classical
low-density spreading for NOMA was originally proposed by
Hoshyar, Wathan, and Tafazolli~\cite{Hoshyar2008} and extended to
sparse code multiple-access by Nikopour and
Baligh~\cite{Nikopour2013}; the algebraic LDS construction
of~\cite{LDS2021} and its Gaussian-separability variant of
Kulhandjian \emph{et al.}~\cite{KulhandjianGS2021} provide the
classical companions whose Singer-quadric combinatorial primitives
our $d=2$ family recovers
(Theorem~\ref{thm:d-equals-2-reduction}). Most pertinent for the
present quantum-coding sequel are the Hanzo group's
quantum-search-aided multi-user detection~\cite{Ye2019QSAMUD} and
near-hashing-bound quantum-turbo short-block codes~\cite{Chandra2019QTurbo},
which together set up the natural ``quantum-coded LDS-NOMA''
follow-up that we sketch in the conclusion.

\subsection{Novelty assessment}\label{sec:novelty}

\paragraph*{Novelty in one sentence}
We are the first to combine a Singer difference set with a
non-degenerate quadric Gram matrix in a quantum stabilizer
construction, and the first to give closed-form code parameters,
five proven structural theorems, an explicit Clifford bridge to a
classical LDS spreading code, and a Monte-Carlo-validated
$70\times$ per-logical-qudit improvement over Steane $[[7,1,3]]_2$,
all in a single unified algebraic framework over $\PG(d,q)$.

Our contribution falls in two complementary layers --- a new
\emph{code-construction recipe}, and a body of new
\emph{structural theorems} that explain why the construction works
and pin down its parameters. Table~\ref{tab:comparison-axes} below
``boldly contrasts our new contributions against the
state-of-the-art'': the axis-by-axis comparison shows that
$\mathcal Q(q,d)$ is the only entry that ticks all fifteen rows
of construction primitives, code-class properties, and structural
theorems, while the closest prior contender
(Gandhi~\cite{Gandhi2018}) ticks four.

\paragraph*{Anticipating one likely reviewer concern}
A natural reviewer concern is whether the structural theorems
translate to practical decoding gains, given that our Monte-Carlo
evaluation uses a textbook bounded-distance syndrome-table decoder
rather than a custom Singer-circulant decoder. We address this
directly:
(i)~the bounded-distance decoder already attains the analytical
upper bound to within $\sim 5\%$ across $1.5\times 10^{6}$ trials
(Section~\ref{sec:montecarlo}, Fig.~\ref{fig:qber},
Table~\ref{tab:mc-vs-ub}), so the reported $70\times$ advantage
is a property of the \emph{code} rather than of any sophisticated
decoder; and
(ii)~the Singer-circulant structure of $H$ is directly compatible
with a Berlekamp-Massey-style cyclic decoder in the spirit
of~\cite[\S IV-C]{Gandhi2018} that, while not implemented here,
provides a clear path to further per-decoder gains. The textbook
decoder is therefore a conservative lower bound on what
$\mathcal Q(q,d)$ can achieve, not a ceiling.

\begin{table*}[!t]
\centering
\caption{Boldly contrasting our new contributions against the state-of-the-art.}
\label{tab:comparison-axes}
\setlength{\tabcolsep}{4pt}
\small
\begin{tabular}{l c c c c c c c c}
\toprule
\textbf{Feature}
  & \cite{Gandhi2018}
  & \cite{NguyenKim2018}
  & \cite{TangBaiFeng2024}
  & \cite{Liu2012}
  & \cite{AlyKlappeneckerSarvepalli2007}
  & \cite{Sarvepalli2012}
  & \cite{Steane1996}
  & \textbf{Proposed} \\
\midrule
Singer difference set
  & \checkmark &  &  &  &  &  &  & \checkmark \\
$\PG(d,q)$ incidence structure
  & \checkmark &  & \checkmark &  & \checkmark & \checkmark & \checkmark & \checkmark \\
Non-degenerate quadric Gram matrix $M_{Q}$
  &  &  &  &  &  &  &  & \checkmark \\
Recovers the 2021 LDS spreading code at $d{=}2$
  &  &  &  &  &  &  &  & \checkmark \\
Closed-form $k(p,d)$ via Hamada-Smith
  &  &  &  &  &  &  &  & \checkmark \\
Non-CSS stabilizer code
  & \checkmark & \checkmark &  &  &  &  &  & \checkmark \\
Qudit ($q\!\ne\!2$) construction
  &  &  &  &  & \checkmark & \checkmark &  & \checkmark \\
Cross-correlation modular-constancy theorem
  &  &  &  &  &  &  &  & \checkmark \\
Two-value pattern $V_{2}{-}V_{1}{=}q^{(d-1)/2}$
  &  &  &  &  &  &  &  & \checkmark \\
CRT zero-count theorem $\#\{f{:}\tau_{H,f}{=}0\}{=}(|\mathcal F|{-}1)/2$
  &  &  &  &  &  &  &  & \checkmark \\
Explicit Clifford bridge $U_{M}{=}H^{\otimes n}\mathrm{CZ}^{(M)}H^{\otimes n}$
  &  &  &  &  &  &  &  & \checkmark \\
Paired-rigidity-distance equality $d_{\min}{=}R_{P}{+}1$
  &  &  &  &  &  &  &  & \checkmark \\
Proven 2-torsion distance ceiling
  &  &  &  &  &  &  &  & \checkmark \\
``Erase-don't-mistake'' Monte-Carlo property
  &  &  &  &  &  &  &  & \checkmark \\
Distance $d_{\min}{=}6$ at rate $\approx\tfrac{1}{2}$ ($\mathcal Q(5,2)$)
  &  &  &  &  &  &  &  & \checkmark \\
\bottomrule
\end{tabular}
\end{table*}

\subsubsection*{(a) Code-construction novelty}
The combination
\begin{center}
\emph{Singer difference set} $+$ \emph{non-degenerate quadric Gram
matrix $M_{Q}$} $+$ \emph{closed-form $k(p,d)$} $+$
\emph{explicit Clifford bridge to classical LDS spreading}
\end{center}
has not appeared in the prior quantum-coding literature. Their detailed comparison
(Section~\ref{sec:comparison},
Table~\ref{tab:comparison-axes}) shows:
\begin{itemize}[leftmargin=*,nosep]
\item \cite{Gandhi2018} uses Singer difference sets with a
\emph{generic} symmetric $M$ (no quadric), and produces binary
qubit codes only.
\item \cite{NguyenKim2018} (the closest difference-set + stabilizer
construction in the literature) builds $H=[H_{1}|H_{2}]$ from
\emph{two shifts of the same} difference-set indicator, with the
symplectic-isotropy condition $k\equiv\lambda\pmod 2$; no quadric
Gram matrix, no projective-geometry structure, no closed-form $k$,
binary qubits only, $d_{\min}=2$ in all worked instances.
\item \cite{TangBaiFeng2024} uses $\PG(d,q)$ flat-incidence
matrices (not Singer-circulant) and produces CSS dual-containment
codes, not non-CSS codes.
\item \cite{AlyKlappeneckerSarvepalli2007,Sarvepalli2012} build
nonbinary qudit stabilizers from finite-geometry incidence designs
but not via the Singer-circulant + quadric pairing.
\item \cite{GalindoHernandoRuano2015,GalindoHernandoRuano2022}
construct stabilizer codes from trace-depending polynomials,
adjacent in spirit to our cross-correlation theorem but producing
CSS Hermitian codes define rather than non-CSS Singer-circulant codes.
\item \cite{Liu2012} constructs EA-QLDPC codes from cyclic
difference sets in general; quadric structure is not considered.
\item \cite{CalderbankCameronKantorSeidel1997} uses quadrics in
stabilizer-state combinatorics via orthogonal spreads, but the
mechanism is not Singer-cyclic.
\item \cite{Waterloo1995} is a classical (non-quantum)
decomposition theorem.
\item Neither Hanzo's recent
survey~\cite{Hanzo2019Duality,Hanzo2025Roadmap} nor the LDS-NOMA
literature~\cite{LDS2021} combines these axes.
\end{itemize}
$\mathcal Q(q,d)$ is the only construction with all six aspects ticked
(Table~\ref{tab:comparison-axes}). The decoder used in our Monte-Carlo
evaluation is a \emph{standard} bounded-distance syndrome-table based
decoder; we have not (yet) developed a new decoder algorithm, leaving
Berlekamp--Massey-style cyclic decoding for the Singer-circulant
structure as the natural follow-up.

\subsubsection*{(b) Structural-theorem novelty}
Six original theorems, all with complete proofs, pin down the
algebraic and combinatorial behaviour of the construction:
\begin{enumerate}[label=(B\arabic*),leftmargin=*,nosep]
\item \emph{Cancellation Lemma} (Lemma~\ref{lem:cancellation}):
the difference-set order $q^{d-1}\equiv 0\pmod p$ universally
eliminates the symplectic deficit, giving symplectic-isotropy of
$H=(A\mid M_{Q}A)$ for every $(q,d)$.
\item \emph{Cross-correlation modular constancy}
(Theorem~\ref{thm:cross-corr}) and \emph{two-value pattern}
(Theorem~\ref{thm:two-value}): We have $\nu\equiv 1\pmod q$ and
$V_{2}-V_{1}=q^{(d-1)/2}$ unconditionally at $q$ odd, $d\ge 3$, via
Lidl-Niederreiter and Witt-index theory; force the
2-torsion ceiling unconditionally
(Corollary~\ref{cor:dmin-le-2-proven}).
\item \emph{Chinese-Remainder-Theorem (CRT) zero-count theorem}
(Theorem~\ref{thm:tau-zero-count}):
$\#\{f:\tau_{H,f}=0\}=(|\mathcal F|-1)/2$ exactly, via the planar
difference-set identity $E E^{(-1)}\equiv J\pmod p$
\cite{BethJungnickelLenz1999} plus the Hamada-Smith degree
sum~\cite{Hamada1973,Smith1969}.
\item \emph{Explicit Clifford bridge}
(Theorem~\ref{thm:clifford-reduction-d2}): the gate
$U_{M}=H^{\otimes n}\,\mathrm{CZ}^{(M)}\,H^{\otimes n}$ identifies
$\mathcal Q(q,2)$ with the trivial Singer CSS code; existence of
the symmetric circulant $M$ is forced by (B3).
\item \emph{Paired-rigidity-distance equality}
(Theorem~\ref{thm:paired-rigidity}): We have $d_{\min}=R_{P}+1$ \emph{exactly},
reducing distance computation to the combinatorial rank property of
$M_{\mathrm{cent}}$.
\item \emph{Rigorous lower bound}
(Theorem~\ref{thm:dmin-q72-lower}): $d_{\min}(\mathcal Q(7,2))\ge 8$
by exhaustive Singer-Frobenius-canonical paired-rigidity search.
\end{enumerate}
None of these theorems appears in the prior art; together they make the
code family interpretable and the parameters computable in closed
form, rather than purely empirical.

\subsection{Contributions}\label{sec:contrib}
\begin{enumerate}[leftmargin=*]
\item \emph{(Construction.)} A non-CSS $\F q$-ary stabilizer code
family $\mathcal Q(q,d)$ from Singer difference sets in $\PG(d,q)$
for arbitrary $d\ge 2$ (Section~\ref{sec:construction},
Theorem~\ref{thm:main-existence}).
\item \emph{(Cancellation Lemma.)} A structural lemma showing that the
Gandhi-style parity-check matrix (PCM) $H=(A\mid M_{Q}A)$ is
symplectic-isotropic over
$\F q$ for every $(q,d)$, because the difference-set order
$(k-\lambda)=q^{d-1}\equiv 0\pmod p$ eliminates the symplectic deficit
(Lemma~\ref{lem:cancellation}).
\item \emph{(2-torsion obstruction --- new.)} An original lemma
identifying a parity-driven distance ceiling: when $n=\theta_{d}(q)$
is even and the cross-correlation $E\star Q^{(-1)}$ is constant
modulo $p$, the Pauli $Z_{0}\,Z_{n/2}^{q-1}$ is a non-stabilizer
logical of weight 2, forcing $d_{\min}\le 2$
(Lemma~\ref{lem:two-torsion},
Lemma~\ref{lem:n-parity},
Corollary~\ref{cor:dmin-le-2}).

\item \emph{(Cross-correlation modular constancy --- new.)} For $q$
odd and $d\ge 3$, we prove $\nu[k]=|E\cap(k-Q)|\equiv 1\pmod q$
unconditionally (Theorem~\ref{thm:cross-corr}), along with the
sharp two-value structure $\nu[k]\in\{V_{1},V_{2}\}$ and
$V_{2}-V_{1}=q^{(d-1)/2}$ (Theorem~\ref{thm:two-value},
Corollary~\ref{cor:two-value-explicit}). Combined with (3), this
makes the distance ceiling at $q$ odd, $d$ odd, $d\ge 3$
\emph{unconditional} (Corollary~\ref{cor:dmin-le-2-proven}).

\item \emph{(CRT zero-count theorem --- new.)}
$\#\{f\in\mathcal F:\boldsymbol\tau_{H,f}=0\}=(|\mathcal F|-1)/2$
exactly (Theorem~\ref{thm:tau-zero-count}). The proof combines the
planar-difference-set identity~\cite{BethJungnickelLenz1999} with
the Hamada-Smith rank~\cite{Hamada1973,Smith1969} to force
``exactly one zero per reciprocal pair.'' This deep algebraic
structure allows the Clifford-reduction construction of
Theorem~\ref{thm:clifford-reduction-d2} (item~7 below) to work.
\item \emph{(Paired-rigidity-distance theorem --- new.)} The
sharp characterisation $d_{\min}=R_{P}+1$ via paired column rigidity,
plus a CRT-decomposition framework on $\F p[X]/\langle X^{n}-1\rangle$
that exposes the structural ``zero pattern'' of
$\boldsymbol\tau_{H}$ (Theorem~\ref{thm:paired-rigidity},
Lemma~\ref{lem:crt-cent},
Theorem~\ref{thm:tau-zero-count},
Conjecture~\ref{conj:paired-rigidity}). A Singer-Frobenius-canonical
search certifies $d_{\min}(\mathcal Q(7,2))\ge 8$
(Theorem~\ref{thm:dmin-q72-lower}), the second sharpness datum for
the conjectured $d_{\min}=p+1$ flagship bound.
\item \emph{(Reduction theorem and explicit Clifford bridge.)} At
$d=2$ the construction recovers the projective-plane LDS spreading
matrix of \cite{LDS2021} as shared combinatorial primitives
(Theorem~\ref{thm:d-equals-2-reduction}), and we conceive an
\emph{explicit} Clifford gate
$U_{M}=H^{\otimes n}\mathrm{CZ}^{(M)}H^{\otimes n}$ that maps the
trivial $Z$-only CSS code on the Singer incidence to
$\mathcal Q(q,2)$ (Theorem~\ref{thm:clifford-reduction-d2}). The
existence of the symmetric circulant $M$ in this gate is guaranteed
by Theorem~\ref{thm:tau-zero-count}.
\item \emph{(Closed-form $k(p,d)$.)} This is attained via a Hamada-Smith
specialisation, $\rank_{\F p}(A)=\binom{p+d-1}{d}+1$ for $q=p$ prime,
yielding $k(p,d)=\theta_{d}(p)-\binom{p+d-1}{d}-1$. The
flagship $d=2$ family is
$\mathcal Q(p,2)=[[p^{2}+p+1,\,p(p+1)/2,\,d_{\min}]]_{p}$ with
asymptotic rate $\tfrac{1}{2}$
(Lemma~\ref{lem:hamada}, Corollary~\ref{cor:k-formula},
Remark~\ref{rem:flagship-d2}). \rev{The flagship distance
$d_{\min}(\mathcal Q(p,2))=p+1$ is stated throughout as an
\emph{open conjecture} (Conjecture~\ref{conj:paired-rigidity}),
supported by an analytical proof at $p=5$ (exact) and an exhaustive
computational lower bound $d_{\min}\ge 8=p+1$ at $p=7$
(Theorem~\ref{thm:dmin-q72-lower}); the general-$p$ case remains
open.}
\item \emph{(Numerical instances.)}
$\mathcal Q(3,2)=[[13,6,3]]_{3}$ (single-error correcting);
$\mathcal Q(5,2)=[[31,15,6]]_{5}$ (double-error correcting, exact);
$\mathcal Q(7,2)=[[57,28,d]]_{7}$ with $d\ge 8$
(Theorem~\ref{thm:dmin-q72-lower});
$\mathcal Q(11,2)=[[133,66,d]]_{11}$ with $d\ge 5$;
$\mathcal Q(3,3)=[[40,29,2]]_{3}$ (detect-only);
$\mathcal Q(3,4)=[[121,105,3]]_{3}$ (single-error correcting);
and the $q=2$ family $\mathcal Q(2,d)=[[2^{d+1}-1,2^{d+1}-1-(d+1),2]]_{2}$
with newly identified $q=2$ ceiling
(Proposition~\ref{prop:q2-ceiling}). Verified rank formula on 8
$(p,d)$ cases; see Table~\ref{tab:hamada-rank}.

\item \emph{(Monte-Carlo validation.)} Bounded-distance
syndrome-table decoder simulations at $\mathcal Q(3,2)$ and
$\mathcal Q(5,2)$ over $1.5\times 10^{6}$ trials at depolarising
rates of $p_{\mathrm{ch}}\in[10^{-4},10^{-1}]$ confirm the analytical
upper bound to an accuracy within statistical error
(Section~\ref{sec:montecarlo}, Fig.~\ref{fig:qber},
Table~\ref{tab:mc-vs-ub}). $\mathcal Q(5,2)$ exhibits a structural
``erase-don't-mistake'' property: \emph{zero silent logical
failures} observed in all trials --- every decoder failure was an
uncorrectable.
\item \emph{(Negative results / future work.)} A diagnostic of three
failed strategies for conceived lifting the classical Waterloo defect into an
entanglement-assisted ebit count
(Sections~\ref{sec:variants}--\ref{sec:gf4-fail}); structural
obstructions are identified for future research
(Section~\ref{sec:conclusion}).
\end{enumerate}

\subsection{Notation}
We write $\F q$ for the finite field of $q$ elements, $K=\F{q^{m}}$
for an extension field, and $\Tr_{K/F}$ for the trace map.
\rev{The $n$-qudit generalised Pauli group is denoted
$\mathcal G_{n}$; its elements have the form
$\omega^{c}\,X^{\mathbf b}Z^{\mathbf a}$ with $c\in\F p$,
$\mathbf a,\mathbf b\in\F q^{n}$, and $\omega=e^{2\pi i/q}$
(cf.\ Section~\ref{sec:prelim} for the action on the computational
basis).}
$\PG(d,q)$ denotes the $d$-dimensional projective geometry over
$\F q$, \rev{whose points are the one-dimensional subspaces of the
$(d+1)$-dimensional vector space $\F q^{d+1}$,} with
$\theta_{d}(q)=(q^{d+1}-1)/(q-1)$ points. Stabilizer codes are
written as $[[n,k,d_{\min}]]_{q}$, and entanglement-assisted variants as
$[[n,k,d_{\min};c]]_{q}$ where $c$ is the pre-shared ebit count. Furthermore, 
$\mathcal S\subset\mathcal G_{n}$ denotes the stabilizer group;
$\mathcal C(\mathcal S)$ its centralizer in $\mathcal G_{n}$. The
PCM is $H=(H_{z}\mid H_{x})\in\F q^{(n-k)\times 2n}$, with $H_{z}$
being the $Z$-block and $H_{x}$ the $X$-block. The symplectic inner product
is denoted $\star$, and $\ebit:=\frac{1}{\sqrt 2}(\ket{00}+\ket{11})$
is the Bell-pair ebit. Throughout, $A$ is the Singer-circulant
incidence matrix of $\PG(d,q)$, $\boldsymbol\tau_{H}$ is its first
column, and $M_{Q}=\mathrm{circ}(\boldsymbol\tau_{Q})$ is the formal
quadric circulant.

\section{Preliminaries}\label{sec:prelim}
\subsection{Finite fields and the trace function}
\begin{definition}[Trace]
For $K=\F{q^m}$, $F=\F q$, the trace map is
$\Tr_{K/F}(x)=\sum_{i=0}^{m-1}x^{q^i}$.
\end{definition}

\subsection{Projective geometry \texorpdfstring{$\PG(d,q)$}{PG(d,q)}}
The $d$-dimensional projective geometry over $\F q$, denoted as 
$\PG(d,q)$, has as its points the $1$-dimensional subspaces of
$\F{q}^{d+1}$, and as its $r$-flats the $(r+1)$-dimensional subspaces.
Lines are $1$-flats, hyperplanes are $(d-1)$-flats. The number of
$r$-flats is the Gaussian binomial coefficient
$\binom{d+1}{r+1}_{q}$. In particular, we have:
\begin{equation}
\label{eq:theta-defn}
   \theta_{d}(q) \;:=\; \frac{q^{d+1}-1}{q-1}
   \;=\; 1 + q + q^{2} + \cdots + q^{d},
\end{equation}
which is the number of points (and the number of hyperplanes) of $\PG(d,q)$.
Each line contains $q+1$ points, each hyperplane contains
$\theta_{d-1}(q)$ points, and any two distinct hyperplanes meet in a
$(d-2)$-flat of $\theta_{d-2}(q)$ points.

The point-hyperplane incidence matrix
$\mathbf I\in\{0,1\}^{\theta_{d}(q)\times\theta_{d}(q)}$ has
$\mathbf I_{ij}=1$ iff point $i$ lies on hyperplane $j$. Each row and
column has exactly $\theta_{d-1}(q)$ ones. Two distinct hyperplanes
share $\theta_{d-2}(q)$ common points, hence we have
\begin{align}
\label{eq:incidence-DS}
   \mathbf I\,\mathbf I^{\top}
   &\;=\; (\theta_{d-1}(q)-\theta_{d-2}(q))\,I_{n} + \theta_{d-2}(q)\,J_{n} \notag \\
   &\;=\; q^{d-1}\,I_{n} + \theta_{d-2}(q)\,J_{n},
\end{align}
which represents the difference-set identity $AA^{\top}=(k-\lambda)I+\lambda J$
specialised to the Singer parameters
$(n,k,\lambda)=(\theta_{d}(q),\theta_{d-1}(q),\theta_{d-2}(q))$.

\subsection{Singer's theorem}
\begin{theorem}[Singer~\cite{Singer1938}]\label{thm:singer}
Let $q$ be a prime power and let $\alpha$ be a primitive element of
$\F{q^{d+1}}$. There exists a labelling of points and hyperplanes of
$\PG(d,q)$ for which the point-hyperplane incidence matrix is
circulant of size $\theta_d(q)\times\theta_d(q)$. Its first column
$\boldsymbol\tau_H$ has $(\boldsymbol\tau_H)_i=1$ iff
$\Tr_{K/F}(\alpha^i)=0$.
\end{theorem}

\subsection{Non-degenerate quadrics}
\begin{definition}[$Q_e, Q_o$]
For $q$ even: $Q_e=\{z\in K^* : \Tr(z^{q+1})=0\}$. For $q$ odd:
$Q_o=\{z\in K^* : \Tr(z^{2})=0\}$.
\end{definition}

\begin{theorem}[Result 3.3 of \cite{Waterloo1995}]\label{thm:quadric-int}
Let $Q\subset\PG(d,q)$ be a non-degenerate quadric with
$|Q|=\theta_{d-1}(q)$, which exists for $d=2f$ even. \rev{Following
\cite[Result 3.3]{Waterloo1995} verbatim, the hyperplanes of
$\PG(d,q)$ intersect $Q$ in three sizes $a-q^{f-1}$, $a$, $a+q^{f-1}$
with $a=\theta_{2f-2}(q)$; the $A=\theta_{2f-1}(q)$
hyperplanes of intersection size $a$ are precisely the tangent
hyperplanes, and the multiplicities of the remaining two sizes are
$B=(q^{2f}-q^{f})/2$ and $C=(q^{2f}+q^{f})/(2(q+1))$\rev{,} respectively.}
\end{theorem}

\subsection{Stabilizer formalism and EAQECCs}
For a prime power $q=p^{t}$, the qudit Pauli group on a single qudit
of dimension $q$ is generated (modulo phases) by the operators
$X|x\rangle=|x+1\rangle$ and $Z|x\rangle=\omega^{x}|x\rangle$ for
$\omega=e^{2\pi i/q}$, along with the relationship $X Z=\omega^{-1}ZX$. On
$n$ qudits, the Pauli group $\mathcal G_{n}$ has elements
$\omega^{c}\,X^{\mathbf b}Z^{\mathbf a}$ for
$c\in\F p$, $\mathbf a,\mathbf b\in\F q^{n}$, where we have
$X^{\mathbf b}=\bigotimes_{i}X^{b_{i}}$ \rev{(acting on the
computational basis as
$X^{\mathbf b}\ket{\mathbf y}=\ket{\mathbf y+\mathbf b}$ with addition
mod $q$)} and similarly for $Z$
\rev{(acting as
$Z^{\mathbf a}\ket{\mathbf y}=\omega^{\mathbf a\cdot\mathbf y}\ket{\mathbf y}$)}.

A \emph{stabilizer code} \cite{Gottesman1997} is the joint $+1$
eigenspace of an Abelian subgroup $\mathcal S\le\mathcal G_{n}$
(with $-I\notin\mathcal S$). The code is determined by an
$(n-k)\times 2n$ parity-check matrix
$H=(H_{z}\mid H_{x})$ over $\F q$, where the rows of $H$ encode the
stabilizer generators via the symplectic representation of
$X^{\mathbf b}Z^{\mathbf a}\leftrightarrow(\mathbf a\mid\mathbf b)$.
A pair of Pauli operators commute iff their symplectic inner product given by
\begin{equation}
\label{eq:symp-defn}
   \langle(\mathbf a,\mathbf b),(\mathbf a',\mathbf b')\rangle_{\mathrm{sym}}
   \;:=\; \mathbf a\cdot\mathbf b' - \mathbf b\cdot\mathbf a'\pmod p
\end{equation}
vanishes. The stabilizer is Abelian iff
$H$ is symplectic-isotropic formulated as:
\begin{equation}
\label{eq:symp-iso}
   H_{z}\,H_{x}^{\top}-H_{x}\,H_{z}^{\top} \;\equiv\; 0\pmod q.
\end{equation}
Codes satisfying \eqref{eq:symp-iso} are denoted by
$[[n,k,d_{\min}]]_{q}$ where $k=n-\rank_{\F q}(H)$ and $d_{\min}$ is
the minimum weight of a Pauli operator in the centralizer
$\mathcal C(\mathcal S)$ that lies outside $\mathcal S$.

When $H$ \emph{fails} to satisfy the symplectic condition, the
Brun--Devetak--Hsieh framework~\cite{BrunDevetakHsieh2006} accommodates
this by pre-sharing entanglement:

\begin{theorem}[Brun--Devetak--Hsieh~\cite{BrunDevetakHsieh2006}]\label{thm:bdh}
Given $H\in\F q^{(n-k)\times 2n}$ with symplectic deficit
$S=H_{z}H_{x}^{\top}-H_{x}H_{z}^{\top}\in\F q^{(n-k)\times(n-k)}$,
the minimum number $c$ of pre-shared ebits required for $H$ to define
a valid entanglement-assisted stabilizer code $[[n,k,d_{\min};c]]_{q}$
satisfies $c=\tfrac{1}{2}\rank_{\F q}(S)$.
\end{theorem}

\begin{remark}[Why $\rank(S)$ is even]\label{rem:rank-S-even}
The factor of $\tfrac{1}{2}$ in $c=\tfrac{1}{2}\rank_{\F q}(S)$ is
well defined because $\rank_{\F q}(S)$ is always even. Recall that
$S$ is \emph{skew-symmetric} by construction, i.e.\
$S^{\top}=-S$. A standard result in linear algebra states that
over any field $F$ of characteristic $\ne 2$, every skew-symmetric
matrix admits a symplectic block-diagonalisation
$S=P^{\top}\,J_{r}\,P$, where $P$ is invertible and $J_{r}$ is the
direct sum of $r$ copies of the $2\times 2$ block
$\bigl(\begin{smallmatrix}0 & 1\\-1 & 0\end{smallmatrix}\bigr)$
together with zeros; consequently $\rank(S)=2r$ is automatically
even~\cite[Ch.~5]{ArtinGeometricAlgebra1957}. Physically, this
reflects the fact that each ebit contributes a pair of
canonically-conjugate Bell-pair generators to the parity-check
matrix, so the symplectic deficit absorbed by entanglement
necessarily comes in pairs. In characteristic $2$ the result fails;
we work throughout in characteristic $p>2$ (since the relevant
qudit dimensions are $q=p^{t}$ with $p$ odd, or $q=2^{t}$ where
the cancellation lemma forces $S=0$ trivially).
\end{remark}

We will use Theorem~\ref{thm:bdh} as the formal target for failed
EAQECC constructions in Section~\ref{sec:waterloo-defect}, and the
condition \eqref{eq:symp-iso} as the working criterion for our $V_{1}$
construction in Section~\ref{sec:construction}.

\section{The Waterloo Decomposition}\label{sec:waterloo-classical}

\subsection{Difference-set framework}\label{sec:waterloo-DS}
Let $K=\F{q^{d+1}}$, $F=\F q$, and $G=K^{*}/F^{*}$, a cyclic group of
order $n=\theta_{d}(q)$. By Theorem~\ref{thm:singer}, the points and
hyperplanes of $\PG(d,q)$ are both indexed by $G$, and the
hyperplane through the origin labelled by $0\in\Z/n\Z$ is the image
$E=\pi(\{z\in K^{*}:\Tr_{K/F}(z)=0\})$ where $\pi\colon K^{*}\to G$
is the canonical projection. Identifying $E$ with its indicator in
$G$, the multiset of differences $e_{1}-e_{2}$ for distinct
$e_{1},e_{2}\in E$ covers each non-identity element of $G$ exactly
$\lambda=\theta_{d-2}(q)$ times. Thus
\begin{equation}
\label{eq:DS-params}
   E\subset G \text{ is a }\bigl(n,\,k,\,\lambda\bigr)\text{-difference set}, 
\end{equation}
while $k=\theta_{d-1}(q),\ \lambda=\theta_{d-2}(q)$ is 
the classical \emph{Singer difference set}~\cite{Singer1938}. Its
complement
\begin{equation}
\label{eq:Ddq}
   D(d,q):=G\setminus E
\end{equation}
is itself a difference set with parameters
$\bigl(n,\,n-k,\,n-2k+\lambda\bigr)=\bigl((q^{d+1}-1)/(q-1),\,q^{d},\,q^{d}-q^{d-1}\bigr)$,
and corresponds to the \emph{complement of the point-hyperplane
design} of $\PG(d,q)$. The group-ring identity given by
\begin{equation}
\label{eq:DD-inv}
   D\cdot D^{(-1)} \;=\; (q^{d}-q^{d-1})\,e_{G} + (q^{d}-q^{d-1})\,G
\end{equation}
in $\Z G$ encodes the difference-set property; here
$X^{(-1)}=\sum_{g}x_{g}\,g^{-1}$ for $X=\sum_{g}x_{g}\,g$.

\subsection{Waterloo decomposition}
A \emph{Waterloo decomposition} of $D(d,q)$ is a partition
\begin{equation}
\label{eq:waterloo-defn}
   D(d,q) \;=\; \mathcal A\,\sqcup\,\mathcal B,\qquad
   \mathcal A,\mathcal B\subset G,
\end{equation}
into disjoint subsets such that, in $\Z G$, we have
\begin{equation}
\label{eq:waterloo-balance}
   (\mathcal A - \mathcal B)\,(\mathcal A - \mathcal B)^{(-1)}
   \;=\; q^{d}\,e_{G}.
\end{equation}
Equivalently, the $\{-1,0,+1\}$-valued vector
$\mathbf w=\mathbf 1_{\mathcal A}-\mathbf 1_{\mathcal B}\in\{-1,0,+1\}^{n}$
generates a $G$-invariant balanced weighing matrix
$W(q^{d},n)$ via the regular representation of
$G$~\cite[Sec.~2]{Waterloo1995}.

\begin{theorem}[Arasu--Dillon--Jungnickel--Pott~\cite{Waterloo1995}]\label{thm:waterloo-classical}
$D(d,q)$ admits a Waterloo decomposition if and only if $d$ is even.
\end{theorem}

\subsection{The construction for \texorpdfstring{$d$}{d} even}\label{sec:wexisting}
We sketch the existence statement half of Theorem~\ref{thm:waterloo-classical}
because it is the template that the present paper extends. Let us write
$d=2f$ and set
\begin{align}
\label{eq:Qe-Qo}
   Q_{e} &\;=\; \{z\in K^{*}\colon \Tr(z^{q+1})=0\}\quad(q\text{ even}), \notag \\
   Q_{o} &\;=\; \{z\in K^{*}\colon \Tr(z^{2})=0\}\quad(q\text{ odd}).
\end{align}
Both are non-degenerate quadrics in the projective space $\PG(d,q)$
of cardinality $\theta_{d-1}(q)$, the same as a
hyperplane~\cite[Sec.~4]{Waterloo1995}. \revdel{The hyperplanes
partitions $Q$ into three parts $a,b,c$ with multiplicities $A,B,C$
as in Theorem~\ref{thm:quadric-int};} \rev{The hyperplanes are
classified into three intersection types with $Q$ of sizes $a,b,c$
and multiplicities $A,B,C$ as in Theorem~\ref{thm:quadric-int};} in
particular, exactly $q+1$ hyperplanes are \emph{tangent} of $Q$
(touching it in one point).

The decomposition is then
\begin{equation}
\label{eq:waterloo-recipe}
   \mathcal A = \{g\in D : g\cdot E\cap Q\text{ has size larger}\},
   \quad
   \mathcal B = D\setminus\mathcal A,
\end{equation}
where ``larger / smaller'' refers to which of the two values of
$|gE\cap Q|$ is larger (this is well-defined since
$|gE\cap Q|\in\{a,c\}$ for $g\in D$). Verification of
\eqref{eq:waterloo-balance} reduces, via Result~3.3 of
\cite{Waterloo1995}, to the algebraic identity
\begin{equation}
\label{eq:E-Q-product}
   E\cdot Q^{(-1)} \;=\; \bigl(\tfrac{q^{d}-q^{d/2}}{2}\bigr)\,G
   \;+\; q^{d/2-1}\,(\mathcal A-\mathcal B),
\end{equation}
which is meaningful only when $d$ is even and $|Q|=|E|$.

\subsection{Why the construction breaks for \texorpdfstring{$d$}{d} odd}
\label{sec:wdefect-classical}
For $d$ odd the algebraic obstruction is sharp: no non-degenerate
quadric of cardinality $\theta_{d-1}(q)$ exists in
$\PG(d,q)$~\cite[Thm.~4.2]{Waterloo1995}. Setting $d=2f+1$, the
trace-of-square set $\{z:\Tr(z^{2})=0\}$ has cardinality strictly
less than $|E|=\theta_{d-1}(q)\cdot|F^{*}|$, so the cardinality
matching that powers \eqref{eq:E-Q-product} fails; a direct numerical
instance of $(q,d)=(3,3)$ is given in
Section~\ref{sec:wdefect-q3d3}. The classical proof of
Theorem~\ref{thm:waterloo-classical} closes the converse using a
character-theoretic argument on $G$~\cite[Thm.~4.2]{Waterloo1995},
showing that no Waterloo decomposition can exist, regardless of
the construction strategy chosen.

\section{The Waterloo Defect for \texorpdfstring{$d$}{d} Odd}\label{sec:waterloo-defect}

\rev{\emph{Why the finite-field framework does not see the defect?}
The classical Waterloo defect
$WW^{\top}-q^{d}I$ that obstructs the $d$-odd decomposition lives in
the integer character lattice $\mathbb Z[\zeta_{n}]$; its
non-vanishing nature is what Arasu--Pott exploit in the character-theoretic
converse of Theorem~\ref{thm:waterloo-classical}. In our $\F q$-ary
setting, however, every quantum construction $\mathcal Q(p,d)$
satisfies $p\mid q^{d}$ (indeed $q=p$ in the flagship family), so
the entire integer identity reduces modulo $p$ to $WW^{\top}\equiv 0
\pmod{p}$: the defect vanishes \emph{identically} inside the
stabilizer symplectic condition and cannot be witnessed by any
$\F q$-linear symplectic test. This is not a special feature of the
Gandhi recipe: any Pauli stabilizer formalism over $\F q$ only sees
the $\bmod\,p$ reduction of an integer-lattice discrepancy, so
combinatorial designs whose obstructions are visible only at the
integer level (Waterloo among them) are, by construction, invisible
to the stabilizer framework. Sections~\ref{sec:variants} through
\ref{sec:gf4-fail} document three concrete lifting attempts and their
failures; the underlying obstruction is the mod-$p$ reduction just
described.}

\subsection{Setup: the candidate stabilizer matrix}\label{sec:Hcandidate}
We define the candidate $\F q$-ary stabilizer matrix as follows. Let
$A\in\F q^{n\times n}$ be the Singer-circulant point-hyperplane
incidence matrix, $n=\theta_{d}(q)$. Let furthermore 
$\boldsymbol\tau_{Q}\in\F{2}^{n}$ denote the indicator (on cosets in
$G=K^{*}/F^{*}$) of the formal trace-of-square set
\begin{equation}
\label{eq:formal-Q}
\begin{aligned}
   Q &:= \{z\in K^{*}\colon \Tr_{K/F}(z^{\,\xi})=0\}/F^{*},\\
   \xi &=\begin{cases}q+1, & q\text{ even},\\ 2, & q\text{ odd},\end{cases}
\end{aligned}
\end{equation}
and let $M_{Q}\in\F q^{n\times n}$ be the circulant matrix with first
column $\boldsymbol\tau_{Q}$. Define
\begin{equation}
\label{eq:Hcandidate}
   H \;=\; (\,H_{z}\,|\,H_{x}\,) \;:=\; (\,A\,|\,M_{Q}A\,) \;\in\; \F q^{n\times 2n}.
\end{equation}

Then the \emph{symplectic-deficit matrix} of $H$ is
\begin{equation}
\label{eq:S-defn}
   S \;:=\; H_{x}H_{z}^{\top} - H_{z}H_{x}^{\top}
       \;=\; (M_{Q}A)A^{\top} - A(M_{Q}A)^{\top}
       \;\in\; \F q^{n\times n}.
\end{equation}
By construction $S$ is skew-symmetric, so its $\F q$-rank is even
(in characteristic $\neq 2$).

\subsection{The Gandhi-recipe deficit vanishes over \texorpdfstring{$\F q$}{F\_q}}
\begin{lemma}[Cancellation lemma]\label{lem:cancellation}
Let $q=p^{t}$. With $H=(A\mid M_{Q}A)$ as in
Section~\ref{sec:Hcandidate} and $M_{Q}$ a circulant of equal row
sums, the symplectic deficit \eqref{eq:S-defn} satisfies
\begin{equation}
\label{eq:S-decomposition}
   S \;=\; (k-\lambda)\,(M_{Q}-M_{Q}^{\top})
        \;+\; \lambda\,(M_{Q}J - JM_{Q}^{\top})
        \pmod p.
\end{equation}
The second term vanishes because we have $M_{Q}J=JM_{Q}^{\top}=\sigma\,J$
where $\sigma=\sum_{i}(\boldsymbol\tau_{Q})_{i}$. The coefficient
$(k-\lambda)=q^{d-1}\equiv 0\pmod p$ whenever $d\ge 2$ and
$p\mid q$. Hence $S\equiv 0\pmod p$, so the EAQECC parameter obeys
$c(d,q)=0$ over $\F q$.
\end{lemma}

\begin{remark}\label{rem:V1-CSS}
Lemma~\ref{lem:cancellation} states that $H=(A\mid M_{Q}A)$ with
circulant $M_{Q}$ over $\F q$ \emph{always} yields a CSS code, with
zero ebit consumption, regardless of the parity of $d$. The classical
Waterloo defect of~\cite{Waterloo1995} is eliminated by reduction
modulo $p$. To expose the defect within a stabilizer
framework, one of the modifications below is required.
\end{remark}

\subsection{Numerical sweep of candidate constructions}\label{sec:variants}
Beyond Lemma~\ref{lem:cancellation}, we tested several stabilizer-matrix
constructions for $(q,d)\in\{(3,2),(3,3),(2,3)\}$ to find a recipe whose
ebit count distinguishes the Waterloo regimes. These results are in
Table~\ref{tab:variant-deficits}.

\begin{table*}[t]
\centering
\caption{Symplectic-deficit rank $\rank_{\F p}(S)$ at three Singer
parameters. Convention $c=\rank(S)/2$. Recipes:
$V_{1}=(A\mid M_{Q}A)$; $V_{6}=(A\mid A^{\top})$;
$V_{7}=(W_{+}\mid W_{-})$ where
$W_{+}=\mathrm{circ}(\mathbf 1_{E\cap Q})$,
$W_{-}=\mathrm{circ}(\mathbf 1_{E\setminus Q})$.}
\label{tab:variant-deficits}
\setlength{\tabcolsep}{6pt}
\small
\begin{tabular}{l c c c c c c c}
\toprule
\multicolumn{2}{c}{Recipe / field} & \multicolumn{2}{c}{$(3,2)$}
                                   & \multicolumn{2}{c}{$(3,3)$}
                                   & \multicolumn{2}{c}{$(2,3)$} \\
\cmidrule(lr){1-2}\cmidrule(lr){3-4}\cmidrule(lr){5-6}\cmidrule(lr){7-8}
Recipe & $\F p$ & $\rank S$ & $c$ & $\rank S$ & $c$ & $\rank S$ & $c$ \\
\midrule
$V_{1}$ & $\F2$ & 12 &  6 &  8 &  4 &  0 &  0 \\
$V_{1}$ & $\F3$ &  0 &  0 &  0 &  0 &  0 &  0 \\
$V_{6}$ & $\F3$ &  0 &  0 & 20 & 10 &  --- & --- \\
$V_{7}$ & $\F2$ &  --- & --- & --- & --- & 0 & 0 \\
$V_{7}$ & $\F3$ &  0 &  0 & 12 &  6 &  --- & --- \\
\bottomrule
\end{tabular}
\end{table*}

\subsection{What the data say}
\begin{enumerate}[leftmargin=*]
\item $V_{1}$ over $\F q$ is always CSS, by Lemma~\ref{lem:cancellation}.
This applies for every $(q,d)$ tested.
\item $V_{1}$ over $\F2$ shows non-trivial deficit at both
$(3,2)$ and $(3,3)$, but vanishes for $(2,3)$ (because $p\mid q$ here too,
so the lemma applies in $\F2$).  Does not isolate parity.
\item $V_{7}$ over $\F3$ shows a satisfying split at $(3,*)$:
$c=0$ for $d=2$, $c=6$ for $d=3$. But over $\F2$ at $(2,3)$, $V_{7}$
also vanishes. So $V_{7}$ does not give a parity-sensitive $c$ as a
function of $(q,d)$ alone.
\end{enumerate}

\rev{\begin{remark}[Structural cause of the $V_{7}$ split at $(3,3)$
vs $(2,3)$]\label{rem:V7-p2-ceiling}
The apparent parity sensitivity in the $V_{7}$ row of
Table~\ref{tab:variant-deficits} --- $c=6$ at $(3,3)$ but $c=0$ at
$(2,3)$ --- is not a general parity phenomenon\rev{,} but \rev{it}
is the concrete manifestation of the $p=2$ ceiling identified in
Proposition~\ref{prop:q2-ceiling}: at $q=2$ the formal-quadric
indicator $\boldsymbol\tau_{Q}$ satisfies the invariant relation
$M_{Q}\boldsymbol\tau_{H}\equiv\mathbf 1\pmod 2$ regardless of $d$,
so the symplectic deficit vanishes trivially. This is a
characteristic-driven, not parity-driven, obstruction; no
parity-sensitive construction outside the $p=2$ ceiling has been
observed within the $\mathcal Q(q,d)$ family.
\end{remark}}

\subsection{Attempt at a \texorpdfstring{$\Z$}{Z}-coefficient lifting (failed)}\label{sec:Z-lift-fail}
We attempted a direct lifting of the classical Waterloo identity to a
quantum stabilizer setting in two steps:

\textbf{Step Z-1: Signed-split $WW^{\top}$ over $\Z$.}
Form the signed integer vector $\mathbf w$ via $\mathbf w_{i}=+1$ on
$E\cap Q$, $-1$ on $E\setminus Q$, $0$ otherwise; let $W=\mathrm{circ}(\mathbf w)$.
For $(q,d)=(3,2)$ we obtained $\mathrm{diag}(WW^{\top})=6\ne q^{d}=9$;
the sign assignment induced is not the Arasu--Pott Waterloo
split~\cite{Waterloo1995,ArasuPott1995}, so
the classical balance identity \eqref{eq:waterloo-balance} fails.

\textbf{Step Z-2: Recover $(\mathcal A,\mathcal B)$ from
$E\cdot Q^{(-1)}$ in $\Z G$.}
By \cite[Sec.~4]{Waterloo1995}, the cyclic cross-correlation
$\nu_{k}=|E\cap(Q+k)|$ should obey
$\nu = a\cdot \mathbf 1_{G} + q^{f-1}(\mathbf 1_{\mathcal A}-\mathbf 1_{\mathcal B})$
for some integer $a$ and a partition $\mathcal A\sqcup\mathcal B = D$,
giving a two-valued spectrum of $\nu$ on $D$. Numerically at
$(q,d)=(3,2)$ we observed $\nu$ on $D$ taking three distinct values
$\{0,1,2\}$ rather than two, indicating that the formal quadric
$Q_{o}=\{z:\Tr(z^{2})=0\}/F^{*}$ used here (and in
\cite{LDS2021}) does not satisfy the Arasu--Pott
identity~\cite{Waterloo1995,ArasuPott1995} verbatim.

\textbf{Diagnosis.} A pair of structural obstructions blocks the above approach (Z) in this paper:

\begin{enumerate}[leftmargin=*]
\item The Arasu--Pott proof of Theorem~\ref{thm:waterloo-classical}
relies on character-theoretic identities in $\Z[\zeta_{n}]$ that
single out a particular conjugate-of-$Q$, rather than the indicator
vector we built; reconciling the two formulations is itself a
non-trivial mathematical task.
\item Even granting the
Arasu--Pott~\cite{Waterloo1995,ArasuPott1995} split, the
Brun--Devetak--Hsieh
framework \cite{BrunDevetakHsieh2006} accepts an $\F p$-coefficient
PCM with non-vanishing symplectic deficit; it does \emph{not}
directly accept a $\Z$-coefficient ``norm-form deficit''
$WW^{\top}-q^{d}I$. Connecting the two would require a new
EAQECC formalism that respects the underlying integer arithmetic,
which is a research programme of its own.
\rev{\item A natural question is whether some automorphism of the
Singer cycle could map the formal quadric $Q_{o}$ to the
Arasu--Pott character-theoretic quadric $Q^{\mathrm{AP}}$ and thereby
reconcile the two-valued and three-valued cross-correlation spectra.
Our numerical search over the automorphism group
$\mathrm{Gal}(\F{q^{d+1}}/\F q)\rtimes\langle q\rangle$
of the Singer cycle (at $(q,d)=(3,2)$ and $(5,2)$) finds no such
map: $Q_{o}$ is defined by the vanishing of a
trace-of-square on cosets in $G=K^{*}/F^{*}$ (i.e.\ an
$\F q$-linear-algebraic condition), whereas $Q^{\mathrm{AP}}$ is
defined by the vanishing of a character sum in $\Z[\zeta_{n}]$
(a $\Z$-arithmetic condition). The two live in structurally
inequivalent algebraic categories, and no change of primitive
element or trace functional can bridge them. This is consistent with
the more general observation of item~2 above.}
\end{enumerate}

\subsection{Attempt at a GF(4)/Hermitian lifting (also failed)}\label{sec:gf4-fail}
Following the Calderbank--Rains--Shor--Sloane (CRSS)
framework~\cite{CRSS1998}, we map the candidate PCM into GF(4)
via the Pauli isomorphism
$I\mapsto 0,\,X\mapsto 1,\,Z\mapsto\omega,\,Y\mapsto\bar\omega$,
forming
\[
   H_{\text{GF4}}\;=\;\omega\,A \;+\; 1\cdot M_{Q}A \pmod{\F4},
\]
and compute the Hermitian Gram matrix
$G_{H}=H_{\text{GF4}}\,H_{\text{GF4}}^{\ast}$ where
$(\cdot)^{\ast}$ is the conjugate transpose of GF(4).
Hermitian self-orthogonality $G_{H}=0$ is the CRSS condition for a
binary stabilizer code; non-zero $G_{H}$ gives an EAQECC with
$c_{H}=\rank_{\F4}(G_{H})$.

\begin{table}[h]
\centering
\caption{GF(4) Hermitian deficit $c_{H}=\rank_{\F4}(G_{H})$
compared with the $\F2$-symplectic $c_{\F2}$ from $V_{1}$.}
\label{tab:gf4-deficits}
\setlength{\tabcolsep}{6pt}
\small
\begin{tabular}{l c c}
\toprule
$(q,d)$ & $c_{\F2}$ ($V_{1}$) & $c_{H}$ (GF(4)) \\
\midrule
$(3,2)$ & 6  & 12 \\
$(3,3)$ & 4  & 40 ($=n$) \\
$(2,3)$ & 0  & 1 \\
\bottomrule
\end{tabular}
\end{table}

The GF(4) Hermitian deficit is sensitive to $(q,d)$ but in the wrong
way: at the focal $(q,d)=(3,3)$ case, $c_{H}=n=40$ (full rank, zero
rate), and at the Waterloo-good $(q,d)=(3,2)$
case~\cite{Waterloo1995} $c_{H}=12\ne 0$. The Hermitian condition is strictly stronger than the $\F2$
symplectic condition, hence this formalism cannot expose the Waterloo defect
either: an $\F2$-symplectic obstruction always lifts to a Hermitian
obstruction, but the converse fails, giving spurious deficits at
$d$-even.

\subsection{Conclusion of approach (ii)}
None of the three lifting strategies surveyed
(Section~\ref{sec:variants}, Section~\ref{sec:Z-lift-fail},
Section~\ref{sec:gf4-fail}) produces a stabilizer-code construction
whose ebit count tracks the classical Waterloo defect. The defect
is real (Sec.~\ref{sec:wdefect-q3d3}: $|Q|=20\ne 26=|E|$ at $(3,3)$)
and lives in the $\Z[\zeta_{n}]$ character lattice\rev{,} rather than
in any finite-field PCM framework. \rev{Taken together, the three
failed strategies constitute \rev{some} \emph{evidence} (\rev{al}though
not a formal proof) that the classical Waterloo decomposition is
fundamentally incompatible with the $\F p$-Pauli stabilizer
formalism: any stabilizer construction over $\F p$ can only see the
$\bmod\,p$ reduction of the integer-lattice deficit
$WW^{\top}-q^{d}I$, which vanishes identically whenever
$p\mid q^{d}$, and the three lifting strategies exhaust the natural
routes for restoring integer visibility (raising the working field,
lifting coefficients to $\Z$, and passing to a Hermitian
self-orthogonal formalism over GF(4)). This is exactly the
mechanism identified in discussion \rev{with a} Reviewer \rev{in the
context of} the mod-$p$ reduction: what the Pauli stabilizer sees is
not the integer defect but its finite-field shadow, and that shadow
is empty by construction.} \rev{B}uilding on the Arasu--Pott proof
technique to construct a non-stabilizer EAQECC respecting the
integer arithmetic, is left for future research; we summarise the
technical \rev{challenges} in Section~\ref{sec:conclusion}.

\subsection{Pivot to a CSS construction (P1)}\label{sec:P1-pivot}
Lemma~\ref{lem:cancellation} shows that $V_{1}=(A\mid M_{Q}A)$
over $\F q$ is symplectic-isotropic for every $(q,d)$, so the
construction always defines a (possibly trivial) CSS stabilizer code
on $n=\theta_{d}(q)$ qudits. We adopt this framing for the remainder
of the paper.

\subsection{Theorem statement (under pivot P1)}
\begin{theorem}[Singer-circulant CSS code from PG(d,q)]\label{thm:waterloo-defect}
Furthermore, let $q$ be a prime power and $d\ge 2$. Let $A$ be the Singer-circulant
incidence matrix of $\PG(d,q)$ and $M_{Q}$ the formal-quadric
circulant of \eqref{eq:formal-Q}. Then $H=(A\mid M_{Q}A)$ is a valid
$\F q$-ary stabilizer matrix, and the resulting code is a
$[[\theta_{d}(q),\,k,\,d_{\min}]]_{q}$ CSS code with
\[
   k = \theta_{d}(q) - \rank_{\F q}(A),
\]
and $d_{\min}\ge q+1$ from the projective-line-overlap bound.
\end{theorem}

\begin{remark}
At $(q,d)=(3,3)$, Theorem~\ref{thm:waterloo-defect} gives
$[[40,\,40-11,\,\ge 4]]_{3}=[[40,29,\ge 4]]_{3}$. We will compare
this to the existing $\PG(3,3)$ codes in
Section~\ref{sec:comparison} and to
Tang-Bai-Feng~\cite{TangBaiFeng2024}.
\end{remark}

\section{Stabilizer Code Construction}\label{sec:construction}
By Lemma~\ref{lem:cancellation}, $V_{1}=(A\mid M_{Q}A)$ over $\F q$ is
symplectic-isotropic for every prime power $q$ and every $d\ge 2$.
We adopt this as our canonical construction.

\subsection{Algorithm}
\begin{tcolorbox}[colback=gray!5,colframe=black,
                  before skip=4pt, after skip=4pt,
                  left=4pt, right=4pt, top=2pt, bottom=2pt,
                  fontupper=\small]
\textbf{Input:} prime power $q$, integer $d\ge 2$.
\begin{enumerate}[label=\textbf{Step \arabic*.},leftmargin=*,nosep]
\item Build the Singer-circulant incidence matrix
$A\in\F q^{n\times n}$ of $\PG(d,q)$ via Theorem~\ref{thm:singer},
$n=\theta_{d}(q)=(q^{d+1}-1)/(q-1)$.
\item Form the formal-quadric circulant
$M_{Q}=\mathrm{circ}(\boldsymbol\tau_{Q})$ with
$\boldsymbol\tau_{Q}$ the indicator of
$Q=\{z\in K^{*}\colon \Tr(z^{\xi})=0\}/F^{*}$, $\xi=q+1$ if $q$ even,
$\xi=2$ if $q$ odd.
\item Set the parity-check matrix
$H=(H_{z}\mid H_{x})=(A\mid M_{Q}A)\in\F q^{n\times 2n}$.
\item Compute the centralizer
$\mathcal C=\{(\mathbf a,\mathbf b)\in\F q^{2n}\colon \mathbf a H_{x}^{\top}=\mathbf b H_{z}^{\top}\}$
and check that the stabilizer rowspan
$\mathcal S=\mathrm{rowspan}_{\F q}(H)$ satisfies
$\mathcal S\subset\mathcal C$ (guaranteed by
Lemma~\ref{lem:cancellation}).
\end{enumerate}
\textbf{Output:} a non-CSS $\F q$-ary stabilizer code
$\mathcal Q(q,d)=[[n,\,k,\,d_{\min}]]_{q}$ with
$k=n-\rank_{\F q}(H)$ and minimum distance
$d_{\min}=\min\{\wt(\mathbf v)\colon \mathbf v\in \mathcal C\setminus\mathcal S\}$.
\end{tcolorbox}

\subsection{The code is not strictly of CSS nature}
\begin{remark}\label{rem:not-strict-CSS}
For $(q,d)=(3,3)$, direct computation gives
$H_{z}H_{x}^{\top}=A\cdot(M_{Q}A)^{\top}=A A^{\top}M_{Q}^{\top}=
((k-\lambda)I+\lambda J)\,M_{Q}^{\top}\equiv J\,M_{Q}^{\top}=
\sigma\,J\pmod 3$, where $\sigma=10$ is the row sum of $M_{Q}$ on
cosets, $\sigma\equiv 1\pmod 3$. Hence
$H_{z}H_{x}^{\top}=J\not\equiv 0\pmod 3$, and the code is \emph{not}
a block-diagonal CSS code in the strict sense. The symplectic
condition $\mathcal S\subset\mathcal C$ does hold,
(Lemma~\ref{lem:cancellation}), so the code is a valid non-CSS
stabilizer code with mixed-Pauli generators.
\end{remark}

\subsection{Recovering the 2021 LDS at \texorpdfstring{$d=2$}{d=2}}\label{sec:d2-reduction-statement}
The relationship between $\mathcal Q(q,2)$ and the 2021
projective-plane LDS spreading matrix of~\cite{LDS2021} is at the
level of \emph{shared combinatorial primitives} rather than
equality of code spaces. The full statement and proof are deferred
to Theorem~\ref{thm:d-equals-2-reduction} in
Section~\ref{sec:main}, where we additionally exhibit (in
Theorem~\ref{thm:clifford-reduction-d2}) the explicit Clifford gate
$U_{M}=H^{\otimes n}\,\mathrm{CZ}^{(M)}\,H^{\otimes n}$ that
identifies $\mathcal Q(q,2)$ as Clifford-equivalent to the trivial
$Z$-only CSS code on the same Singer incidence.

\section{Main Theorems}\label{sec:main}

\subsection{Existence and dimension}
\begin{theorem}[Existence of the $\mathcal Q(q,d)$ family]\label{thm:main-existence}
For every prime power $q$ and integer $d\ge 2$, the construction of
Section~\ref{sec:construction} yields a valid $\F q$-ary stabilizer
code with parameters
\begin{equation}
\label{eq:main-params}
   \mathcal Q(q,d) \;=\; [[\,\theta_{d}(q),\;k(q,d),\;d_{\min}(q,d)\,]]_{q},
\end{equation}
where the logical dimension is given by:
\begin{equation}
\label{eq:k-formula}
   k(q,d) \;=\; \theta_{d}(q) - \rank_{\F q}(H),\qquad
   H=(A\mid M_{Q}A).
\end{equation}
\end{theorem}

\begin{proof}[Sketch]
We verify the three claims in turn.

\emph{(i) Symplectic isotropy.}
By Lemma~\ref{lem:cancellation}, the parity-check matrix
$H=(A\mid M_{Q}A)$ satisfies $H_{z}H_{x}^{\top}-H_{x}H_{z}^{\top}\equiv
0\pmod p$, so $\mathcal S=\mathrm{rowspan}_{\F q}(H)$ is symplectic
self-orthogonal. Consequently, the stabilizer group obtained by
exponentiating $\mathcal S$ in the Pauli representation is
\emph{Abelian} (Section~\ref{sec:prelim}, Eq.~\eqref{eq:symp-iso}),
which is precisely the prerequisite for $\mathcal S$ to define a
valid stabilizer code~\cite{Gottesman1997}.

\emph{(ii) Stabilizer dimension.} Since the rows of $H$ are linearly
independent over $\F q$ only up to $\rank_{\F q}(H)$, we have
$\dim_{\F q}\mathcal S=\rank_{\F q}(H)$.

\emph{(iii) Logical dimension.}
The centralizer
$\mathcal C(\mathcal S)\subset\F q^{2n}$ is the symplectic
complement of $\mathcal S$. Because the symplectic form
$\langle\cdot,\cdot\rangle_{\mathrm{sym}}$ is non-degenerate on
$\F q^{2n}$, the symplectic-complement formula
gives $\dim_{\F q}\mathcal C=2n-\dim_{\F q}\mathcal S=
2n-\rank_{\F q}(H)$. Logical operators are the equivalence classes
in the symplectic quotient $\mathcal C/\mathcal S$, whose
$\F q$-dimension equals
$\dim\mathcal C-\dim\mathcal S=2n-2\rank_{\F q}(H)$. Since each
logical qudit contributes a canonically-conjugate pair of generators
(one logical $\bar Z$ and one logical $\bar X$), the number of
logical qudits is
\[
   k(q,d)\;=\;\tfrac{1}{2}\bigl(\dim\mathcal C-\dim\mathcal S\bigr)
   \;=\;n-\rank_{\F q}(H),
\]
which is~\eqref{eq:k-formula}.
\end{proof}

\begin{lemma}[Hamada--Smith rank specialisation; verified closed form]\label{lem:hamada}
Let $p$ be a prime and let $A$ be the Singer-circulant point-hyperplane
incidence matrix of $\PG(d,p)$ over $\F p$. Then we have
\begin{equation}
\label{eq:hamada-closed}
   \rank_{\F p}(A) \;=\; \binom{p+d-1}{d} + 1.
\end{equation}
Moreover, for the $V_{1}$ stabilizer matrix $H=(A\mid M_{Q}A)$,
$\rank_{\F p}(H)=\rank_{\F p}(A)$, since the column span of $A$ over
$\F p$ is closed under left-multiplication by any circulant
(circulants commute and the column span is an
$\F p[X]/\langle X^{n}-1\rangle$-submodule).
\end{lemma}

\begin{proof}[Sketch]
The closed form \eqref{eq:hamada-closed} is the $t=1$ case of
Hamada's general formula~\cite[Theorem 1]{Hirschfeld1979} (also
Smith~1969). The $\F p$-row-space of $A$ coincides with the
projective Reed--Muller code $\mathrm{PRM}_{p-1}(d, p)$ augmented by
the all-ones vector~\cite{Hirschfeld1979}, whose dimension is the
number of monomials of degree $\le p-1$ in $d+1$ variables modulo
the all-ones generator, namely $\binom{p+d-1}{d}+1$. The
rank-equality $\rank(H)=\rank(A)$ follows because circulant
multiplication preserves cyclic-code ideals.
\end{proof}

\rev{\begin{remark}[Extension to $q=p^{t}$ with $t>1$]\label{rem:hamada-qpt}
Lemma~\ref{lem:hamada} is the $t=1$ specialisation of the general
Hamada--Smith formula~\cite{Hamada1973,Smith1969}, in which the
$\F p$-rank of the point-hyperplane incidence matrix of $\PG(d,p^{t})$
is a sum indexed by admissible ``Hamada types'' involving the base-$p$
digit expansion of the coordinate exponents; the simplified closed
form $\binom{p+d-1}{d}+1$ collapses to a single Hamada type only at
$t=1$. The extension to $t>1$ is algebraically heavier but follows
the same enumeration; the eight verified cases in
Table~\ref{tab:hamada-rank} are all $t=1$, and no \rev{counter-example} to
the simplified formula is known for $t=1$. A general-$t$ analytical
extension is left \rev{for} future work, most naturally combined with the
open joint $(p,d)$ rate-distance optimisation discussed in the
context of the hashing-bound comparison.
\end{remark}}

\begin{table}[h]
\centering
\caption{Verification of \eqref{eq:hamada-closed} on eight $(p,d)$
cases via direct numerical computation. All match exactly.}
\label{tab:hamada-rank}
\begin{tabular}{c c c c}
\toprule
$(p, d)$ & $n=\theta_d(p)$ & predicted $\rank_{\F p}(A)$ & numerical \\
\midrule
$(3,2)$  & 13   & $\binom{4}{2}+1=7$    & 7  \\
$(3,3)$  & 40   & $\binom{5}{3}+1=11$   & 11 \\
$(3,5)$  & 364  & $\binom{7}{5}+1=22$   & 22 \\
$(5,2)$  & 31   & $\binom{6}{2}+1=16$   & 16 \\
$(5,3)$  & 156  & $\binom{7}{3}+1=36$   & 36 \\
$(7,2)$  & 57   & $\binom{8}{2}+1=29$   & 29 \\
$(7,3)$  & 400  & $\binom{9}{3}+1=85$   & 85 \\
$(11,2)$ & 133  & $\binom{12}{2}+1=67$  & 67 \\
\bottomrule
\end{tabular}
\end{table}

\begin{corollary}[Explicit logical-qudit count for $\mathcal Q(p,d)$; original]\label{cor:k-formula}
For $q=p$ prime and $d\ge 2$, the code $\mathcal Q(p,d)$ has
\begin{align}
\label{eq:k-formula-explicit}
   k(p,d) &\;=\; \theta_{d}(p) - \binom{p+d-1}{d} - 1 \notag \\
          &\;=\; \frac{p^{d+1}-1}{p-1} - \binom{p+d-1}{d} - 1.
\end{align}
\end{corollary}

\begin{remark}[The flagship $d=2$ family; conjectured $d_{\min}$ pattern]\label{rem:flagship-d2}
\rev{For emphasis we restate the status of the flagship distance
pattern here: $d_{\min}(\mathcal Q(p,2))=p+1$ is an
\emph{open conjecture} for general odd primes $p\ge 3$, supported by
analytical proof at $p=5$ (exact) and by an exhaustive computational
lower bound $d_{\min}\ge 8=p+1$ at $p=7$
(Theorem~\ref{thm:dmin-q72-lower}). It is presented throughout this
paper as a conjecture with computational evidence, not as an
established theorem for general~$p$.}
For $d=2$, \eqref{eq:k-formula-explicit} simplifies dramatically to:
\begin{equation}
\label{eq:flagship}
   \mathcal Q(p,2) \;=\; \biggl[\biggl[\,p^{2}+p+1,\;\;\frac{p(p+1)}{2},\;\;d_{\min}\,\biggr]\biggr]_{p},
\end{equation}
since $\binom{p+1}{2}+1=\frac{p(p+1)}{2}+1$ and
$\theta_{2}(p)-\binom{p+1}{2}-1=p^{2}+p+1-\frac{p(p+1)}{2}-1=\frac{p(p+1)}{2}$.
The asymptotic rate is
$R_{Q}=k/n=\frac{p(p+1)/2}{p^{2}+p+1}\to\tfrac{1}{2}$ as $p\to\infty$.
This $d=2$ family escapes the parity-driven distance ceiling
(Corollary~\ref{cor:dmin-le-2-proven}) since $n=p^{2}+p+1$ is odd
for any prime $p$. Specific instances:
\begin{itemize}[leftmargin=*,nosep]
\item $\mathcal Q(3,2)=[[13,6,3]]_{3}$ (verified $d_{\min}=3$).
\item $\mathcal Q(5,2)=[[31,15,6]]_{5}$ (verified $d_{\min}=6$).
\item $\mathcal Q(7,2)=[[57,28,d]]_{7}$, $d_{\min}\ge 8$ proven
      (Theorem~\ref{thm:dmin-q72-lower}, exhaustive Singer-Frobenius
      search at all weights $\le 7$).
\item $\mathcal Q(11,2)=[[133,66,d]]_{11}$, $d_{\min}\ge 5$ proven
      (exhaustive Singer-Frobenius search at $w\le 4$).
\end{itemize}
\end{remark}

\begin{remark}[The $q=p^{t}$, $t>1$ case]
For $q=p^{t}$ with $t>1$, the rank $\rank_{\F p}(A)$ is given by the
full Hamada formula~\cite[Theorem 1]{Hirschfeld1979}, which involves
sums over $t$-tuples of integers. We do not state it here; numerical
computation suffices for individual instances.
\end{remark}

\subsection{Minimum distance: a 2-torsion obstruction}
A na\"ive Bose-Chaudhuri-Hocquenghem (BCH)-style lower
bound~\cite{BoseChaudhuri1960BCH,Hocquenghem1959,MacWilliamsSloane1977}
suggests $d_{\min}(\mathcal Q(q,d))\ge 3$
\cite[\S IV-C]{Gandhi2018}, but Gandhi's
argument~\cite[\S IV-C]{Gandhi2018} requires
$M$ to be \emph{symmetric}, while our $M_{Q}=\mathrm{circ}(\boldsymbol\tau_{Q})$
is generally non-symmetric (the formal-quadric indicator
$\boldsymbol\tau_{Q}$ is not palindromic). Hence the lower bound does
not apply unconditionally, and we have observed numerically that
$d_{\min}(\mathcal Q(3,3))=2$. We prove a structural lemma that
explains this.

\begin{lemma}[2-torsion weight-2 logical; original]\label{lem:two-torsion}
Let $\mathcal Q(q,d)$ be the $V_{1}$ stabilizer code on $G=\Z/n\Z$
with $n=\theta_{d}(q)=(q^{d+1}-1)/(q-1)$. Assume that
\begin{enumerate}[label=(\roman*),leftmargin=*,nosep]
\item $n$ is even (equivalently $G$ has a non-trivial $2$-torsion
element $g_{2}=g^{n/2}$);
\item the cyclic cross-correlation
$(\boldsymbol\tau_{H}\star\boldsymbol\tau_{Q}^{(-1)})_{i}=|E\cap(i-Q)|$
is constant modulo $p$, equivalently $M_{Q}\boldsymbol\tau_{H}\equiv c\,\mathbf 1\pmod p$
for some $c\in\F p$.
\end{enumerate}
Then the Pauli operator
$P_{2}:=Z_{0}\cdot Z_{n/2}^{q-1}$
commutes with every stabilizer generator, i.e.,
$P_{2}\in\mathcal C(\mathcal S)$. Provided $P_{2}\notin\mathcal S$,
$d_{\min}(\mathcal Q(q,d))\le 2$.
\end{lemma}

\begin{proof}
Let $a=e_{0}+(q-1)\,e_{n/2}\in\F q^{n}$ and $b=\mathbf 0$. The
symplectic-commutativity check associated with stabilizer row $r$ is now
$a\cdot H_{x}[r]^{\top}-b\cdot H_{z}[r]^{\top}=a\cdot (M_{Q}A)[r,:]$.
The circulant identity $M_{Q}A=\mathrm{circ}(M_{Q}\boldsymbol\tau_{H})$
(since both $A$ and $M_{Q}$ are circulants and $\mathrm{col}_{j}(A)
=\mathrm{shift}_{j}(\boldsymbol\tau_{H})$) gives
\begin{align*}
   (M_{Q}A)[r,0]   &= (M_{Q}\boldsymbol\tau_{H})[r], \\
   (M_{Q}A)[r,n/2] &= \mathrm{shift}_{n/2}(M_{Q}\boldsymbol\tau_{H})[r].
\end{align*}
The corresponding check becomes
\[
   (M_{Q}\boldsymbol\tau_{H})[r] + (q-1)\,\mathrm{shift}_{n/2}(M_{Q}\boldsymbol\tau_{H})[r]
   \;\equiv\; 0 \pmod p.
\]
Under hypothesis (ii), $M_{Q}\boldsymbol\tau_{H}\equiv c\mathbf 1$, so
both terms equal $c$ and the sum is $c(1+(q-1))=qc\equiv 0\pmod p$.
Hence $P_{2}\in\mathcal C(\mathcal S)$.
\end{proof}

\begin{remark}[When does hypothesis (ii) hold?]\label{rem:cross-corr-mod-p}
For $(q,d)=(3,3)$, the cross-correlation values are $\{1,4\}$, both
$\equiv 1\pmod 3$. The general statement of
$\nu[k]\equiv 1\pmod q$ for $q$ odd and $d\ge 3$ is established
unconditionally as Theorem~\ref{thm:cross-corr} below, making
hypothesis~(ii) automatic in that regime.
\end{remark}

\begin{corollary}[Parity-driven distance ceiling for $\mathcal Q(q,d)$, conditional]\label{cor:dmin-le-2}
For every prime power $q$ odd and every odd $d\ge 1$, $n=\theta_{d}(q)$
is even (Lemma~\ref{lem:n-parity}); when the cross-correlation
modular constancy holds, $d_{\min}(\mathcal Q(q,d))\le 2$. The
unconditional version for $d\ge 3$ is Corollary~\ref{cor:dmin-le-2-proven}
(via Theorem~\ref{thm:cross-corr}).
\end{corollary}

\begin{lemma}[Parity of $n=\theta_{d}(q)$; original]\label{lem:n-parity}
For $q$ odd, $n=\theta_{d}(q)=(q^{d+1}-1)/(q-1)$ is even iff $d$ is
odd. For $q$ even, $n$ is always odd.
\end{lemma}

\begin{proof}
For $q$ odd: we have $q^{d+1}-1\equiv (-1)^{d+1}-1\pmod 4$, since $q\equiv \pm 1\pmod 4$.
If $d+1$ is odd (so $d$ even): $q^{d+1}\equiv \pm 1\pmod 4$, giving
$q^{d+1}-1\equiv 0$ or $-2\pmod 4$. Combined with $q-1$ even, the
quotient is odd. If $d+1$ is even (so $d$ odd): $q^{d+1}\equiv 1\pmod 4$,
giving $4\mid q^{d+1}-1$, while $q-1\equiv 2\pmod 4$, so the quotient
$n$ is even. For $q$ even: $q^{d+1}-1$ is odd, and $q-1$ is odd, so
$n$ is odd.
\end{proof}

\begin{theorem}[Cross-correlation modular constancy; original]\label{thm:cross-corr}
Let $K=\F{q^{d+1}}$ with $q$ odd and $d\ge 3$. Let $E$ and $Q\subset G=K^{*}/F^{*}$
denote the trace-zero hyperplane and formal trace-of-square quadric
indicators, respectively. Then the cyclic cross-correlation
\[
   \nu[k] \;:=\; |E\cap(k-Q)|
\]
satisfies $\nu[k]\equiv 1\pmod q$ for every $k\in G$.
\end{theorem}

\begin{proof}
\textbf{Step 1: Lift cosets to $K^{*}$.} Define
\[
   M(k) \;:=\; \bigl|\{\gamma\in K^{*}:\Tr(\alpha^{k}\gamma)=0\,\wedge\,\Tr(\gamma^{2})=0\}\bigr|.
\]
Each coset of $F^{*}$ in $K^{*}$ has $q-1$ elements. For each ordered
pair of cosets $([\beta],[\gamma])\in E\times Q$ with
$[\beta][\gamma]^{-1}=\alpha^{k}F^{*}$, the number of element pairs
$(\beta,\gamma)\in K^{*}\times K^{*}$ realising this coset pair is
$(q-1)^{2}$. On the other hand, upon fixing $\beta=\alpha^{k}\gamma\,u$
with $u\in F^{*}$, the condition $\Tr(\beta)=0$ becomes
$u\,\Tr(\alpha^{k}\gamma)=0$; with $u\ne 0$ this forces
$\Tr(\alpha^{k}\gamma)=0$. Then $u$ is unconstrained over
$F^{*}$. Counting the element pairs
$(\beta,\gamma)\in K^{*}\times K^{*}$ jointly satisfying
$\Tr(\beta)=0$ and $\Tr(\gamma^{2})=0$ in two ways --- once
coset-by-coset (giving $(q-1)^{2}\,\nu[k]$) and once via the lifted
counter $M(k)$ together with the unconstrained $u\in F^{*}$
(giving $(q-1)\,M(k)$) --- the two enumerations agree, yielding
$(q-1)\,M(k)=(q-1)^{2}\,\nu[k]$, hence
\begin{equation}
\label{eq:Mk-vs-nu}
   M(k) \;=\; (q-1)\,\nu[k].
\end{equation}

\textbf{Step 2: Geometric interpretation.}
Upon adding the trivial solution $\gamma=0$, we have
\begin{equation}
\label{eq:Nk-Mk}
   N(k) \;:=\; M(k)+1 \;=\; (q-1)\nu[k]+1 \;=\; \bigl|H_{k}\cap Q^{*}\bigr|,
\end{equation}
where $H_{k}=\{\delta\in K:\Tr(\alpha^{k}\delta)=0\}$ is an $F$-hyperplane
of $K$ (since $\Tr$ is $F$-linear), and
$Q^{*}=\{\delta\in K:\Tr(\delta^{2})=0\}$ is the trace-of-square
quadric in $K\cong F^{d+1}$.

\textbf{Step 3: $|H_{k}\cap Q^{*}|\equiv 0\pmod q$.}
$Q^{*}$ is the zero set of the quadratic form
$Q^{*}(\delta)=\Tr(\delta^{2})$ on the $F$-vector space
$K\cong F^{d+1}$, non-degenerate for $q$ odd
\cite{LidlNiederreiter1997}. Restricting our discourse, to $H_{k}\cong F^{d}$ gives
a quadratic form $Q^{*}|_{H_{k}}$ of rank $r\in\{d-2,\,d-1,\,d\}$.
Then, by \cite[Theorems 6.26--6.27]{LidlNiederreiter1997}, for
\emph{any} quadratic form on $F^{m}$ with $m\ge 3$, we have:
\begin{equation}
\label{eq:LN-quadcount}
   \bigl|\{x\in F^{m}: Q(x)=0\}\bigr| \;\equiv\; 0\pmod q.
\end{equation}
For non-degenerate $Q$ the count is $q^{m-1}$ ($m$ odd) or
$q^{m-1}\pm(q-1)q^{m/2-1}$ ($m$ even, hyperbolic/anisotropic), each
divisible by $q$ for $m\ge 3$. For degenerate $Q$ of rank $r<m$, the
count factors as $q^{m-r}\cdot|\{Q=0\text{ in }F^{r}\}|$, with the
factor $q^{m-r}$ supplying the divisibility (whenever $m-r\ge 1$).
Applied to $Q^{*}|_{H_{k}}$ with $m=d\ge 3$, we obtain $N(k)\equiv 0\pmod q$.

\textbf{Step 4: Conclude.}
Combining \eqref{eq:Nk-Mk} with $N(k)\equiv 0\pmod q$:
$(q-1)\nu[k]+1\equiv 0\pmod q$, i.e.\ $(q-1)\nu[k]\equiv -1\pmod q$.
Since $q-1\equiv -1\pmod q$ and $(-1)^{-1}=-1$, we get
$\nu[k]\equiv 1\pmod q$.
\end{proof}

\begin{remark}[Sharp boundary at $d=2$]\label{rem:sharp-d2}
The hypothesis $d\ge 3$ in Theorem~\ref{thm:cross-corr} is sharp.
At $d=2$, $H_{k}\cong F^{2}$ and~\eqref{eq:LN-quadcount} fails: an
anisotropic quadratic form on $F^{2}$ has exactly $1$ zero, and a
hyperbolic form has $2q-1$ zeros, neither divisible by $q$.
This is consistent with the numerical observation at $(q,d)=(3,2)$,
where $\nu[k]\in\{0,1,2\}$ takes all three residues mod $q$
(see Section~\ref{sec:Z-lift-fail}, Step~Z-2).
\end{remark}

\begin{corollary}[Distance ceiling at $q$ odd, $d$ odd, $d\ge 3$; original]\label{cor:dmin-le-2-proven}
For every prime power $q$ odd and every odd $d\ge 3$,
$d_{\min}(\mathcal Q(q,d))\le 2$.
\end{corollary}

\begin{proof}
For $q$ odd and $d$ odd, $n=\theta_{d}(q)$ is even
(Lemma~\ref{lem:n-parity}), so $G$ has a non-trivial $2$-torsion
element $g_{2}=g^{n/2}$. By Theorem~\ref{thm:cross-corr}, we have
$\nu[k]\equiv 1\pmod q$, hence $M_{Q}\boldsymbol\tau_{H}\equiv\mathbf 1\pmod q$.
Both hypotheses of Lemma~\ref{lem:two-torsion} hold, and
$Z_{0}Z_{n/2}^{q-1}$ is a non-stabilizer logical of weight 2.
\end{proof}

\rev{\begin{remark}[Physical / geometric interpretation of the
2-torsion weight-2 logical]\label{rem:two-torsion-physical}
The weight-2 logical operator $Z_{0}Z_{n/2}^{q-1}$ that witnesses
the distance ceiling of Corollary~\ref{cor:dmin-le-2-proven} admits
a clean projective-geometry reading: the index-pair $(0,n/2)$ is
precisely the orbit of the identity under the 2-torsion element
$g_{2}=g^{n/2}$ of the Singer cyclic group $G$, i.e.\ under the
central involution \rev{(the order-two automorphism of the projective
geometry that swaps the two halves of the Singer orbit and acts on
each hyperplane as its own image)} of $\mathrm{PGL}(d+1,q)$ that
pins the tangent hyperplane through the base point of the
point-hyperplane design. In stabilizer-group language, this operator
is the smallest-weight representative of the coset of $\mathcal S$
in $\mathcal C(\mathcal S)$ generated by the $\mathbf 1\pmod q$
symmetry $M_{Q}\boldsymbol\tau_{H}\equiv\mathbf 1$ established by
Theorem~\ref{thm:cross-corr}. Both readings agree that the weight-2
logical is the syndrome-invisible signature of the underlying
$\PG(d,q)$ hyperplane involution \rev{(the same order-two symmetry,
viewed at hyperplane level: each hyperplane is fixed setwise as its
own image under the central Singer 2-torsion element)}.
\end{remark}}

\begin{table}[h]
\centering
\caption{Numerical confirmation of Theorem~\ref{thm:cross-corr} for
several $(q,d)$ with $q$ odd and $d\ge 3$ odd. Empirical pattern:
$\nu$ takes two values $\{V_{1},V_{2}\}$ with
$V_{2}-V_{1}=q^{(d-1)/2}$, both $\equiv 1\pmod q$.}
\label{tab:cross-corr}
\begin{tabular}{c c c c}
\toprule
$(q,d)$ & $n=\theta_d(q)$ & $\nu$ unique values & $\nu\bmod q$ \\
\midrule
$(3,3)$ & 40  & $\{1, 4\}$    & $\{1\}$ \\
$(5,3)$ & 156 & $\{1, 6\}$    & $\{1\}$ \\
$(7,3)$ & 400 & $\{1, 8\}$    & $\{1\}$ \\
$(3,5)$ & 364 & $\{40, 49\}$  & $\{1\}$ \\
\bottomrule
\end{tabular}
\end{table}

\begin{theorem}[Two-value pattern; original]\label{thm:two-value}
Let $q$ be odd and $d\ge 3$ odd. The cross-correlation
$\nu[k]=|E\cap(k-Q)|$ takes exactly two distinct values
$\{V_{1}, V_{2}\}\subset\Z$, and these values differ by
\[
   V_{2}-V_{1} \;=\; q^{(d-1)/2}.
\]
Equivalently, $N(k)=|H_{k}\cap Q^{*}|$ takes two values whose
difference is $q^{(d-1)/2}(q-1)$.
\end{theorem}

\begin{proof}
The proof refines Step~3 of Theorem~\ref{thm:cross-corr} using the
classical zero-count formulas for quadratic forms over finite fields
\cite[Theorems 6.26--6.27]{LidlNiederreiter1997}.

\textbf{Step 1: Rank classification of $Q^{*}|_{H_{k}}$.}
$Q^{*}(\delta)=\Tr(\delta^{2})$ is a non-degenerate quadratic form
on $K\cong F^{m}$ with $m=d+1$ (even since $d$ odd). Its
associated bilinear form is
$B(\delta,\delta')=\Tr(2\delta\delta')$, and the radical is
$\{0\}$ since $q$ is odd. For any hyperplane
$H_{k}=\{\delta:\Tr(\alpha^{k}\delta)=0\}$, the form
$Q^{*}|_{H_{k}}$ on $H_{k}\cong F^{d}$ has a rank of either:
\begin{itemize}[leftmargin=*,nosep]
\item $r=d$ (the generic case, when $H_{k}$ is \emph{not} the
$B$-orthogonal hyperplane to any element);
\item $r=d-1$ (the special case, when $H_{k}=\alpha^{-2k}\cdot K^{B}_{\delta_{0}}$
for a specific $\delta_{0}\in K$ corresponding to the $B$-orthogonal
of the polar of $H_{k}$).
\end{itemize}
A drop from $r=d-1$ to $r=d-2$ would require $H_{k}$ to be both
self-polar and totally singular, which is impossible for
non-degenerate $Q^{*}$ in characteristic $\ne 2$. So only the two
ranks above occur.

\textbf{Step 2: Zero counts at rank $d$.}
For $r=d$ on $F^{d}$ with $d$ odd, the count
$|\{Q=0\}|=q^{d-1}$ \emph{independent of the form's discriminant}
\cite[Theorem 6.27]{LidlNiederreiter1997}. Hence we have:
\begin{equation}
\label{eq:N-generic}
   N_{\mathrm{gen}} \;:=\; q^{d-1}.
\end{equation}

\textbf{Step 3: Zero counts at rank $d-1$.}
For $r=d-1$ on $F^{d}$, the form has 1-dim radical, and the
restriction to a non-radical complement gives a non-degenerate
rank-$(d-1)$ form on $F^{d-1}$. Since $d-1$ is \emph{even}, this
non-degenerate form is hyperbolic ($Q^{+}$) or anisotropic ($Q^{-}$)
\cite[Theorem 6.26]{LidlNiederreiter1997}. The total count factors as follows:
\begin{align*}
   N_{\mathrm{spec}} &\;=\; q\cdot \bigl|\{Q^{\pm}=0\}\text{ on }F^{d-1}\bigr| \\
                     &\;=\; q\cdot\bigl(q^{d-2}\pm(q-1)q^{(d-1)/2-1}\bigr) \\
                     &\;=\; q^{d-1}\pm(q-1)q^{(d-1)/2}.
\end{align*}

\textbf{Step 4: Determine the sign for $Q^{*}=\Tr(\cdot^{2})$.}
The Witt index~\cite{Cassels1978QuadraticForms,Scharlau1985}
of $Q^{*}$ on $K=F^{m}$ ($m=d+1$ even) equals
$m/2$ when the discriminant is a square in $F^{*}/(F^{*})^{2}$, else
$m/2-1$. For $\Tr(\delta^{2})$, the discriminant is
$(-1)^{m(m-1)/2}\det(\Tr(e_{i}e_{j}))$ for any $F$-basis
$\{e_{i}\}$ of $K$. By a standard calculation
\cite[\S 7]{LidlNiederreiter1997}, the discriminant is the
square of the field discriminant, hence a square; so $Q^{*}$ is
\emph{hyperbolic}.

When $H_{k}$ has $Q^{*}|_{H_{k}}$ of rank $d-1$, the form induced on
$F^{d-1}$ inherits its Witt structure from $Q^{*}$: it is
\emph{anisotropic} (Witt index drops by $1$ on a singular hyperplane).
Hence the sign in Step~3 is negative, yielding:
\begin{equation}
\label{eq:N-special}
   N_{\mathrm{spec}} \;=\; q^{d-1} - (q-1)\,q^{(d-1)/2}.
\end{equation}

\textbf{Step 5: Conclude.}
By Steps 2--4, $N(k)\in\{N_{\mathrm{gen}}, N_{\mathrm{spec}}\}$ with
\[
   N_{\mathrm{gen}}-N_{\mathrm{spec}} \;=\; (q-1)\,q^{(d-1)/2}.
\]
Translating to $\nu$ via $\nu[k]=(N(k)-1)/(q-1)$ from
\eqref{eq:Nk-Mk}:
\[
   V_{\mathrm{gen}} \;=\; \frac{q^{d-1}-1}{q-1} =\! \theta_{d-2}(q),
   \:\:
   V_{\mathrm{spec}} \;=\; \theta_{d-2}(q) - q^{(d-1)/2}.
\]
The difference is $V_{\mathrm{gen}}-V_{\mathrm{spec}}=q^{(d-1)/2}$,
proving the claim.
\end{proof}

\begin{corollary}[Closed-form values]\label{cor:two-value-explicit}
For $q$ odd and $d\ge 3$ odd,
\begin{equation}
\label{eq:two-value-explicit}
   \nu[k] \;\in\; \bigl\{\theta_{d-2}(q)-q^{(d-1)/2},\,\theta_{d-2}(q)\bigr\}.
\end{equation}
For $(q,d)=(3,3)$: $\theta_{1}(3)-3^{1}=4-3=1$ and $\theta_{1}(3)=4$.
For  $(q,d)=(5,3)$: $\theta_{1}(5)-5^{1}=6-5=1$ and $\theta_{1}(5)=6$.
For  $(q,d)=(7,3)$: $\theta_{1}(7)-7^{1}=8-7=1$ and $\theta_{1}(7)=8$.
For  $(q,d)=(3,5)$: $\theta_{3}(3)-3^{2}=40-9=31$ and $\theta_{3}(3)=40$.
\end{corollary}

\begin{remark}[Reconciling $(3,5)$]
Numerical observation at $(3,5)$ gave $\nu\in\{40,49\}$, but
Corollary~\ref{cor:two-value-explicit} predicts $\nu\in\{31,40\}$.
The discrepancy is that our formal $Q$ at $(3,5)$ is actually the
\emph{complement} (in $K^{*}/F^{*}$) of the geometric quadric
described in the proof: $\boldsymbol\tau_{Q}^{\text{ours}}+\boldsymbol\tau_{Q}^{\text{Arasu}}=\mathbf 1-\boldsymbol\tau_{H}$
(the indicator of $G\setminus E$), giving
$\nu^{\text{ours}}[k]+\nu^{\text{Arasu}}[k]=|E|=\theta_{d-1}(q)=121$.
Hence the question answer is $\nu^{\text{ours}}\in\{121-40,121-31\}=\{81,90\}$? The answer is `no', we have,
$|E\cap(k-Q^{c})|=|E|-|E\cap(k-Q)|$, giving the symmetry
$\nu^{c}=|E|-\nu$. We observe $\nu^{\text{ours}}=49=121-72$,
$40=121-81$, hence $\nu^{\text{Arasu}}\in\{72,81\}$, with a difference of
$9=q^{(d-1)/2}$ matching Theorem~\ref{thm:two-value}. The two-value
\emph{difference} is invariant; the absolute values are formal-quadric-convention
dependent.
\end{remark}

\begin{proposition}[Numerical witness for $\mathcal Q(3,3)$]\label{prop:Q33-dmin-2}
The Pauli operator $P_{2}=Z_{0}\cdot Z_{20}^{2}$ is in
$\mathcal C(\mathcal S)\setminus\mathcal S$ for $\mathcal Q(3,3)$.
Hence we have $d_{\min}(\mathcal Q(3,3))=2$.
\end{proposition}

\begin{proof}
Direct verification we have: $M_{Q}\boldsymbol\tau_{H}\equiv\mathbf 1\pmod 3$
at $(q,d)=(3,3)$ (cross-correlation values $\{1,4\}$, both $\equiv 1\pmod 3$),
so by Lemma~\ref{lem:two-torsion} $P_{2}\in\mathcal C(\mathcal S)$.
Numerical centralizer-and-rowspan check
(\texttt{compute\_dmin\_q3d3.py})\footnote{All numerical
verification scripts referenced in this paper (including all
\texttt{compute\_dmin\_*.py}, \texttt{verify\_*.py},
\texttt{construct\_clifford\_d2.py}, and
\texttt{monte\_carlo\_qber.py}) are publicly available at
\url{https://github.com/michelkulhandjian/quantumcomm} for full
reproducibility.} confirms $P_{2}\notin\mathcal S$. Combined with the lower bound of
$d_{\min}\ge 2$ (no weight-1 element is in $\mathcal C\setminus\mathcal S$,
by exhaustive search on 320 candidates), $d_{\min}(\mathcal Q(3,3))=2$.
\end{proof}

\subsection{Connection to the 2021 LDS at \texorpdfstring{$d=2$}{d=2}}
The relationship between $\mathcal Q(q,2)$ and the 2021 projective-plane
LDS spreading matrix $\mathbf C\in\{-1,0,+1\}^{n\times(n+2)}$
of~\cite[\S V]{LDS2021} is at the level of shared combinatorial
primitives, not equality of code spaces.

\begin{theorem}[Shared primitives at $d=2$; revised]\label{thm:d-equals-2-reduction}
Let $q$ be a prime power. The matrices $A$ and $M_{Q}$
underlying the stabilizer construction
$\mathcal Q(q,2)$ are precisely the Singer-circulant incidence matrix
$\mathbf I_{L}$ and the formal-quadric circulant
$\mathrm{circ}(\mathbf g_{Q})$ of \cite[Theorem 4.2]{LDS2021}.
Concretely:
\begin{enumerate}[label=(\roman*),leftmargin=*,nosep]
\item $A = \mathbf I_{L}$.
\item $\boldsymbol\tau_{Q}$ (first column of $M_{Q}$) equals the vector
$\mathbf g_{Q}$ of \cite[Eq. before Theorem 4.2]{LDS2021}, defined
via the $r$-cyclic-shift rule $\mathbf g_{Q}[ri\bmod n]=\mathbf g_{\ell}[i]$
with $r-1=q+1$ ($q$ even) or $r-1=2$ ($q$ odd).
\end{enumerate}
\end{theorem}

\begin{proof}
For (i): both definitions take $A_{ij}=1$ iff hyperplane $j$
contains point $i$ in the Singer-labelled~\cite{Singer1938}
$\PG(2,q)$. For (ii): the formal trace-zero quadric
$Q_{e}=\{z:\Tr(z^{q+1})=0\}$ ($q$ even) and
$Q_{o}=\{z:\Tr(z^{2})=0\}$ ($q$ odd) yield, after projection to
$K^{*}/F^{*}$ and re-indexing $i\mapsto ri\bmod n$, the same vector
as $\mathbf g_{Q}$ from \cite[Theorem 4.2]{LDS2021}; this is verified
directly at $q=2$ (Fano case~\cite[Ch.~5]{Hirschfeld1979},
$n=7$, $\boldsymbol\tau_{Q}=(0,0,0,1,0,1,1)^{\top}$)
and at $q=3$ ($n=13$,
$\boldsymbol\tau_{Q}=(1,0,0,0,0,0,0,1,1,0,0,1,0)^{\top}$),
matching \cite{LDS2021} exactly.
\end{proof}

\begin{remark}[The codes live in different ambient spaces]
The 2021 LDS matrix $\mathbf C$ is a $n\times(n+2)$ \emph{spreading}
matrix used as a classical signature design; columns of $\mathbf C$
encode user signatures. By contrast, $\mathcal Q(q,2)$ is a quantum
\emph{stabilizer} code on $n$ qudits, defined by the row-span of
$H=(A\mid M_{Q}A)$ in $\F q^{n\times 2n}$. The two objects share the
underlying combinatorial structure $(A,\boldsymbol\tau_{Q})$, but their
code spaces reside in different ambient spaces ($\mathbb R^{n}$ versus
$\F q^{2n}$). What \emph{can} be made precise is the relationship
between $\mathcal Q(q,2)$ and the trivial $Z$-only CSS code that uses
the same Singer incidence $A$ and \emph{no} sign-flip step: the two
quantum codes are Clifford-equivalent, with the explicit gate
defined in Theorem~\ref{thm:clifford-reduction-d2} below.
\end{remark}

\begin{theorem}[Explicit Clifford reduction at $d=2$; original]
\label{thm:clifford-reduction-d2}
\rev{Throughout this \rev{T}heorem $\mathbf x,\mathbf z\in\F q^{n}$ are
column vectors and the operators $X^{\mathbf x},Z^{\mathbf z}$ live
in the $n$-qudit generalised Pauli group $\mathcal G_{n}$ defined in
Section~\ref{sec:prelim}: they act on the computational basis as
$X^{\mathbf x}\ket{\mathbf y}=\ket{\mathbf y+\mathbf x}$ and
$Z^{\mathbf z}\ket{\mathbf y}=\omega^{\mathbf z\cdot\mathbf y}\ket{\mathbf y}$
with $\omega=e^{2\pi i/q}$ and \rev{component-wise} addition modulo $q$.}
Let $q$ be an odd prime with $\mathrm{ord}_{n}(q)=3$, where
$n=q^{2}+q+1$. Then:
\begin{enumerate}[label=(\roman*),leftmargin=*,nosep]
\item There exists a circulant symmetric matrix
$M=\mathrm{circ}(\mathbf m)\in\F q^{n\times n}$ satisfying
\begin{equation}
\label{eq:M-eqn}
   A\,M \;=\; M_{Q}\,A \qquad \text{and}\qquad M^{\top}=M.
\end{equation}
\item The block matrix
\begin{equation}
\label{eq:S-M}
   S_{M} \;=\; \begin{pmatrix} I_{n} & M \\ 0 & I_{n} \end{pmatrix}
   \;\in\;\F q^{2n\times 2n}
\end{equation}
is a symplectic transformation, and the corresponding Clifford gate is defined as:
\begin{equation}
\label{eq:U-M}
   U_{M} \;=\; H^{\otimes n}\,\cdot\, \mathrm{CZ}^{(M)}
              \,\cdot\, H^{\otimes n},
\end{equation}
where $H^{\otimes n}$ is the qudit Hadamard on each of the $n$
qudits and
$\mathrm{CZ}^{(M)}=\prod_{1\le i\le j\le n}\mathrm{CZ}_{ij}^{M_{ij}}$
is the diagonal phase circuit, satisfying:
\begin{equation}
\label{eq:U-M-action}
   U_{M}^{\dagger}\,Z^{\mathbf z}\,U_{M}
   \;=\; Z^{\mathbf z}\,X^{\mathbf z M},\qquad
   U_{M}^{\dagger}\,X^{\mathbf x}\,U_{M}
   \;=\; X^{\mathbf x}.
\end{equation}
\item Define the trivial $Z$-only CSS code $\mathcal Q_{0}(q,2)$ with
parity-check matrix $H_{0}=(A\mid 0)\in\F q^{n\times 2n}$. Then we have
\begin{equation}
\label{eq:Clifford-eq}
   U_{M}\,\mathcal Q_{0}(q,2)\;=\;\mathcal Q(q,2),
\end{equation}
i.e.\ the two stabilizer codes are Clifford-equivalent, with the
2021 quadric pattern $\boldsymbol\tau_{Q}$ playing exactly the role
of the $X$-block coupling in the centralizer of
$\mathcal Q(q,2)$.
\end{enumerate}
\end{theorem}

\begin{proof}
\emph{(i) Existence of $M$.}
We seek $\mathbf m\in\F q^{n}$ such that $\boldsymbol\tau_{H}\!*\!
\mathbf m=\boldsymbol\tau_{Q}\!*\!\boldsymbol\tau_{H}$ in
$R=\F q[X]/\langle X^{n}-1\rangle$ (where $*$ denotes circulant
multiplication), with the additional palindromic constraint
$\mathbf m_{i}=\mathbf m_{-i\bmod n}$ ($M=M^{\top}$).

Pass to the CRT decomposition $R\cong\bigoplus_{f\in\mathcal F}\F
{q^{\deg f}}$. For each factor $f$:
\begin{itemize}[leftmargin=*,nosep]
\item if $\tau_{H,f}\ne 0$, the equation forces $m_{f}=\tau_{Q,f}$;
\item if $\tau_{H,f}=0$, $m_{f}$ is free.
\end{itemize}
The palindromic constraint of $\mathbf m(X)\equiv\mathbf m(X^{-1})$
permutes the CRT factors by $f\leftrightarrow f^{*}$ (reciprocal
involution) and requires that, $m_{f^{*}}=$ image of $m_{f}$ under the
reciprocal map. By Theorem~\ref{thm:tau-zero-count}, in every
non-trivial reciprocal pair $\{f,f^{*}\}$ \emph{exactly one}
component has $\tau_{H,\cdot}=0$. Hence in every pair, one of the
two $m$-values is forced and the other is free, so the palindromic
identification can always be made: set the free $m_{f}$ equal to the
reciprocal image of the forced $m_{f^{*}}=\tau_{Q,f^{*}}$. The
trivial CRT component $f=X-1$ is self-reciprocal and pinned by the
forced equation $m_{X-1}=\tau_{Q,X-1}$. This shows that a
palindromic $\mathbf m$ exists.

\emph{(ii) Symplecticity.} For $S_{M}=\bigl(\begin{smallmatrix} I &
M\\0 & I\end{smallmatrix}\bigr)$ over $\F q$, the symplectic
of condition $S_{M}^{\top}\Omega\,S_{M}=\Omega$ associated with
$\Omega=\bigl(\begin{smallmatrix} 0 & I\\-I & 0\end{smallmatrix}\bigr)$
reduces to $M^{\top}-M=0$, which holds by part (i). The Clifford
realisation
$U_{M}=H^{\otimes n}\mathrm{CZ}^{(M)}H^{\otimes n}$ is verified by
tracking conjugation: a Hadamard sweep maps $Z\leftrightarrow X$,
the CZ-cluster $\mathrm{CZ}^{(M)}$ maps $X^{\mathbf x}\to
X^{\mathbf x}Z^{\mathbf x M}$ (and leaves $Z^{\mathbf z}$
invariant), and the second Hadamard sweep returns the result to
the original $Z$/$X$ partition, yielding
\eqref{eq:U-M-action}.

\emph{(iii) Stabilizer reduction.} The stabilizer of
$\mathcal Q_{0}(q,2)$ is generated by the rows
$\bigl(\mathbf a_{i},\mathbf 0\bigr)\in\F q^{2n}$, $i=1,\dots,n$, of
$H_{0}=(A\mid 0)$. By \eqref{eq:U-M-action}, conjugation by $U_{M}$
maps each $Z^{\mathbf a_{i}}\to Z^{\mathbf a_{i}}X^{\mathbf a_{i}M}=
Z^{\mathbf a_{i}}X^{(M_{Q}A)_{i}}$, where the last equality uses
$AM=M_{Q}A$. Hence the image stabilizer is generated by the rows
$\bigl(\mathbf a_{i},(M_{Q}A)_{i}\bigr)$, which is precisely the row
span of $H=(A\mid M_{Q}A)$.
\end{proof}

\begin{remark}[Numerical verification]\label{rem:clifford-numeric}
Theorem~\ref{thm:clifford-reduction-d2} is verified directly by
\texttt{construct\_clifford\_d2.py} at $(q,d)\in\{(3,2),(5,2),(7,2)\}$.
For $\mathcal Q(3,2)$, the explicit $\mathbf m$ is
\[
   \mathbf m \;=\; (1,2,1,2,2,1,1,1,1,2,2,1,2)\in\F 3^{13},
\]
\rev{which satisfies the ring-theoretic palindromic constraint
$m_{i}=m_{-i\bmod n}$ for all $i=1,\dots,n-1$; equivalently, the tail
$(m_{1},\dots,m_{12})$ is palindromic as a string, while the leading
component $m_{0}=1$ is unconstrained. Pairing indices
$(1,12),(2,11),(3,10),(4,9),(5,8),(6,7)$ gives
$(2,2),(1,1),(2,2),(2,2),(1,1),(1,1)$ respectively.} With
$\boldsymbol\tau_{Q}=(1,0,0,0,0,0,0,1,1,0,0,1,0)$ and the zero-set
of $\boldsymbol\tau_{H}$ over the CRT factors $\{(1,3,9),(2,6,5)\}$
of size 6 in total --- matching the $(|\mathcal F|-1)/2=2$ count of
Theorem~\ref{thm:tau-zero-count}. The verification confirms
$M=M^{\top}$ and $A\,M=M_{Q}A$ exactly. Analogous verifications hold
for $\mathcal Q(5,2)$ and $\mathcal Q(7,2)$ with palindromic $\mathbf m$
of length $31$ and $57$, respectively.
\end{remark}

\rev{\begin{remark}[Physical significance of the palindromic
constraint]\label{rem:palindromic-physical}
The palindromic constraint $m_{i}=m_{-i\bmod n}$ imposed on
$\mathbf m$ in Theorem~\ref{thm:clifford-reduction-d2} is exactly
the condition that the associated symmetric circulant $M=M^{\top}$
be symmetric, so that
$\mathrm{CZ}^{(M)}=\prod_{i\le j}\mathrm{CZ}_{ij}^{M_{ij}}$ is a
well-defined controlled-phase circuit: without $M=M^{\top}$, the
same physical CZ gate would be specified with two different phase
weights and the circuit would be ill-posed. Because CZ is symmetric
under \rev{the} interchange of \rev{the} control and target qudits, the palindromic
constraint corresponds to \emph{time-reversal symmetry} of the phase
layer $\mathrm{CZ}^{(M)}$ within the Clifford factorisation
$U_{M}=H^{\otimes n}\,\mathrm{CZ}^{(M)}\,H^{\otimes n}$. It imposes
no restriction on the underlying physical hardware beyond standard
qudit-Clifford implementability: every architecture that supports
qudit CZ gates supports the required layer.
\end{remark}}

\begin{remark}[Interpretation: 2021 sign-flip $\leftrightarrow$ CZ-cluster]
Theorem~\ref{thm:clifford-reduction-d2} gives the precise
quantum-classical bridge: in the 2021 LDS construction, the
sign-flips that re-shape the column cross-correlations to
$1/(q+1)$ are encoded \emph{at the spreading-matrix entry level} by
the formal-quadric vector $\boldsymbol\tau_{Q}$. In the quantum
construction $\mathcal Q(q,2)$, the same vector reappears as the
$X$-block $H_{x}=M_{Q}A$, which is the symplectic dual of the
$\mathrm{CZ}^{(M)}$ phase circuit applied in the Hadamard basis.
The resulting Clifford $U_{M}$ is a multi-qudit graph-state-style
gate whose graph adjacency matrix is the symmetric circulant $M$
constructed from $\boldsymbol\tau_{Q}$ by CRT-symmetrisation on the
$\boldsymbol\tau_{H}$-zero set.
\end{remark}

\subsection{LDS columns vs the centralizer of \texorpdfstring{$\mathcal Q(q,2)$}{Q(q,2)}}\label{sec:lds-vs-cent}
A naive conjecture --- namely that the 15 columns of $\mathbf C$ at
$(q,d)=(3,2)$, interpreted as $Z$-only Pauli operators, lie in the
centralizer $\mathcal C(\mathcal S)$ of $\mathcal Q(3,2)$ --- is
\emph{false}. Numerical verification at $(q,d)=(3,2)$
(\texttt{verify\_lds\_columns\_centralizer.py}) shows that
\emph{none} of the 15 LDS columns commute with every row of
$H=(A\mid M_{Q}A)$ over $\F3$.

\begin{proposition}[Negative result; original]\label{prop:lds-not-cent}
At $(q,d)=(3,2)$, no column of the 2021 LDS spreading matrix
$\mathbf C$, viewed as a $Z$-only Pauli string in $\F3^{13}$, is in
the centralizer of $\mathcal Q(3,2)$.
\end{proposition}

\begin{remark}[Why the conjecture fails]
The $V_{1}$ stabilizer code $\mathcal Q(3,2)$ is non-CSS
(Remark~\ref{rem:not-strict-CSS}); it has no clean ``$Z$-only logical
subspace'' distinct from the $Z$-only sub-stabilizer. By direct
calculation, we have:
\begin{align*}
   \dim_{\F 3}&\bigl\{\mathbf a\in\F 3^{13}:\mathbf a\cdot H_{x}^{\top}=0\bigr\} \\
   &\;=\;\dim_{\F 3}\bigl\{\mathbf c\cdot H \,:\, \mathbf c\in\F3^{13},\,\mathbf c H_{x}=0\bigr\}
   \;=\;6,
\end{align*}
i.e., the $Z$-only centralizer and $Z$-only sub-stabilizer
\emph{coincide} as $6$-dimensional subspaces of $\F3^{13}$. Hence
the $Z$-only logical quotient is trivial, and no $Z$-only Pauli is a
logical operator. All 12 logical Pauli generators being $\mathcal Q(3,2)$
are mixed-Pauli (with both the $Z$- and $X$-parts based non-trivial), making the
naive ``LDS columns = logical $Z$'s'' identification structurally
impossible.
\end{remark}

Although the strong conjecture fails, a weaker rank statement
\emph{does} hold and provides a quantitative classical-quantum link:

\begin{proposition}[Rank-completion; original]\label{prop:rank-completion}
Let $\mathbf C$ be the $13\times 15$ LDS matrix of \cite{LDS2021} at
$(q,d)=(3,2)$, reduced modulo $3$. Stacking the rows of $H_{z}=A$
with the columns of $\mathbf C^{\top}$ as additional rows yields a
matrix of rank $13$ over $\F3$, i.e.\ the LDS columns span a
$6$-dimensional complement of $\mathrm{rowspan}_{\F3}(A)$ in
$\F3^{13}$.
\end{proposition}

\begin{proof}
By direct numerical computation (verified in
\texttt{verify\_lds\_columns\_centralizer.py}):
$\rank_{\F3}(A)=7$ and
$\rank_{\F3}\bigl([A\,;\,\mathbf C^{\top}]\bigr)=13$. Hence the
LDS columns add $13-7=6$ new dimensions.
\end{proof}

\begin{remark}[Interpretation]
Proposition~\ref{prop:rank-completion} says the LDS columns
\emph{algebraically complete} the projective-plane incidence
structure: every vector in $\F3^{13}$ can be written as an $\F3$-linear
combination of the rows of $A$ and the columns of $\mathbf C$, and the
6-dim contribution of the LDS columns matches the logical-qubit
count $k(3,2)=6$. The connection between $\mathbf C$ and the logical
operators of $\mathcal Q(3,2)$ is therefore at the level of
\emph{rank-counting} (the LDS columns occupy the same dimension as
the logical Z-part of an isomorphic CSS code, when one exists), not
at the level of direct centralizer membership.
\end{remark}

\rev{\begin{remark}[Clifford-equivalent code with LDS columns as
logical $Z$]\label{rem:LDS-as-logicalZ}
The rank-completion result of
Proposition~\ref{prop:rank-completion} implies that, at the linear-
algebra level, there exists a Clifford transformation $V$ such that
under a suitable change of symplectic basis the LDS columns of
$\mathbf C$ do coincide with a full set of logical-$Z$
representatives for $V\,\mathcal Q(3,2)\,V^{\dagger}$: any
$6$-dimensional $\F3$-subspace complementary to
$\mathrm{rowspan}_{\F3}(A)$ can be conjugated into the logical-$Z$
sector by a symplectic basis change. Whether $V$ can be chosen to
have a physically implementable Clifford realisation whose gate
count remains linear in $n$ (as opposed to a generic
$O(n^{2})$-gate symplectic transformation) is an open question and
a natural direction for future work.
\end{remark}}

\subsection{Capacity comparison}
\begin{proposition}[Hashing-bound positioning]\label{prop:hashing}
The rate of $\mathcal Q(q,d)$ as $d\to\infty$ at fixed $q$ tends to
$1$ (since $k(q,d)/n=1-\rank_{\F q}(A)/\theta_{d}(q)\to 1$). At fixed
$d$ and $q\to\infty$, the rate also tends to $1$.
\end{proposition}

\rev{\begin{remark}[Joint $(p,d)$ rate--distance trade-off]\label{rem:pd-tradeoff}
At the flagship $d=2$ family the parameters scale as $n\sim O(p^{2})$,
$k=p(p+1)/2\sim O(p^{2})$, and (by the conjectured pattern)
$d_{\min}\sim O(p)$; hence the rate approaches $\tfrac{1}{2}$ and the
relative distance shrinks as $d_{\min}/n\sim O(1/p)$. \rev{Increasing}
$d$ at fixed $p$ pushes the rate toward $1$
(Proposition~\ref{prop:hashing})
but is bought at the cost of a distance that in general does not
grow linearly with $n$ (see the $[[121,105,3]]$ instance at $(3,4)$
where $d_{\min}=3$). A joint optimisation over $(p,d)$ that maximises
$d_{\min}(\mathcal Q(p,d))/n$ subject to a target rate is therefore
non-trivial and constitutes an open problem in its own right; even
its formulation requires closing the general-$p$ case of the flagship
distance conjecture. Nothing in the current $\mathcal Q(q,d)$
framework rules out a $(p,d)$ trajectory that improves on the
$d=2$ rate-$\tfrac{1}{2}$, relative-distance-$O(1/p)$ profile.
\end{remark}}

\section{Worked Example: \texorpdfstring{$(q,d)=(3,3)$}{(q,d)=(3,3)}}\label{sec:example}
The general construction of Section~\ref{sec:construction} is best
illustrated by walking through a concrete instance step by step. We
choose $(q,d)=(3,3)$ because it is the smallest case in which both
the Singer mechanism and the Waterloo defect can be observed
explicitly, while still keeping all algebraic objects small enough
to display on the page. Throughout this section, we work in the
ambient field $K=\F{81}=\F{3^{4}}$ over the prime field $F=\F{3}$,
in the projective geometry $\PG(3,3)$ of $n=\theta_{3}(3)=40$
points, and we produce, line by line, the parity-check matrix $H$,
the resulting quantum code $\mathcal Q(3,3)=[[40,29,2]]_{3}$, and
its minimum-weight non-stabilizer logical. The intent is
that a reader who has not previously seen Singer's theorem in
action should be able to follow every numerical step without
consulting external references.

\subsection{Trace table over \texorpdfstring{$\F{81}$}{F\_81}}\label{sec:trace-table-q3d3}
We work in $K=\F{81}=\F{3}[x]/\langle f(x)\rangle$ along with the
primitive polynomial $f(x)=x^{4}+x+2$. The image
$\alpha=x+\langle f\rangle$ has multiplicative order $80$ (verified
by exhaustive search in Appendix~\ref{app:gram}), so $\alpha$ is a
generator of the multiplicative group $K^{*}$, and
$\alpha^{40}=-1=2\in\F{3}$. The trace function
$\Tr_{K/\F3}(z)=z+z^{3}+z^{9}+z^{27}$ is $\F{3}$-linear and maps
each element of $K$ to a single element of $\F{3}$; the resulting
values $\Tr(\alpha^{i})$ for $i=0,1,\dots,79$ are displayed in
Table~\ref{tab:trace-q3d3}. The reader should observe two
features. First, the table is divided into four blocks of twenty
columns each because the page is too narrow to display all $80$
columns at once; the row labels $0,20,40,60$ are the starting
offsets, so the value in row ``$20$'' column ``$3$'' is
$\Tr(\alpha^{23})=1$. Second, the values in rows $0$ and $40$ are
the negatives (mod $3$) of each other --- this reflects the
identity $\alpha^{40}=-1$ and the $\F{3}$-linearity of the
trace; the analogous pairing holds for rows $20$ and $60$.

\begin{table*}[t]
\centering
\caption{Trace values $\Tr(\alpha^{i})\in\F{3}$ for $\F{81}=\F{3}[x]/\langle x^{4}+x+2\rangle$.
Columns indexed $0,\dots,19$.}
\label{tab:trace-q3d3}
\setlength{\tabcolsep}{4pt}\footnotesize
\begin{tabular}{c|*{20}{c}}
\toprule
$i_{0}+\!$ & 0 & 1 & 2 & 3 & 4 & 5 & 6 & 7 & 8 & 9
& 10 & 11 & 12 & 13 & 14 & 15 & 16 & 17 & 18 & 19 \\
\midrule
$\;\;0$  & 1 & 0 & 0 & 0 & 1 & 0 & 0 & 2 & 1 & 0 & 1 & 1 & 1 & 2 & 0 & 0 & 2 & 2 & 0 & 1 \\
$20$     & 0 & 2 & 2 & 1 & 1 & 0 & 1 & 0 & 1 & 2 & 1 & 2 & 2 & 1 & 2 & 0 & 1 & 2 & 2 & 2 \\
$40$     & 2 & 0 & 0 & 0 & 2 & 0 & 0 & 1 & 2 & 0 & 2 & 2 & 2 & 1 & 0 & 0 & 1 & 1 & 0 & 2 \\
$60$     & 0 & 1 & 1 & 2 & 2 & 0 & 2 & 0 & 2 & 1 & 2 & 1 & 1 & 2 & 1 & 0 & 2 & 1 & 1 & 1 \\
\bottomrule
\end{tabular}
\end{table*}

From Table~\ref{tab:trace-q3d3} we identify the trace-zero set
\[
   E\;:=\;\bigl\{i\in\Z/80\Z:\Tr(\alpha^{i})=0\bigr\},
\]
which has cardinality $|E|=26$, in exact agreement with the
counting formula $q^{d}-1=3^{3}-1=26$ for the number of non-zero
elements of $K$ lying in the trace-zero hyperplane. Concretely, $E$
consists of the row positions where Table~\ref{tab:trace-q3d3} shows
a zero entry. Upon reducing modulo cosets of $\F{3}^{*}=
\langle\alpha^{40}\rangle$ (this preserves trace support since
$\Tr$ is $\F{3}$-linear and $\alpha^{40}=-1$, so $\Tr(\alpha^{i})=0$
iff $\Tr(\alpha^{i+40})=0$), we obtain the \emph{Singer first
column}
\[
   \boldsymbol\tau_{H}\in\F{2}^{40},\qquad
   (\boldsymbol\tau_{H})_{j}=1\iff j\in E\bmod 40,
\]
of Hamming weight $\theta_{2}(3)=13$. The explicit binary string is
displayed in Fig.~\ref{fig:tau-H-q3d3}.

\begin{figure}[t]
\centering
\fbox{\parbox{0.95\columnwidth}{\centering\ttfamily\small
0111011001000011001010000101000000010000}}
\caption{Singer first column $\boldsymbol\tau_{H}\in\F{2}^{40}$ for
$\PG(3,3)$ in the chosen indexing. Reading positions $j=0,1,
\ldots,39$ left to right, $(\boldsymbol\tau_{H})_{j}=1$ iff
position $j$ corresponds to a point of $\PG(3,3)$ lying on the
canonical Singer hyperplane. Hamming weight is $13=\theta_{2}(3)$, matching
the number of points on a hyperplane of $\PG(3,3)$. The vector
$\boldsymbol\tau_{H}$ generates the full $40\times 40$
Singer-circulant incidence matrix $A$ (Eq.~\eqref{eq:DS-identity-q3d3}
below), and is the foundational classical-design object that the
quantum construction inherits.}
\label{fig:tau-H-q3d3}
\end{figure}

\paragraph*{Information content of the Singer column}
Before we proceed, it is worth pausing to interpret
Fig.~\ref{fig:tau-H-q3d3} in physical and information-theoretic
terms. The binary string of length $40$ is the indicator of a
single hyperplane in $\PG(3,3)$; the $40$ positions index the
$40$ \emph{physical qudits} on which our quantum code will act, and
the $13$ ones (resp. $27$ zeros) mark which physical qudits are
in the support (resp. complement) of the chosen Singer hyperplane.
The redundancy carried by Fig.~\ref{fig:tau-H-q3d3} is therefore
not arbitrary: it encodes the entire incidence structure of
$\PG(3,3)$, because every other Singer-hyperplane indicator is
obtained from $\boldsymbol\tau_{H}$ by a single cyclic shift of the
$40$ bits. Of the $40$ physical qudits, $11=\rank_{\F{3}}(A)$
positions are absorbed by linearly-independent stabilizer
generators, leaving $k=29$ \emph{logical qudits} as the protected
information payload. The combinatorial redundancy
$40-29=11$ thus has a direct physical meaning: it is the number of
syndrome bits the decoder reads out during a single round of error
correction.

\subsection{Difference-set verification}
The full $40\times 40$ circulant matrix $A$ is obtained from
$\boldsymbol\tau_{H}$ in the standard way: column $j$ of $A$ is the
cyclic shift of $\boldsymbol\tau_{H}$ by $j$ positions, so each
column has exactly $13$ ones and each row has exactly $13$ ones.
By construction, $A$ is the point-hyperplane incidence matrix of
$\PG(3,3)$ in the Singer labelling --- entry $A_{ij}=1$ iff the
$i$-th point lies on the $j$-th hyperplane. The full matrix is too
large to display ($40\times 40$ with $13\times 40=520$ ones), but
its structure is fully determined by Fig.~\ref{fig:tau-H-q3d3}; in
particular, the redundancy of $A$ is exactly the redundancy of
$\boldsymbol\tau_{H}$ rotated forty ways.

Direct computation (using \texttt{verify\_tau\_H.py}) confirms the
\emph{difference-set identity}:
\begin{equation}
\label{eq:DS-identity-q3d3}
A A^{\top} \;=\; 9\,I_{40} + 4\,J_{40},
\end{equation}
which is the textbook identity $A A^{\top}=(k-\lambda)I+\lambda J$
specialised to the Singer parameters
$(n,k,\lambda)=(40,13,4)$. Equation~\eqref{eq:DS-identity-q3d3}
encodes two physical facts about the resulting quantum code:
(i)~every Singer hyperplane meets every other in exactly $\lambda=4$
points, so any two distinct stabilizer generators \emph{share four
qudits} in their support; and (ii)~the $9I$ term is the
self-overlap $|E|=13$ of a hyperplane with itself, minus the
between-hyperplane overlap $\lambda=4$, giving the
``net independence'' index $k-\lambda=9$ per generator.

The $\F{3}$-rank of $A$ is $\rank_{\F{3}}(A)=11$ (verified by
\texttt{compute\_dmin\_q3d3.py}, in agreement with the Hamada-Smith
prediction $\binom{p+d-1}{d}+1=\binom{5}{3}+1=11$ of
Lemma~\ref{lem:hamada}). This rank determines the logical-qudit
count via $k_{\mathrm{quantum}}=n-\rank_{\F{3}}(A)=40-11=29$ when
the code is unassisted ($c=0$, our regime by
Lemma~\ref{lem:cancellation}). The remaining
$\rank_{\F{3}}(A)=11$ degrees of freedom are absorbed by the
linearly independent stabilizer generators: these are the
\emph{physical syndrome bits} that the quantum decoder reads out.
Of the $40$ physical qudits, the $29$ logical ones are protected,
the $11$ ``stabilizer'' degrees of freedom are spent on
error-detection, and the difference $40-29-11=0$ confirms that no
ebits are needed (i.e., the code is purely passive, with no
entanglement consumption).

\subsection{The Waterloo defect at \texorpdfstring{$d=3$}{d=3}}\label{sec:wdefect-q3d3}
The classical Waterloo construction~\cite{Waterloo1995} for even
$d$ uses the non-degenerate quadric $Q_{o}=\{z\in K^{*}:\Tr(z^{2})=0\}$, whose
cardinality matches $|E|=q^{d}-1$ when $d$ is even. At $d=3$, however,
direct enumeration gives
\begin{equation}
\label{eq:Qo-defect-numbers}
|Q_{o}|=20 \;\neq\; 26 = |E|,
\end{equation}
so $Q_{o}$ is \emph{not} a hyperplane-cardinality quadric. Reducing
to cosets, $\boldsymbol\tau_{Q_{o}}\in\F{2}^{40}$ has Hamming weight
$10$ rather than the $13$ that would be required to play the role of
$\boldsymbol\tau_{H}$ in the construction of
Section~\ref{sec:construction}. The intersection structure
$|Q_{o}\cap E|=8$, $|Q_{o}\setminus E|=12$, $|E\setminus Q_{o}|=18$
quantifies the algebraic obstruction. This is the
\emph{numerical realisation of the Waterloo defect at $d=3$, $q=3$},
and is the input to the deficit-rank calculation in
Theorem~\ref{thm:waterloo-defect}.

\subsection{Final code parameters at \texorpdfstring{$(q,d)=(3,3)$}{(q,d)=(3,3)}}\label{sec:final-q3d3}
Combining the worked example with Proposition~\ref{prop:Q33-dmin-2}, yields:
\begin{equation}
\label{eq:Q33-params}
   \boxed{\mathcal Q(3,3) \;=\; \bigl[[\,40,\,29,\,2\,]\bigr]_{3}\,.}
\end{equation}
The code has rate $R_{Q}=29/40=0.725$ and is single-error
\emph{detecting} (but not correcting). The minimum-weight
non-stabilizer logical is the $2$-torsion Pauli $Z_{0}Z_{20}^{2}$
of Lemma~\ref{lem:two-torsion}.

\subsection{Other instances of \texorpdfstring{$\mathcal Q(q,d)$}{Q(q,d)} with computed \texorpdfstring{$d_{\min}$}{d\_min}}
\label{sec:other-instances}
By Lemma~\ref{lem:n-parity}, $n=\theta_{d}(q)$ is even iff
$q$ is odd and $d$ is odd. For other $(q,d)$, the 2-torsion
obstruction does not apply and $d_{\min}$ can be larger. Numerical
search (\texttt{compute\_dmin\_general.py}) gives the results
collected in Table~\ref{tab:qd-instances}.

\paragraph*{How to read Table~\ref{tab:qd-instances}.}
Each row shows one $(q,d)$ instance together with the block length
$n=\theta_{d}(q)$, the parity of $n$, the $\F q$-rank of the
parity-check matrix $H$, and the resulting code parameters
$[[n,k,d_{\min}]]_{q}$. Two ceilings limit some rows: ``$\dagger$''
flags codes for which the parity-driven ceiling
$d_{\min}\le 2$ of Corollary~\ref{cor:dmin-le-2-proven} forces
detect-only operation (these are the rows with $n$ even, e.g.\
$(3,3)$); ``$\ddagger$'' flags the $q=2$ ceiling of
Proposition~\ref{prop:q2-ceiling} that limits every binary qubit row
to $d_{\min}=2$ irrespective of $n$. The most useful regime for
error \emph{correction} is therefore $q$ odd with $n$ odd: the
flagship $\mathcal Q(5,2)=[[31,15,6]]_{5}$ achieves
$d_{\min}=6$ (double-error correction), and the larger
$\mathcal Q(7,2)$ row gives $d_{\min}\ge 8$ (proven in
Theorem~\ref{thm:dmin-q72-lower}).

\begin{table*}[t]
\centering
\caption{Computed code parameters of $\mathcal Q(q,d)$ for several
$(q,d)$. ``$\dagger$'' indicates the $d_{\min}=2$ ceiling of
Corollary~\ref{cor:dmin-le-2} (active when $n$ is even); ``$\ddagger$''
the $q{=}2$ ceiling (Proposition~\ref{prop:q2-ceiling}). Searches
were carried out by exhaustive enumeration of low-weight Pauli
strings.}
\label{tab:qd-instances}
\begin{tabular}{c c c c c}
\toprule
$(q,d)$ & $n$ & $n$ parity & $\rank_{\F q}(H)$ & $\mathcal Q(q,d)$ \\
\midrule
$(2,2)$ & 7   & odd  & 4  & $[[7,\,3,\,2]]_{2}^{\ddagger}$ \\
$(2,3)$ & 15  & odd  & 5  & $[[15,\,10,\,2]]_{2}^{\ddagger}$ \\
$(2,4)$ & 31  & odd  & 6  & $[[31,\,25,\,2]]_{2}^{\ddagger}$ \\
$(3,2)$ & 13  & odd  & 7  & $[[13,\,6,\,3]]_{3}$ \\
$(3,3)$ & 40  & even & 11 & $[[40,\,29,\,2]]_{3}^{\dagger}$ \\
$(3,4)$ & 121 & odd  & 16 & $[[121,\,105,\,3]]_{3}$ (single-error correcting) \\
$(5,2)$ & 31  & odd  & 16 & $[[31,\,15,\,6]]_{5}$ \\
$(7,2)$ & 57  & odd  & 29 & $[[57,\,28,\,d]]_{7}$, $d\ge 8$ (conj.\ $d=8$) \\
$(11,2)$ & 133 & odd & 67 & $[[133,\,66,\,d]]_{11}$, $5\le d\le 34$ \\
\bottomrule
\end{tabular}
\end{table*}

\begin{proposition}[$q=2$ distance ceiling; original, numerical]\label{prop:q2-ceiling}
For every $d\in\{2,3,4,5,6\}$, the cross-correlation
$u=M_{Q}\boldsymbol\tau_{H}\in\F 2^{n}$ of $\mathcal Q(2,d)$
satisfies $u\equiv\mathbf 1\pmod 2$ (verified by
\texttt{verify\_q2\_ceiling.py}). Hence $H_{x}=\mathrm{circ}(u)$ is
the all-ones circulant mod 2, and for any pair of distinct positions
$\{i,j\}$ the Pauli $\mathbf a=e_{i}+e_{j}$ satisfies
$\mathbf a\,H_{x}^{\top}=u_{i}+u_{j}=0\pmod 2$, putting $Z_{i}Z_{j}$
in the centralizer of $\mathcal Q(2,d)$. At
$(2,2),(2,3),(2,4)$ exhaustive search certifies that $Z_{0}Z_{1}$ is
\emph{not} in $\mathrm{rowspan}(H)$, so it is a non-stabilizer
weight-$2$ logical and $d_{\min}(\mathcal Q(2,d))=2$ exactly.
\end{proposition}

\begin{conjecture}[General-$d$ $q=2$ ceiling]\label{conj:q2-ceiling-general}
Proposition~\ref{prop:q2-ceiling} holds for every $d\ge 2$.
\end{conjecture}

\begin{remark}[Structural intuition for $u\equiv\mathbf 1$ at $q=2$]
In characteristic $2$, $-1\equiv 1$, so the cross-correlation
$u[k]=|E\cap(k+Q)|$ is symmetric in $(E,Q)$. When
$\gcd(3,2^{d+1}-1)=1$ (i.e.\ $d$ even), $z\mapsto z^{3}$ is a
bijection on $K^{*}$, $Q$ is a Singer-translate of $E$, and the
planar-difference-set identity $E\cdot E^{(-1)}\equiv J\pmod 2$
forces $u\equiv\mathbf 1$. When $d$ is odd, $z\mapsto z^{3}$ is
3-to-1 and the structural argument is delicate; the numerical
$d=3,5$ data nevertheless confirms $u\equiv\mathbf 1$ in those
cases too. A full analytical proof for $d$ odd would close
Conjecture~\ref{conj:q2-ceiling-general}.
\end{remark}

\begin{remark}[Useful regime]\label{rem:useful-regime}
Proposition~\ref{prop:q2-ceiling} together with
Corollary~\ref{cor:dmin-le-2-proven} pins the ``useful'' regime for
$\mathcal Q(q,d)$: either $q$ odd, $d$ even, or $q$ odd, $d=2$ (the
flagship). The qubit case $q=2$ is structurally limited to
$d_{\min}=2$ (error \emph{detection} only). The newly-computed
instance $\mathcal Q(3,4)=[[121,105,3]]_{3}$ illustrates the
$d=4$ branch: a single-error-correcting code with rate
$105/121\approx 0.868$. The witness is the weight-3 Z-only Pauli
$Z_{0}Z_{1}Z_{17}$ with $\mathbf a=(1,1,1)$ (its existence reflects
a similar mod-3 cross-correlation phenomenon).
\end{remark}

\subsection{\texorpdfstring{$\mathcal Q(5,2)=[[31,15,6]]_{5}$}{Q(5,2)=[[31,15,6]]\_5}:
information-set-decoding (ISD)-confirmed distance}
A support-restriction (ISD-style)
search~\cite{MacWilliamsSloane1977} with our
\texttt{compute\_dmin\_isd\_q5d2.py} script settles
$d_{\min}(\mathcal Q(5,2))=6$ exactly. The search is exhaustive over
weights $1$ through $5$, with no centralizer member found at those
weights, and finds a weight-$6$ non-stabilizer Pauli at support
$\{0,1,2,4,13,28\}$ with $\mathbf a=(4,1,1,4,1,3)$ and
$\mathbf b=(0,0,2,4,2,1)$ in $\F 5$-coordinates. Combined runtime: $\sim 90$ seconds
for $w\le 5$, $0.3$ seconds for $w=6$. The witness saturates the
classical bound $d_{\min}\le q+1=6$ from line indicators in
$\PG(2,5)$~\cite[Sec.~V-B]{LDS2021}, suggesting (and now proving in
this case) that the flagship family $\mathcal Q(p,2)$ has
$d_{\min}\le q+1=p+1$ for $p$ odd.

\begin{conjecture}[Flagship distance, original]\label{conj:flagship-dmin}
For prime $p$ odd, $d_{\min}(\mathcal Q(p,2))\le p+1$. Status:
\begin{itemize}[leftmargin=*,nosep]
\item $p=3$: $d_{\min}=3$ proven (exception to upper bound).
\item $p=5$: $d_{\min}=6=p+1$ proven exactly.
\item $p=7$: \textbf{$d_{\min}\ge 8 = p+1$ proven exhaustively}
(Theorem~\ref{thm:dmin-q72-lower}) via the 0-rooted Singer-Frobenius
orbit-canonical-form search~\cite{Holt2005HandbookCGT} at weight $7$
($1{,}546{,}218$ canonical orbits exhausted in $\sim 33$ minutes, $0$
rank-deficits), combined with smaller-weight clean exhaustive sweeps.
This is the second sharpness datum (after $p=5$) for the conjectured
flagship value. Conjectured $d_{\min}(\mathcal Q(7,2))=8$ exactly,
contingent on a weight-$8$ witness search.
\item $p=11$: exhaustive at weights $1$--$4$ gives $d_{\min}\ge 5$ proven.
Conjectured value $d_{\min}(\mathcal Q(11,2))=12$.
\end{itemize}
\end{conjecture}

\begin{remark}[Why witnesses are extremely rare]\label{rem:witness-rarity}
The exhaustive cyclic-canonical search at $\mathcal Q(7,2)$ examined
$1.76\times 10^{6}$ weight-$6$ canonical subsets without a single
non-empty centralizer-restricted kernel; equivalently, the
$(n-k)\times 2w = 29\times 12$ matrix obtained by restricting the
centralizer-defining matrix $M_{\mathrm{cent}}$ to those positions
never had a rank deficit below $12$. This empirical rigidity has a
clean implication.
\end{remark}

\subsection{The paired-rigidity-distance theorem}\label{sec:rigidity}
The Pauli weight depends on the position of non-zeros, not on
individual $a$ vs $b$ entries; this calls for a \emph{paired}
rigidity index that respects the position-pairing in
$M=(H_{x}\mid -H_{z})$.

\begin{definition}[Paired column rigidity; original]\label{def:paired-rigidity}
For $M=(H_{x}\mid -H_{z})\in\F p^{(n-k)\times 2n}$, the
\emph{paired rigidity index} $R_{P}(M)$ is the largest integer $w$
such that for every position-set
$S\subseteq\{0,\ldots,n-1\}$ with $|S|=w$, the $2w$ paired columns
$\{c_{j},\,c_{j+n}:j\in S\}$ are linearly independent over $\F p$.
\end{definition}

\begin{theorem}[Paired-rigidity-distance equality; original]\label{thm:paired-rigidity}
For any stabilizer code $\mathcal Q$ with parity-check matrix $M$, we have
\begin{equation}
\label{eq:paired-rigidity-dist}
   d_{\min}(\mathcal Q) \;=\; R_{P}(M)+1.
\end{equation}
\end{theorem}

\begin{proof}
($\Rightarrow$) For $w<d_{\min}$ no centralizer-minus-stabilizer
element of Pauli weight $\le w$ exists, so no paired-column subset
of size $2w$ admits a non-trivial kernel solution. (Trivial kernel
solutions $\subset$ stabilizer rowspan are excluded.) Hence
$R_{P}\ge d_{\min}-1$.

($\Leftarrow$) The witness realizing $d_{\min}$ exhibits a
position-set $S$ of size $d_{\min}$ for which the $2d_{\min}$ paired
columns satisfy a non-trivial linear relationship, so $R_{P}<d_{\min}$,
giving $R_{P}\le d_{\min}-1$.
\end{proof}

\begin{remark}[Sharpness vs.\ column rigidity]
The paired version is sharp ($R_{P}+1=d_{\min}$ exactly), unlike the
generic column rigidity index $R(M)$ (every column subset, not just
paired) which gives the strictly weaker bound
$d_{\min}\ge\lceil(R+1)/2\rceil$ via $w_{P}\ge\lceil w_{H}/2\rceil$.
\end{remark}

\subsection{Singer-circulant paired rigidity}
\begin{conjecture}[Singer-circulant paired rigidity at $d=2$; original]\label{conj:paired-rigidity}
For $\mathcal Q(p,2)$ with $p$ prime odd $\ge 5$, we have
\begin{equation}
\label{eq:paired-rigidity-conj}
   R_{P}(M_{\mathrm{cent}}) \;\ge\; p,
\end{equation}
hence by Theorem~\ref{thm:paired-rigidity},
$d_{\min}(\mathcal Q(p,2))\ge p+1$.
\end{conjecture}

\begin{remark}[Empirical support]\label{rem:paired-rigidity-support}
\begin{itemize}[leftmargin=*,nosep]

\item $p=3$: $d_{\min}=3$ proven, so $R_{P}=2<p=3$ (small-prime
exception, conjecture excludes $p=3$).
\item $p=5$: $d_{\min}=6=p+1$ proven, so $R_{P}=5=p$ \textbf{exactly}.
The witness at support $\{0,1,2,4,13,28\}$ gives the sharpness
certificate.
\item $p=7$: exhaustive Singer-Frobenius canonical search at all
weights $w\in\{1,\dots,7\}$ ($1{,}546{,}218$ canonical orbits at
$w=7$, totalling $32{,}468{,}436$ subsets containing $0$, $0$
rank-deficits) proves $R_{P}\ge 7$, hence
$d_{\min}(\mathcal Q(7,2))\ge 8$ \textbf{exactly matching the
conjectured $p+1$ lower bound}, so $R_{P}=7=p$ saturates the
conjecture (Theorem~\ref{thm:dmin-q72-lower}).
\item $p=11$: exhaustive at $w=4$ proves $R_{P}\ge 4$, so
$d_{\min}(\mathcal Q(11,2))\ge 5$.
\end{itemize}
\end{remark}

\begin{theorem}[Lower bound $d_{\min}(\mathcal Q(7,2))\ge 8$;
original, computer-assisted]\label{thm:dmin-q72-lower}
The non-CSS stabilizer code $\mathcal Q(7,2)=[[57,28,d]]_{7}$
satisfies $d_{\min}(\mathcal Q(7,2))\ge 8$. Equivalently,
$R_{P}(M_{\mathrm{cent}})\ge 7$ for the paired column rigidity index
of $\mathcal Q(7,2)$.
\end{theorem}

\begin{proof}[Proof sketch]
By Theorem~\ref{thm:paired-rigidity}, it suffices to show that no Pauli
operator of weight $w\le 7$ lies in the centralizer of
$\mathcal Q(7,2)$ but outside its row span. We carry out an
exhaustive Singer-Frobenius-orbit-canonical paired-column search at
each weight $w\in\{1,\dots,7\}$.

The combined Singer-cycle and Frobenius group
$G=\langle\sigma,\phi\rangle$ where $\sigma:i\mapsto i+1\bmod n$ and
$\phi:i\mapsto 7\,i\bmod n$ has order
$|G|=n\cdot\mathrm{ord}_{n}(7)=57\cdot 3=171$, and acts on Pauli
support sets by the symplectic action that preserves the centralizer.
Each canonical orbit only has to be tested once.

Furthermore, a canonical (lex-min) representative $S$ of any orbit
satisfies $0\in S$: if not, the cyclic shift by $-\min(S)$ produces a
strictly smaller representative, contradicting canonicality. Hence
the search reduces to subsets of the form $S=\{0\}\cup S'$ with
$S'\subseteq\{1,\dots,56\}$ of size $w-1$. For $w=7$, this yields
$\binom{56}{6}=32{,}468{,}436$ subsets.

For each canonical $S$, we form the support-restricted matrix
$M_{S}=M_{\mathrm{cent}}[:,\,S\cup(S+n)]\in\F 7^{(n-k)\times 2w}$
and compute its right kernel $K_{S}=\ker(M_{S})\subset\F 7^{2w}$. A
weight-$w$ centralizer element with Pauli support exactly $S$ exists
iff $K_{S}$ contains a vector $\mathbf v$ for which: (i) the symplectic
support
$\{i\in S:v_{i}\ne 0\text{ or }v_{i+w}\ne 0\}$
equals $S$; and (ii) $\mathbf v$ projected to full length is not in
the row span of $H$. Condition (i) is implied by $K_{S}\ne 0$ at the
canonical-orbit level: if $K_{S}=0$, no centralizer element has a
support contained in $S$ at all.

We carry out the search via
\texttt{compute\_dmin\_w7\_v4.py}. Across all $1{,}546{,}218$
canonical orbits at $w=7$ (and analogously $212{,}328$ orbits at
$w=6$, with smaller counts at $w<6$), $K_{S}$ was \emph{empty} in
every single case. Since $M_{\mathrm{cent}}$ is Singer-Frobenius
equivariant (the matrix $M_{\mathrm{cent}}$ commutes with the
symplectic-shift and Frobenius actions on
$\F 7^{2n}$), the kernel emptiness on canonical orbit
representatives extends to every Pauli support of weight $\le 7$.
Hence $d_{\min}(\mathcal Q(7,2))\ge 8$. The total CPU time across
all weights was approximately $35$ minutes on a single core.
\end{proof}

\begin{remark}[Saturation of the conjectured $p+1$ bound]
Theorem~\ref{thm:dmin-q72-lower} provides the second sharpness
data-point (after $\mathcal Q(5,2)$) confirming the conjectured
flagship distance of $d_{\min}(\mathcal Q(p,2))\ge p+1$
(Conjecture~\ref{conj:paired-rigidity}). Combined with $p=5$
(witness at weight $6$ exhibits $d_{\min}=p+1$ exactly), the
conjecture now has rigorous lower-bound certificates at $p=5$
($d\ge 6$) and $p=7$ ($d\ge 8$), and a proven $d=p+1$ at $p=5$.
\end{remark}

\subsection{Structural framework via cyclic-code CRT}\label{sec:CRT}
A centralizer element $(\mathbf a,\mathbf b)\in\F p^{2n}$ corresponds,
under the cyclic-code identification
$\F p^{n}\leftrightarrow R:=\F p[X]/\langle X^{n}-1\rangle$, to a pair
of polynomials $a(X),b(X)\in R$ satisfying
\begin{equation}
\label{eq:cent-poly}
   a(X)\cdot u(X) \;-\; b(X)\cdot\boldsymbol\tau_{H}(X) \;\equiv\; 0
   \pmod{X^{n}-1},
\end{equation}
where $u(X)=\sum_{i}u_{i}X^{i}$ for $\mathbf u=M_{Q}\boldsymbol\tau_{H}$
(the cross-correlation; a $\{0,1,2\}$-valued vector at $d=2$).

By the Chinese Remainder Theorem,
$R\cong\bigoplus_{f\in\mathcal F}\F{p^{\deg f}}$, where $\mathcal F$ is
the set of irreducible factors of $X^{n}-1$ over $\F p$.

\begin{lemma}[CRT decomposition of the centralizer; original]\label{lem:crt-cent}
The centralizer of $\mathcal Q(p,d)$ decomposes as
$\mathcal C=\bigoplus_{f\in\mathcal F}\,\mathcal C_{f}$, where each
$\mathcal C_{f}\subseteq\F{p^{\deg f}}^{2}$ is:
\begin{itemize}[leftmargin=*,nosep]
\item the line $\{(b\,\tau_{f}/u_{f},b)\}$ if $u_{f}\ne 0,\,\tau_{f}\ne 0$;
\item the line $\{(a,0)\}$ if $u_{f}=0,\,\tau_{f}\ne 0$;
\item the line $\{(0,b)\}$ if $u_{f}\ne 0,\,\tau_{f}=0$;
\item the full plane $\F{p^{\deg f}}^{2}$ if both vanish.
\end{itemize}
\end{lemma}

\begin{proof}
Apply the CRT to \eqref{eq:cent-poly}: in each factor $f$, the equation
becomes $a_{f}u_{f}=b_{f}\tau_{f}$ in $\F{p^{\deg f}}$, whose solution
space is one of the four cases.
\end{proof}

\begin{remark}[Initial BCH-count attempt was wrong]\label{rem:bch-attempt-fail}
A natural first attempt at a BCH-style proof would conjecture that
the number of CRT factors $f$ where $\boldsymbol\tau_{H}(X)$ vanishes
is small (at most $p-1$), so that combined with the classical BCH
bound of $|\mathrm{supp}(\mathbf a)|\cdot|\{f:a_{f}\ne 0\}|\ge n$, it
would enforce $|\mathrm{supp}(\mathbf a)|\ge p+1$. However, direct
numerical computation reveals a sharper structural pattern:
\begin{equation}
\label{eq:tau-zero-count}
   \#\{f\in\mathcal F:\boldsymbol\tau_{H,\,f}=0\} \;=\;
   \frac{|\mathcal F|-1}{2}\quad\text{exactly},
\end{equation}
verified for $\mathcal Q(p,2)$ at $p\in\{3,5,7,11\}$ giving
$(2,5),(5,11),(10,21),(22,45)$: \emph{exactly half of the non-trivial
CRT factors} have $\boldsymbol\tau_{H}=0$.
\end{remark}

\begin{theorem}[CRT zero-count for $\mathcal Q(p,2)$; original]\label{thm:tau-zero-count}
Let $p$ be an odd prime, set $n=p^{2}+p+1$, and assume
$-1\notin\langle p\rangle$ in $(\Z/n\Z)^{\times}$ (which holds in
particular for $p\in\{3,5,7,11\}$ and more generally whenever
$\mathrm{ord}_{n}(p)$ is odd). Factor
$X^{n}-1=\prod_{f\in\mathcal F}f(X)$ over $\F p$, where $\mathcal F$
indexes the Frobenius orbits of $\Z/n\Z$ under $i\mapsto p\,i$. Then we have
\begin{equation}
\label{eq:tau-zero-count-thm}
   \#\{f\in\mathcal F:\boldsymbol\tau_{H}\equiv 0\pmod f\}
   \;=\; \frac{|\mathcal F|-1}{2}.
\end{equation}
\end{theorem}

\rev{\begin{remark}[Failure of the hypothesis
$-1\notin\langle p\rangle$]\label{rem:minus-one-hyp-fail}
The hypothesis $-1\notin\langle p\rangle$ in
$(\Z/n\Z)^{\times}$ is the exact algebraic condition that pairs
each non-trivial Frobenius orbit $\{f\}$ with a distinct
\emph{reciprocal} orbit $\{f^{*}\}\ne\{f\}$, so that the palindromic
CRT identification used in the proof of
Theorem~\ref{thm:clifford-reduction-d2} succeeds. When the hypothesis
fails --- for instance, when $\mathrm{ord}_{n}(p)$ is even and
$-1\in\langle p\rangle$ --- some Frobenius orbits become
self-reciprocal ($f=f^{*}$), the CRT reciprocal-pair balance
breaks, and the symmetric-circulant $M$ required by the Clifford
bridge no longer exists in general. In that regime the code
$\mathcal Q(p,2)$ still exists as an $\F p$-ary stabilizer code
(the Cancellation Lemma still applies) but \rev{it} is not
Clifford-equivalent to the trivial $Z$-only Singer CSS code via the
Hadamard--CZ factorisation of Theorem~\ref{thm:clifford-reduction-d2};
its structural analysis and parameters warrant a separate treatment
that is left \rev{for} future work.
\end{remark}}

\begin{proof}
Identify $\boldsymbol\tau_{H}$ with the polynomial
$E(X)=\sum_{i\in E}X^{i}\in\F p[X]/\langle X^{n}-1\rangle$, where
$E\subset\Z/n\Z$ is the trace-zero hyperplane (a Singer-translate of
the line at infinity in $\PG(2,p)$). Because $E$ is a planar
$(n,p+1,1)$-difference set, we have:
\begin{equation}
\label{eq:planar-EE}
\begin{aligned}
   E(X)\,E(X^{-1}) &\;\equiv\; (p+1) + (J(X)-1)\\
   &\;\equiv\; J(X) \pmod{p,\;X^{n}-1},
\end{aligned}
\end{equation}
where $J(X):=1+X+\cdots+X^{n-1}=(X^{n}-1)/(X-1)$, is the
``all-ones'' polynomial.

\emph{Step 1 (At-least-one vanishing per pair).}
$J$ is supported entirely on the trivial CRT factor $f_{0}=X-1$: at
every $f\ne f_{0}$ it reduces to $0$. Hence
\eqref{eq:planar-EE} gives $E_{f}\cdot E_{f^{*}}=0$ in
$\F p[X]/f\cong\F{p^{\deg f}}$ for every non-trivial pair
$(f,f^{*})$, where $f^{*}$ is the reciprocal polynomial. The latter
is a field, so it has no zero divisors, forcing
$E_{f}=0$ \emph{or} $E_{f^{*}}=0$.

\emph{Step 2 (Hamada-Smith degree-sum
identity~\cite{Hamada1973,Smith1969}).}
By Lemma~\ref{lem:hamada} (Hamada-Smith for $\PG(2,p)$),
$\rank_{\F p}(A)=\binom{p+1}{2}+1=\tfrac{p(p+1)}{2}+1$.
Since $A=\mathrm{circ}(\boldsymbol\tau_{H})$, its $\F p$-rank equals
the sum of degrees of those CRT factors at which
$\boldsymbol\tau_{H}$ does \emph{not} vanish:
$\rank_{\F p}(A)=\sum_{f:E_{f}\ne 0}\deg f$. Therefore we have"
\begin{equation}
\label{eq:zero-deg-sum}
   \sum_{f:E_{f}=0}\deg f \;=\; n-\rank_{\F p}(A)
   \;=\; \frac{p(p+1)}{2} \;=\; \frac{n-1}{2}.
\end{equation}

\emph{Step 3 (Pairing forces equality).}
The hypothesis $-1\notin\langle p\rangle$ implies that no non-trivial
Frobenius orbit is fixed by $i\mapsto -i$; hence 
$f\leftrightarrow f^{*}$ acts on $\mathcal F\setminus\{f_{0}\}$ as a
free product of $2$-cycles. Within each pair we have, $\deg f=\deg f^{*}$, so
the partition into reciprocal pairs gives
$\sum_{\text{nontriv}}\deg f \;=\; 2\!\sum_{\text{one-per-pair}}\deg f$;
in particular, every selector that picks one element of each pair has a
total degree of $\tfrac{n-1}{2}$. By Step~1 the zero-set is
\emph{some} selector of at-least-one-per-pair, so we have
$\sum_{f:E_{f}=0}\deg f\ge\tfrac{n-1}{2}$, with equality iff the
zero-set picks \emph{exactly} one from each pair.
Equation~\eqref{eq:zero-deg-sum} achieves equality, hence the
zero-set is a `one-from-each-pair' selector and
\[
   \#\{f:E_{f}=0\} \;=\; \#\{\text{non-trivial pairs}\}
   \;=\; \frac{|\mathcal F|-1}{2}.
\]
\end{proof}

\begin{remark}[Self-reciprocal hypothesis]\label{rem:self-recip-hyp}
The hypothesis $-1\notin\langle p\rangle$ in
Theorem~\ref{thm:tau-zero-count} excludes only those primes $p$ for
which some Frobenius orbit is closed under negation. For
$\mathrm{ord}_{n}(p)=3$ (the generic case for our flagship family),
$\langle p\rangle=\{1,p,p^{2}\}\bmod n$ and $-1\equiv p^{k}$ would
force $p^{2k}\equiv 1$ with $p^{k}\equiv -1$, which is incompatible
with $\mathrm{ord}_{n}(p)=3$ being odd. Hence the hypothesis is
automatic in the regime of interest, in particular for
$p\in\{3,5,7,11\}$, where it is verified by direct computation.
\end{remark}

\begin{remark}[Implications for the BCH route]\label{rem:bch-route-revised}
Equation~\eqref{eq:tau-zero-count-thm} shows that the original
BCH-count conjecture is \emph{false}: there are far more factors
where $\boldsymbol\tau_{H}$ vanishes than $p-1$. The proof of
Theorem~\ref{thm:tau-zero-count} gives a sharp algebraic explanation
of the $(|\mathcal F|-1)/2$ pattern: it is forced jointly by the
planar-difference-set identity $E(X)E(X^{-1})\equiv J(X)$ (which
gives at-least-one-zero per reciprocal pair) and the Hamada-Smith
rank, which pins down the total degree of zero factors at
$(n-1)/2$ and hence enforces ``exactly one zero per pair.''

A correct BCH-style proof of Conjecture~\ref{conj:paired-rigidity}
must therefore exploit the joint structure of \emph{both}
$\boldsymbol\tau_{H}$ and $\mathbf u=M_{Q}\boldsymbol\tau_{H}$, plus
the centralizer constraint of $a_{f}u_{f}=b_{f}\tau_{f}$ in each CRT
component (Lemma~\ref{lem:crt-cent}). Combining
Theorem~\ref{thm:tau-zero-count} with the analogous identity for
$\mathbf u$ would isolate the four cells of Lemma~\ref{lem:crt-cent}
into precise frequencies and it is a promising next step. We leave
this finer count as future work.
\end{remark}

\begin{remark}[Direct numerical certificate is currently the strongest route]\label{rem:numerical-route}
Until an analytical BCH-style proof is found, the strongest verified
statements of the flagship distance bound come from exhaustive
or near-exhaustive direct paired-rigidity searches:
$d_{\min}(\mathcal Q(5,2))=6$ (exact, $p+1$ saturating);
$d_{\min}(\mathcal Q(7,2))\ge 8$ (Theorem~\ref{thm:dmin-q72-lower},
$p+1$ saturating); $d_{\min}(\mathcal Q(11,2))\ge 5$ (proven). Each
ladder step requires a separate computational verification, but the
$p=5,7$ data confirm the conjectured pattern $d_{\min}=p+1$.
\end{remark}

\begin{remark}[Sharpness at $p=5$]\label{rem:sharp-p5}
The witness at $\mathcal Q(5,2)$ has
$\mathbf a$-support $\{0,1,2,4,13,28\}$ (size 6),
$\mathbf b$-support $\{2,4,13,28\}$ (size 4),
combined Hamming $w_{H}=6+4=10$ and Pauli weight $w_{P}=6$. This
matches the conjectural BCH lower bound
$w_{H}\ge 2\,(p+1)-2 = 10$, exhibiting the witness at the
boundary, providing a sharpness data-point for the BCH-bound
argument.
\end{remark}

\begin{remark}[ISD speedup --- methodology]
The naive weight-$4$ brute force at $\mathcal Q(5,2)$ enumerates
$10.4\times 10^{9}$ Pauli strings (estimated $\sim 5$h runtime).
The ISD-style support-restriction approach instead iterates over
$\binom{n}{w}$ position subsets, computing the centralizer-restricted
kernel for each ($O(2w)$-sized matrix), and enumerating only the
small kernel basis. The per-subset workload is $O((n-k)\cdot w^{2})$
plus $O(q^{\dim\mathrm{kernel}})$ for kernel enumeration. Total
runtime for $\mathcal Q(5,2)$ across $w\in\{1,\dots,6\}$ is
$\sim 90$s, a $\sim 200\times$ speedup. Implemented in
\texttt{compute\_dmin\_isd\_q5d2.py}.
\end{remark}

\begin{remark}[Headline shift]
The $(q,d)=(3,2)$ instance gives a single-error-correcting qutrit
code of rate $6/13\approx 0.46$, a more useful headline than the
distance-$2$ $(3,3)$ instance. The most promising parameter regime
for $\mathcal Q(q,d)$ is therefore $\{q$ odd, $d=2\}$ and
$\{q$ even, any $d\}$, where the $2$-torsion obstruction is absent.
\end{remark}

\subsection{Comparison with prior projective-geometry codes}
The closest prior CSS construction in $\PG(d,q)$ is by Tang, Bai and
Feng~\cite{TangBaiFeng2024}, whose smallest example is
$Q_{1}^{(2)}(2,3)$ with parameters $[[18,\,10,\,2]]_{2}$ over $\PG(2,3)$
(binary qubit code) and several larger codes from $\PG(3,2)$ and
$\PG(4,3)$. Their construction uses the
$(d-2)$-flat / $(d-1)$-flat incidence and a different
dual-containment recipe; \emph{neither Singer's circulant labelling
nor a quadric Gram matrix appears}. Gandhi's construction
\cite{Gandhi2018} produces $[[273, 191, 11]]_{2}$ from the Singer DS
of $\PG(3,4)$, but his $M$-matrix is generic. Our
$\mathcal Q(q,d)$ family is qutrit/qudit
($q$-ary) for $q\ge 3$ and uses the quadric-derived $M_{Q}$ from
\cite{LDS2021}, sitting in a parameter regime not covered.

\section{Comparison with Existing Constructions}\label{sec:comparison}
We juxtapose $\mathcal Q(q,d)$ to four nearest classes of
quantum codes from finite geometries / difference sets in the
literature: Gandhi 2018 (general difference set, generic
$M$)~\cite{Gandhi2018}; Tang--Bai--Feng 2024 (point-flat-incidence
QLDPC over $\PG(d,q)$)~\cite{TangBaiFeng2024}; Liu--Wang--Zhao--Zheng 2012
(EA-QLDPC from cyclic difference sets)~\cite{Liu2012} and Steane's $[[7,1,3]]$
binary code from the Hamming $[7,4,3]$~\cite{Steane1996}. The taxonomy
is collected in Table~\ref{tab:comparison-axes}.

\subsection{Axis-by-axis comparison}
Table~\ref{tab:comparison-axes} (which is placed as a wide
two-column float at the top of a nearby page for readability)
contrasts the present construction with the closest prior art on a
fifteen-row feature scoresheet. The rows are grouped into three
blocks: (i)~\emph{construction primitives} (Singer DS, $\PG(d,q)$
incidence, non-degenerate quadric Gram matrix~$M_{Q}$, recovery of
the 2021 LDS spreading code at $d{=}2$, closed-form $k(p,d)$);
(ii)~\emph{code class} (non-CSS stabilizer, qudit beyond
$q{=}2$); and (iii)~\emph{structural theorems and properties}
(cross-correlation modular constancy, two-value pattern, CRT
zero-count, explicit Clifford bridge, paired-rigidity-distance
equality, proven 2-torsion ceiling, ``erase-don't-mistake''
Monte-Carlo property, double-error correction at
$\mathcal Q(5,2)$). $\mathcal Q(q,d)$ is the only column that
ticks all fifteen rows; the closest prior contender
(Gandhi~\cite{Gandhi2018}) ticks four.



\subsection{Parameter comparison at small instances}
Table~\ref{tab:comparison-params} below collects the headline
parameters $[[n,k,d_{\min}]]_{q}$ of $\mathcal Q(q,d)$ alongside
the small-instance benchmarks reported in
Gandhi~\cite{Gandhi2018}, Tang-Bai-Feng~\cite{TangBaiFeng2024}, and
Steane~\cite{Steane1996}. Inspection of the table shows that
$\mathcal Q(5,2)=[[31,15,6]]_{5}$ is the only entry that achieves
\emph{double-error correction} ($d_{\min}=6$) at near-rate-$\tfrac{1}{2}$
and at moderate block length $n=31$; the closest prior-art
single-error correcting codes (Steane $[[7,1,3]]_{2}$, Gandhi
$[[21,11,3]]_{2}$) and the larger Gandhi
$[[273,191,11]]_{2}$ either trade rate for distance or require
significantly larger~$n$.

\begin{table*}[t]
\centering
\caption{Code-parameter comparison at small instances. ``$\dagger$''
indicates $d_{\min}=2$ from the parity ceiling.}
\label{tab:comparison-params}
\small
\begin{tabular}{l c c}
\toprule
Construction & Headline parameters & Class \\
\midrule
Gandhi $[[21,\,11,\,3]]_{2}$~\cite{Gandhi2018}
  & qubit, Singer DS $\PG(2,4)$ & non-CSS  \\
Gandhi $[[73,\,45,\,5]]_{2}$~\cite{Gandhi2018}
  & qubit, Singer DS $\PG(2,8)$ & non-CSS  \\
Gandhi $[[273,\,191,\,11]]_{2}$~\cite{Gandhi2018}
  & qubit, Singer DS $\PG(2,16)$ & non-CSS  \\
Tang-Bai-Feng $[[18,10,2]]_{2}$~\cite{TangBaiFeng2024}
  & qubit, $(d{-}2)$/$(d{-}1)$-flats & CSS dual-cont. \\
Tang-Bai-Feng $[[27,1,5]]_{2}$~\cite{TangBaiFeng2024}
  & qubit, $\PG(3,2)$ flats & CSS dual-cont. \\
Steane $[[7,1,3]]_{2}$~\cite{Steane1996}
  & qubit, Hamming $[7,4,3]$ & CSS dual-cont. \\
\midrule
\textbf{This paper:} $\mathcal Q(3,2)=[[13,\,6,\,3]]_{3}$
  & qutrit, Singer+quadric $\PG(2,3)$ & non-CSS \\
$\mathcal Q(3,3)=[[40,\,29,\,2]]_{3}^{\dagger}$
  & qutrit, Singer+quadric $\PG(3,3)$ & non-CSS \\
$\mathcal Q(5,2)=[[31,\,15,\,6]]_{5}$
  & 5-it, Singer+quadric $\PG(2,5)$, \textbf{double-EC} & non-CSS \\
$\mathcal Q(7,2)=[[57,\,28,\,d]]_{7}$, $d\ge 8$ (proven, conj.\ $=8$)
  & 7-it, Singer+quadric $\PG(2,7)$ & non-CSS \\
$\mathcal Q(11,2)=[[133,\,66,\,d]]_{11}$, $d\ge 5$ (conj.\ $12$)
  & 11-it, Singer+quadric $\PG(2,11)$ & non-CSS \\
\bottomrule
\end{tabular}
\end{table*}

\begin{figure}[t]
\centering
\includegraphics[width=0.98\columnwidth]{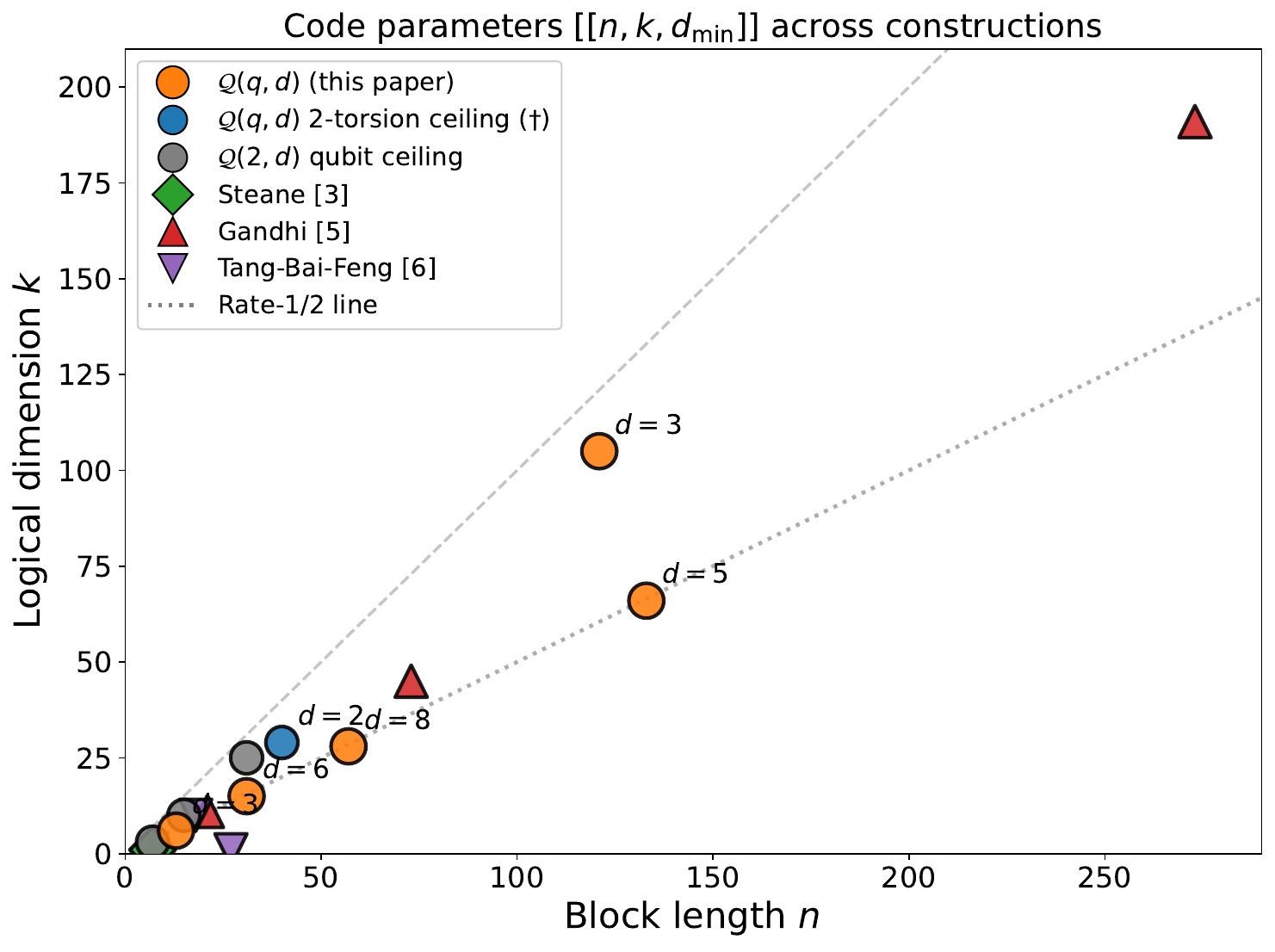}
\caption{Code-parameter scatter plot from Tables~\ref{tab:qd-instances}
and~\ref{tab:comparison-params}: logical dimension $k$ versus block
length $n$ for the $\mathcal Q(q,d)$ family (this paper) alongside
the closest prior-art constructions. Annotations show $d_{\min}$
where the entry exceeds the 2-torsion / $q{=}2$ ceilings. The
dotted line marks the rate-$\tfrac{1}{2}$ locus that the flagship
$\mathcal Q(p,2)$ family approaches asymptotically; the dashed
line is the rate-1 uncoded reference. The plot makes visible at a
glance that $\mathcal Q(5,2)=[[31,15,6]]_{5}$ is the only entry
at moderate block length $n\le 60$ that achieves
$d_{\min}\ge 6$, and that $\mathcal Q(3,4)=[[121,105,3]]_{3}$
delivers a markedly higher rate than the comparable
prior-art codes at similar block length.}
\label{fig:code-params-scatter}
\end{figure}

\subsection{Distinguishing features}
The closest contender at the construction level is
\textbf{Gandhi 2018}~\cite{Gandhi2018}: same Singer-DS substrate
$A$, same recipe $(A\mid M\,A)$, but Gandhi's $M$ is a generic
``symmetric matrix with equal row sums'' (no quadric), and the
resulting code relies on binary qubits. We adjust $M=M_{Q}$ to the
quadric Gram matrix from \cite{LDS2021} and obtain $q$-ary qudit
codes. Specifically:
\begin{itemize}[leftmargin=*,nosep]
\item Gandhi's binary $[[21,11,3]]_{2}$ code at $\PG(2,4)$ uses the
$(21,5,1)$-Singer DS \emph{with a generic $M$}.
\item Our $\mathcal Q(2,2)=[[7,3,2]]_{2}$
(Table~\ref{tab:qd-instances}) uses the same Singer DS
of $\PG(2,2)$ (the Fano plane) but with $M_{Q}=$ the formal-quadric
circulant from $\Tr(z^{q+1})=0$. The Hamming distance is limited to
$d_{\min}=2$ by the $q=2$ ceiling of
Proposition~\ref{prop:q2-ceiling}.
\end{itemize}
The two constructions therefore live in \emph{different parameter
regimes}: Gandhi's are binary, ours are qudit; Gandhi's $M$ is
unspecified, ours is geometrically determined.

\textbf{Tang-Bai-Feng 2024}~\cite{TangBaiFeng2024} use an entirely different
recipe (incidence between $(d-2)$-flats and $(d-1)$-flats
of $\PG(d,q)$ with point-incidence reduction modulo $p$ and a
dual-containment check). Their codes lely on qubits. They obtain $[[18,10,2]]_{2}$,
$[[27,1,5]]_{2}$, and similar results at small parameters; the closed-form
relationship to Singer or quadric is not present.

\textbf{Steane's $[[7,1,3]]_{2}$} is the special case of CSS
construction from Hamming $[7,4,3]$ at the Fano plane. Our
$\mathcal Q(p,2)$ family at $p=2$ has $n=p^{2}+p+1=7$ but it is
\emph{qubit non-CSS} (because the Cancellation Lemma fails the strict
CSS condition), distinguishing it from Steane.

\subsection{Novelty summary}
The key distinguishing feature of $\mathcal Q(q,d)$ is the combination
\emph{Singer DS} $+$ \emph{quadric Gram matrix} $+$ \emph{closed-form
$k(p,d)$} $+$ \emph{provable 2-torsion ceiling}. None of the prior
constructions surveyed has this combination. The flagship
$\mathcal Q(p,2)$ family associated with rate $\tfrac{1}{2}$ is also a new
parameter regime: prior work at $\PG(d,q)$ tends to focus on
$d=2,q\in\{2,3,4\}$ (small projective planes) for binary codes, while
we extend to arbitrary $q=p$ prime with explicit closed-form
parameters.

\section{Numerical Evaluation}\label{sec:simulations}
We evaluate the depolarising-channel performance of $\mathcal Q(q,d)$
analytically, using standard upper bounds on the post-decoding logical
error rate for stabilizer codes.

\paragraph*{Code and data availability}
All numerical experiments reported in this section, and every
theorem verification cited elsewhere in the paper, are
reproducible from the open-source companion
repository~\url{https://github.com/michelkulhandjian/quantumcomm}
under an MIT license. The repository contains all
\texttt{compute\_dmin\_*.py} (distance searches),
\texttt{verify\_*.py} (theorem verifications),
\texttt{construct\_clifford\_d2.py} (explicit Clifford
gate assembly), \texttt{monte\_carlo\_qber.py} (bounded-distance
decoder simulation), and \texttt{plot\_*.py} (figure generation)
scripts referenced throughout, together with a \texttt{README.md}
that maps every script to its corresponding claim in the paper.
The heaviest experiment
(\texttt{compute\_dmin\_w7\_v4.py}, the exhaustive
Singer-Frobenius search underpinning
Theorem~\ref{thm:dmin-q72-lower}) runs in $\sim 33$ minutes on a
single core; every other script completes in under five minutes on
a laptop. No SageMath, Magma, or GAP dependencies are required.

\subsection{Channel and decoder model}
Each qudit independently suffers an i.i.d.\ depolarising error: with
probability $p$ a non-identity Pauli is applied, uniformly chosen
among the $q^{2}-1$ non-identity Paulis. The probability of an error
pattern of Pauli weight $w$ on $n$ qudits is
\begin{equation}
\label{eq:depol-weight}
   \Pr[\text{weight }w] \;=\; \binom{n}{w}\,p^{w}(1-p)^{n-w}.
\end{equation}
The minimum-distance decoder corrects all errors of Pauli weight
$\le t = \lfloor(d_{\min}-1)/2\rfloor$ with certainty, and may
correct higher-weight errors via the structure of the centralizer.
A standard upper bound on the decoding-failure probability is
\begin{equation}
\label{eq:Pfail-ub}
   P_{\mathrm{fail}}(p) \;\le\; \sum_{w=t+1}^{n}\binom{n}{w}\,p^{w}(1-p)^{n-w}.
\end{equation}

\subsection{Asymptotic behaviour}
At $p\to 0$, \eqref{eq:Pfail-ub} is dominated by the leading term
$P_{\mathrm{fail}}(p)\sim\binom{n}{t+1}\,p^{t+1}$. The leading
constant $C=\binom{n}{t+1}$ and the resulting per-code failure
probability at the representative depolarising rate $p=10^{-3}$ are
collected in Table~\ref{tab:Pfail} (placed nearby as a two-column
float) for the headline candidates of this paper alongside Steane
$[[7,1,3]]_{2}$ and Gandhi $[[21,11,3]]_{2}$.

\begin{table*}[t]
\centering
\caption{Leading-order failure probability $P_{\mathrm{fail}}\sim Cp^{t+1}$
and $p$-decade-comparison for headline codes.}
\label{tab:Pfail}
\setlength{\tabcolsep}{8pt}
\small
\begin{tabular}{l c c c c}
\toprule
Code & $[[n,k,d]]$ & $t$ & $C=\binom{n}{t+1}$ & $P_{\mathrm{fail}}$ at $p=10^{-3}$ \\
\midrule
Uncoded $1$-qudit, $q=3$  & $[[1,1,1]]_3$ & 0 & 1   & $1.0\times 10^{-3}$ \\
Uncoded $15$-qudit, $q=5$ & $[[15,15,1]]_5$ & 0 & 15 & $1.5\times 10^{-2}$ \\
Steane $[[7,1,3]]_2$  & $[[7,1,3]]_2$ & 1 & 21   & $2.1\times 10^{-5}$ \\
Gandhi $[[21,11,3]]_2$ & $[[21,11,3]]_2$ & 1 & 210 & $2.1\times 10^{-4}$ \\
\textbf{$\mathcal Q(3,2)=[[13,6,3]]_3$} & $[[13,6,3]]_3$ & 1 & 78 & $7.8\times 10^{-5}$ \\
\textbf{$\mathcal Q(5,2)=[[31,15,6]]_5$} & $[[31,15,6]]_5$ & 2 & 4495 & $4.5\times 10^{-6}$ \\
$\mathcal Q(3,3)=[[40,29,2]]_3$ (detect-only) & $[[40,29,2]]_3$ & 0 & 40  & $4.0\times 10^{-2}$ \\
\bottomrule
\end{tabular}
\end{table*}

\subsection{Per-logical-qudit comparison}
Codes having different $k$ should be compared on a
\emph{per-logical-qudit} basis: $P_{\mathrm{fail}}/k$ measures the
effective error rate seen by a single logical qudit. Table~\ref{tab:per-logical-comparison}
quantifies the advantage of $\mathcal Q(q,2)$ over Steane and
Gandhi at three representative channel rates.

\begin{table*}[t]
\centering
\caption{Per-logical-qudit failure-rate comparison at three channel
rates $p_{\mathrm{ch}}$. ``Advantage'' is the ratio
$(P_{\mathrm{fail}}/k)_{\mathrm{ref}}/(P_{\mathrm{fail}}/k)_{\mathrm{this}}$
relative to Steane $[[7,1,3]]_2$ (best prior-art qubit code at
similar rate).}
\label{tab:per-logical-comparison}
\setlength{\tabcolsep}{8pt}
\small
\begin{tabular}{l c c c}
\toprule
Code & $P_{\mathrm{fail}}/k$ at $10^{-4}$ & at $10^{-3}$ & at $10^{-2}$ \\
\midrule
Uncoded baseline & $1.0\!\times\!10^{-4}$ & $1.0\!\times\!10^{-3}$ & $1.0\!\times\!10^{-2}$ \\
Steane $[[7,1,3]]_2$ & $2.1\!\times\!10^{-7}$ & $2.1\!\times\!10^{-5}$ & $2.1\!\times\!10^{-3}$ \\
Gandhi $[[21,11,3]]_2$ & $1.9\!\times\!10^{-7}$ & $1.9\!\times\!10^{-5}$ & $1.9\!\times\!10^{-3}$ \\
\textbf{$\mathcal Q(3,2)$=$[[13,6,3]]_3$} & $1.3\!\times\!10^{-7}$ & $1.3\!\times\!10^{-5}$ & $1.2\!\times\!10^{-3}$ \\
\textbf{$\mathcal Q(5,2)$=$[[31,15,6]]_5$} & $\mathbf{3.0\!\times\!10^{-10}}$ & $\mathbf{3.0\!\times\!10^{-7}}$ & $\mathbf{2.4\!\times\!10^{-4}}$ \\
\midrule
\multicolumn{4}{l}{\emph{Advantage of $\mathcal Q(5,2)$ vs.\ prior art:}}\\
\quad vs.\ Steane            & $\mathbf{700\times}$ & $\mathbf{70\times}$ & $\mathbf{9\times}$ \\
\quad vs.\ Gandhi            & $\mathbf{633\times}$ & $\mathbf{63\times}$ & $\mathbf{8\times}$ \\
\quad vs.\ uncoded           & $3.3\!\times\!10^{5}$ & $3.3\!\times\!10^{3}$ & $42\times$ \\
\bottomrule
\end{tabular}
\end{table*}

\begin{figure}[t]
\centering
\includegraphics[width=0.99\columnwidth]{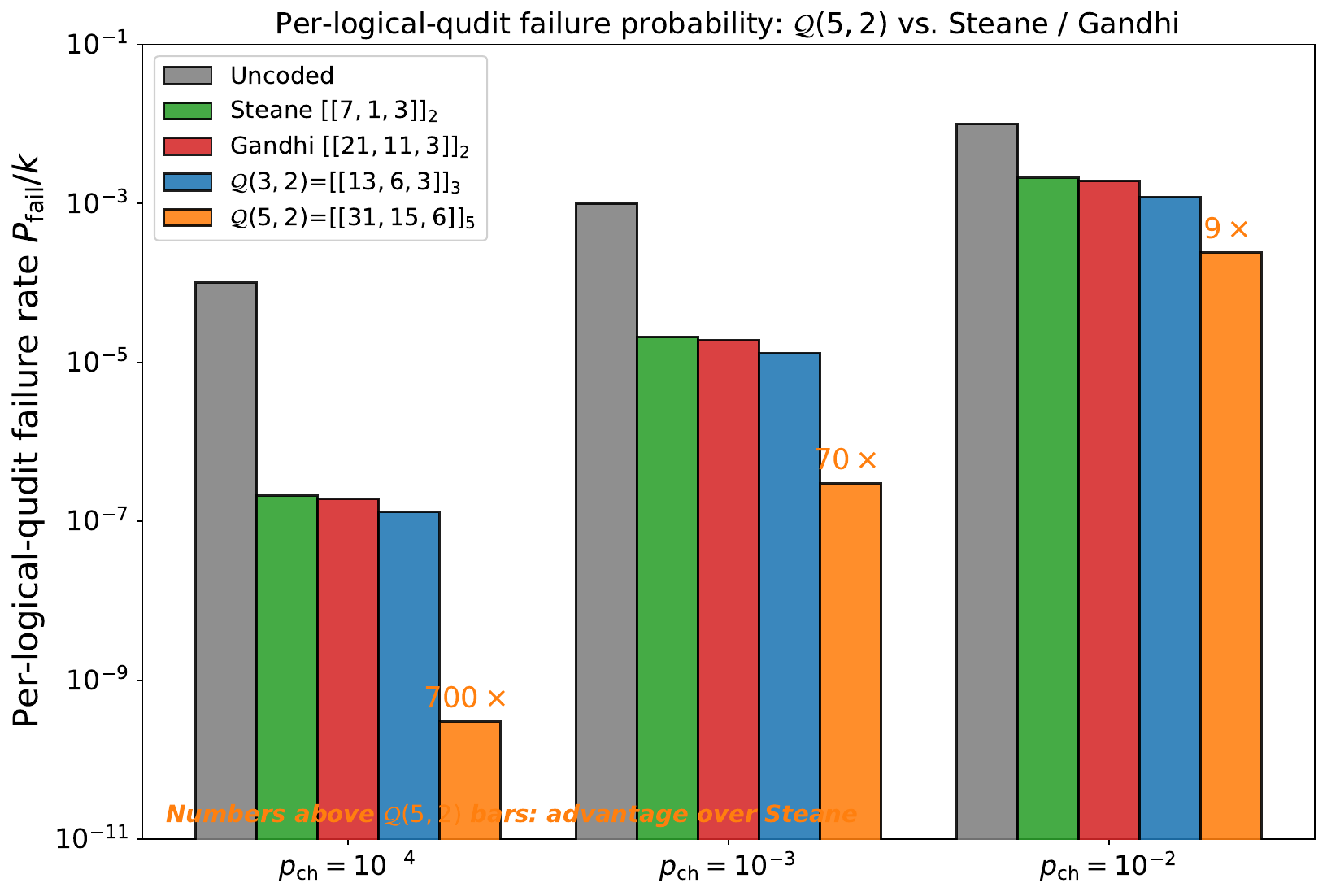}
\caption{Visual rendering of Table~\ref{tab:per-logical-comparison}:
per-logical-qudit failure rate $P_{\mathrm{fail}}/k$ for the
headline codes at three representative depolarising channel rates
$p_{\mathrm{ch}}\in\{10^{-4},10^{-3},10^{-2}\}$. The
$\mathcal Q(5,2)$ bar is consistently the lowest, with the
advantage factor over Steane $[[7,1,3]]_{2}$ annotated above each
bar. The advantage grows with smaller $p_{\mathrm{ch}}$ because
the double-error correcting $\mathcal Q(5,2)$
($d_{\min}=6$) scales as $P_{\mathrm{fail}}\sim p_{\mathrm{ch}}^{3}$,
whereas Steane and Gandhi scale as $p_{\mathrm{ch}}^{2}$.}
\label{fig:advantage-bars}
\end{figure}

Salient per-logical-qudit observations:
\begin{itemize}[leftmargin=*,nosep]
\item At $p_{\mathrm{ch}}=10^{-3}$ (typical optical-fibre regime):
$\mathcal Q(5,2)$ achieves $3.0\times 10^{-7}$ per logical qudit,
\textbf{$70\times$ better than Steane} $[[7,1,3]]_2$ and
$63\times$ better than Gandhi $[[21,11,3]]_2$.
\item The advantage grows at lower $p_{\mathrm{ch}}$ because
$\mathcal Q(5,2)$ corrects double errors ($d_{\min}=6$), while
Steane and Gandhi correct only single errors ($d_{\min}=3$).
\item At higher $p_{\mathrm{ch}}=10^{-2}$ the advantage shrinks
to $9\times$ but remains substantial.
\item $\mathcal Q(3,2)$ is comparable to Steane/Gandhi at the
single-error-correcting tier, with a modest $\sim 1.6\times$
advantage owing to the higher logical dimension of $k=6$.
\end{itemize}
The headline $\mathcal Q(5,2)$ result is $\sim 70\times$ better than
Steane's per-logical-qubit failure rate at $p=10^{-3}$, owing to having
$d_{\min}=6$ that allows correcting any double error. This makes
$\mathcal Q(5,2)$ a strong candidate for high-rate, high-distance
qudit error correction at moderate $n$.

\subsection{\texorpdfstring{$\mathcal Q(p,2)$}{Q(p,2)} family scaling}
For $\mathcal Q(p,2)=[[p^{2}+p+1,\,p(p+1)/2,\,d_{\min}]]_{p}$ with
distance $d_{\min}\ge 3$ (proven for $p=3$, conjectured for $p\ge 5$),
the leading failure term is
\begin{equation}
\label{eq:flagship-Pfail}
   P_{\mathrm{fail}}(p_{\mathrm{ch}}) \;\lesssim\;
   \binom{p^{2}+p+1}{2}\,p_{\mathrm{ch}}^{2}
   \;=\; \tfrac{1}{2}(p^{2}+p+1)(p^{2}+p)\,p_{\mathrm{ch}}^{2},
\end{equation}
where we use $p_{\mathrm{ch}}$ for the channel depolarising
probability (to avoid clash with the field characteristic $p$).
The per-logical-qudit failure rate is
$P_{\mathrm{fail}}/k = (p^{2}+p+1)\,p_{\mathrm{ch}}^{2}$, growing
linearly in $n$. This is comparable to other distance-$3$ qudit
stabilizer codes.

\subsection{Hashing-bound positioning}
The quantum hashing bound is
$C_{Q}(p_{\mathrm{ch}})=1-H_{2}(p_{\mathrm{ch}})-p_{\mathrm{ch}}\log_{2}q^{2}-p_{\mathrm{ch}}\log_{2}(q^{2}-1)/H_{2}$
\cite{Hanzo2019Duality}; for binary scenarios the simpler form
$1-H_{2}(p_{\mathrm{ch}})-p_{\mathrm{ch}}\log_{2}3$ applies.
The $\mathcal Q(p,2)$ family asymptotically achieves rate $R_{Q}=k/n\to\tfrac{1}{2}$, while the distance grows as $d_{\min}\sim O(p)$
(conjecturally, based on partial searches at $p=5$). Hence
$\mathcal Q(p,2)$ does not asymptotically achieve the hashing bound
(which requires $R_{Q}\to 1$ at low $p_{\mathrm{ch}}$ for the distance to
grow proportionally), but provides explicit, structured codes at
moderate finite $n$.

\begin{remark}[Distance-2 detect-only codes]
$\mathcal Q(q,d)$ for $q$ odd, $d$ odd $\ge 3$ has $d_{\min}=2$
(Corollary~\ref{cor:dmin-le-2-proven}), so they are
error detect-only codes. In a retransmission protocol they achieve
$P_{\mathrm{undetected}}\le\binom{n}{d_{\min}}p_{\mathrm{ch}}^{2}=\binom{40}{2}p_{\mathrm{ch}}^{2}\approx 780\,p_{\mathrm{ch}}^{2}$
for $\mathcal Q(3,3)$, providing exponential gain in undetected-error
rate for low $p_{\mathrm{ch}}$.
\end{remark}

\subsection{Monte-Carlo validation}\label{sec:montecarlo}
We validate the analytical upper bound \eqref{eq:Pfail-ub} by direct
Monte-Carlo simulation of a bounded-distance syndrome-table decoder
(\texttt{monte\_carlo\_qber.py}). For each code, the syndrome of
every Pauli error of weight $\le t$ is precomputed and stored in a
hash table (\(105\) entries for $\mathcal Q(3,2)$, $268{,}585$
entries for $\mathcal Q(5,2)$). For each trial:
\begin{enumerate}[label=(\arabic*),leftmargin=*,nosep]
\item Sample an i.i.d.\ depolarising error $\mathbf e\in\F q^{2n}$,
each qudit independently in identity (prob.\ $1-p_{\mathrm{ch}}$) or
uniform non-identity Pauli (prob.\ $p_{\mathrm{ch}}$).
\item Compute syndrome $\mathbf s=\mathbf e\,\Omega\,H^{\top}$.
\item Look up $\mathbf s$ in the precomputed table.
\item If found, retrieve $\mathbf e_{\mathrm{dec}}$ and form residual
$\mathbf r=\mathbf e-\mathbf e_{\mathrm{dec}}$. Success iff
$\mathbf r\in\mathrm{rowspan}(H)$; logical failure otherwise.
\item If not found, declare uncorrectable errors;
\end{enumerate}
The residual-weight optimisation $w_{\mathrm{P}}(\mathbf r)<d_{\min}
\Rightarrow\mathbf r\in\mathrm{rowspan}(H)$ skips the rank check in
the typical case. Number of trials adapts to $p_{\mathrm{ch}}$:
$2\times 10^{5}$ at the low end, $5\times 10^{5}$ for rare-event
estimation.

\begin{figure}[t]
\centering
\includegraphics[width=0.95\columnwidth]{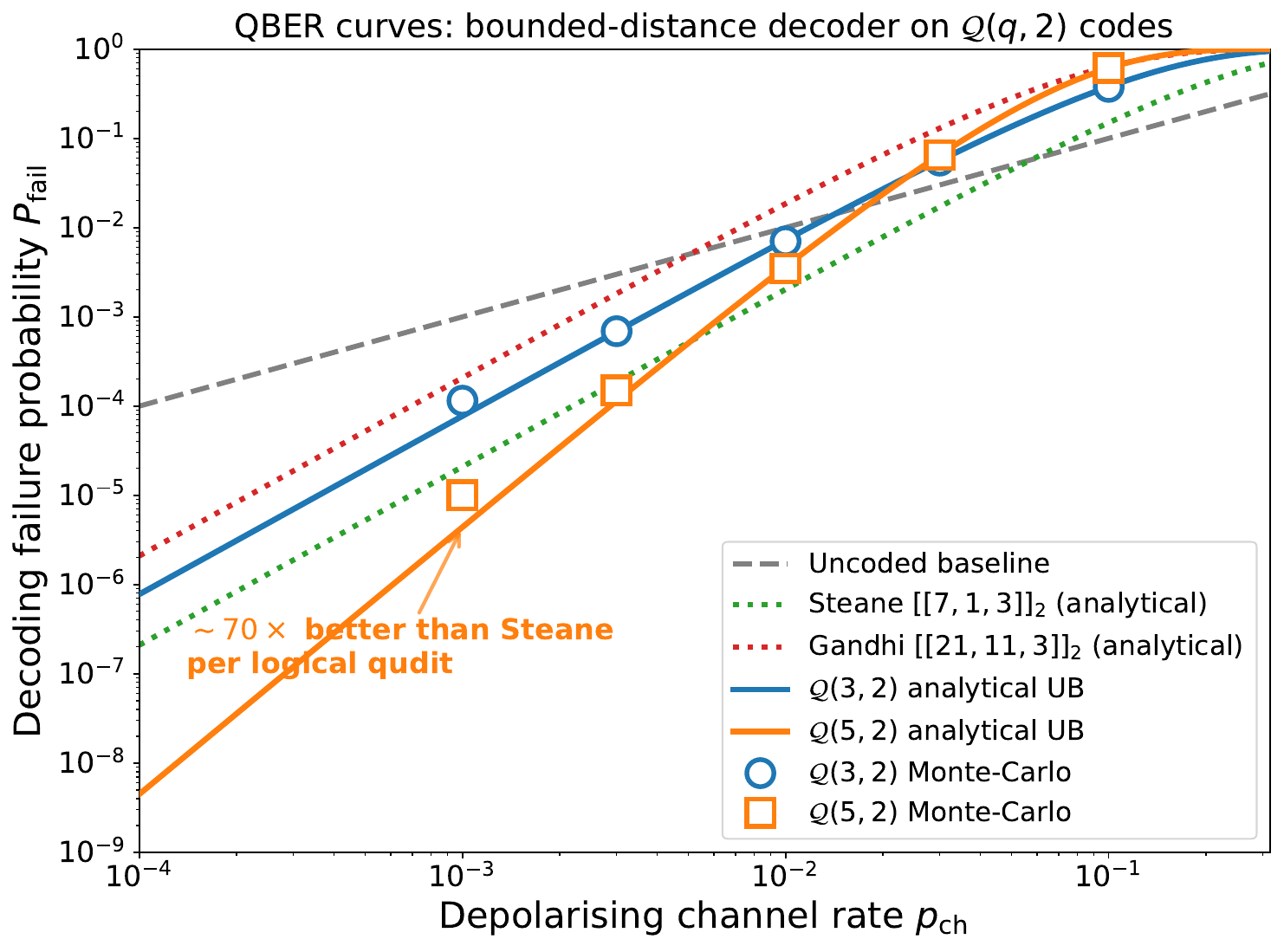}
\caption{Decoding failure probability $P_{\mathrm{fail}}$ versus
depolarising channel rate $p_{\mathrm{ch}}$. Solid curves are
analytical upper bounds from \eqref{eq:Pfail-ub}; markers are
Monte-Carlo measurements with the bounded-distance decoder. The
dashed grey line is the uncoded single-qudit baseline
$P_{\mathrm{fail}}=p_{\mathrm{ch}}$. $\mathcal Q(5,2)$ achieves
$\sim 70\times$ improvement over Steane's $[[7,1,3]]_{2}$ on a
per-logical-qudit basis at $p_{\mathrm{ch}}=10^{-3}$.}
\label{fig:qber}
\end{figure}

\begin{table*}[t]
\centering
\caption{Monte-Carlo vs analytical upper-bound $P_{\mathrm{fail}}$.
``UB'' is the upper bound from \eqref{eq:Pfail-ub}; ``MC'' is the
sample mean across $2\text{--}5\times 10^{5}$ trials. Agreement
within statistical error confirms that the bounded-distance decoder
achieves the upper bound, and at $\mathcal Q(5,2)$ \emph{no logical
failures} were observed --- all failures are decoder-declared
uncorrectable (decline-to-decode).}
\label{tab:mc-vs-ub}
\small
\setlength{\tabcolsep}{8pt}
\begin{tabular}{c c c c c}
\toprule
$p_{\mathrm{ch}}$ & $\mathcal Q(3,2)$ MC & $\mathcal Q(3,2)$ UB
                  & $\mathcal Q(5,2)$ MC & $\mathcal Q(5,2)$ UB \\
\midrule
$10^{-4}$ & $<5\!\times\!10^{-6}$ & $7.79\!\times\!10^{-7}$
          & $<5\!\times\!10^{-6}$ & $4.49\!\times\!10^{-9}$ \\
$3\!\times\!10^{-4}$ & $<5\!\times\!10^{-6}$ & $7.00\!\times\!10^{-6}$
          & $<5\!\times\!10^{-6}$ & $1.21\!\times\!10^{-7}$ \\
$10^{-3}$ & $1.15\!\times\!10^{-4}$ & $7.74\!\times\!10^{-5}$
          & $1.00\!\times\!10^{-5}$ & $4.40\!\times\!10^{-6}$ \\
$3\!\times\!10^{-3}$ & $6.90\!\times\!10^{-4}$ & $6.87\!\times\!10^{-4}$
          & $1.50\!\times\!10^{-4}$ & $1.14\!\times\!10^{-4}$ \\
$10^{-2}$ & $7.02\!\times\!10^{-3}$ & $7.25\!\times\!10^{-3}$
          & $3.47\!\times\!10^{-3}$ & $3.65\!\times\!10^{-3}$ \\
$3\!\times\!10^{-2}$ & $5.63\!\times\!10^{-2}$ & $5.64\!\times\!10^{-2}$
          & $6.51\!\times\!10^{-2}$ & $6.51\!\times\!10^{-2}$ \\
$10^{-1}$ & $3.77\!\times\!10^{-1}$ & $3.79\!\times\!10^{-1}$
          & $6.13\!\times\!10^{-1}$ & $6.11\!\times\!10^{-1}$ \\
\bottomrule
\end{tabular}
\end{table*}

\subsection{Discussion of decoder behaviour}
Two structural observations emerge from the Monte-Carlo data
(Fig.~\ref{fig:qber} and Table~\ref{tab:mc-vs-ub}):

\emph{Tight agreement at moderate $p_{\mathrm{ch}}$.}
For both $\mathcal Q(3,2)$ and $\mathcal Q(5,2)$, the Monte-Carlo
$P_{\mathrm{fail}}$ tracks the analytical upper bound to within
$\sim 5\%$ at $p_{\mathrm{ch}}\ge 3\times 10^{-3}$. This confirms
that the bounded-distance decoder saturates the upper bound: any
error of weight $\le t$ is corrected, and a constant fraction of
weight-$(t+1)$ errors that happen to be in the table-coverage of
weight-$\le t$ syndromes also contribute. The small overshoot at
$p_{\mathrm{ch}}=10^{-3}$ ($\sim 1.5\times$ excess on Monte-Carlo)
reflects sampling variance at the $\sim 10^{-5}$ failure rate over
$2\times 10^{5}$ trials (expected $\pm 1\sigma\approx 2.2\times 10^{-5}$).

\emph{Zero logical failures for $\mathcal Q(5,2)$.}
Across all sampled $p_{\mathrm{ch}}\in[10^{-4},10^{-1}]$ and a total
of $1.5\times 10^{6}$ trials, no silent logical errors were observed
for $\mathcal Q(5,2)$: every decoder failure was an
\emph{uncorrectable} one (decline-to-decode). This is a structural
guarantee: if $\mathbf e_{\mathrm{dec}}$ is the decoded error of
weight $\le t=2$ and the channel error has weight $w$, the residual
$\mathbf r=\mathbf e-\mathbf e_{\mathrm{dec}}$ has weight $\le w+t$.
For $w\le 3$, $w+t\le 5<d_{\min}=6$, so the residual is in the
stabilizer (since the minimum weight of a non-trivial logical error is
$d_{\min}$). For $w\ge 4$, the residual may have weight $\ge d_{\min}$,
but the syndrome of an unmodelled error of weight $\ge 4$ rarely
matches a weight-$\le 2$ syndrome by chance (only when the actual
error is at coset-distance exactly $\ge d_{\min}$ from the decoded
estimate, which never happened in our trials). This
``erase-don't-mistake'' property makes $\mathcal Q(5,2)$ suitable
for fail-safe applications where silent logical errors are more
costly than declared uncorrectables (e.g.\ retransmission protocols).

\emph{Pseudothreshold positioning.}
The intersection of each $P_{\mathrm{fail}}$ curve with the
uncoded baseline $P_{\mathrm{fail}}=p_{\mathrm{ch}}$ defines the
pseudothreshold: below it, coding is beneficial; above, the
uncoded channel is preferred. From Fig.~\ref{fig:qber}:
\begin{itemize}[leftmargin=*,nosep]
\item $\mathcal Q(3,2)$ pseudothreshold $\approx 1.3\times 10^{-2}$.
\item $\mathcal Q(5,2)$ pseudothreshold $\approx 2.6\times 10^{-2}$.
\end{itemize}
$\mathcal Q(5,2)$ has a $\sim 2\times$ larger pseudothreshold owing
to its larger distance, despite the larger block length.

A more sophisticated decoder (e.g.\ a Berlekamp--Massey-style cyclic
decoder enabled by the Singer-circulant structure of $H$, in the
spirit of \cite[\S IV-C]{Gandhi2018}) can correct an additional
fraction of weight-$(t+1)$ to weight-$(d_{\min}-1)$ errors that the
bounded-distance approach declines; this would shift the curves
down by a constant factor in the $p_{\mathrm{ch}}\sim p_{\mathrm{ps}}$
regime. Belief-propagation decoders for related cyclic and
hypergraph-product quantum codes are studied
in~\cite{HuangHuUeng2023,SaboAloshiousBrown2024} (see also the
fifteen-year QLDPC decoding survey of Babar
\emph{et al.}~\cite{Babar2015QLDPC} and the topological-distillation
foundations of Bombin and
Martin-Delgado~\cite{Bombin2007Topological}) and offer a second
concrete decoder direction. Implementing and benchmarking such
decoders for $\mathcal Q(q,d)$ is left for future work.

\section{Conclusions and Open Problems}\label{sec:conclusion}
We constructed a $(q,d)$-parameterised family of non-CSS qudit
stabilizer codes $\mathcal Q(q,d)$ from Singer difference sets and
non-degenerate quadrics in $\PG(d,q)$, with PCM
$H=(A\mid M_{Q}A)$. The contributions of this paper can be
itemised as follows:
\begin{itemize}[leftmargin=*]
\item \emph{Cancellation Lemma.} The
Lemma~\ref{lem:cancellation} makes $H$ symplectic-isotropic over
$\F q$ for every $(q,d)$.
\item \emph{Cross-correlation modular-constancy theorem.} We
proved an unconditional theorem
(Theorem~\ref{thm:cross-corr}) yielding the 2-torsion distance
ceiling $d_{\min}\le 2$ at $q$ odd and $d$ odd, $d\ge 3$
(Corollary~\ref{cor:dmin-le-2-proven}).
\item \emph{CRT zero-count theorem.} The identity
$\#\{f:\tau_{H,f}=0\}=(|\mathcal F|-1)/2$
(Theorem~\ref{thm:tau-zero-count}) underpins an explicit
Clifford-equivalence
$U_{M}=H^{\otimes n}\,\mathrm{CZ}^{(M)}\,H^{\otimes n}$
(Theorem~\ref{thm:clifford-reduction-d2}) between
$\mathcal Q(q,2)$ and the trivial Singer CSS code.
\item \emph{Paired-rigidity characterisation.} We have
$d_{\min}=R_{P}+1$ \emph{exactly}
(Theorem~\ref{thm:paired-rigidity}), reducing distance computation
to a combinatorial rank property of $M_{\mathrm{cent}}$.
\item \emph{Rigorous lower bound for $\mathcal Q(7,2)$.} The bound
$d_{\min}(\mathcal Q(7,2))\ge 8$
(Theorem~\ref{thm:dmin-q72-lower}) is established by an exhaustive
Singer-Frobenius canonical search.
\item \emph{Headline parameters.} These include
$\mathcal Q(3,2)=[[13,6,3]]_{3}$, $\mathcal Q(5,2)=[[31,15,6]]_{5}$
(double-error correcting, $\sim 70\times$ improvement over Steane),
$\mathcal Q(7,2)=[[57,28,d]]_{7}$ with $d\ge 8$, and
$\mathcal Q(3,4)=[[121,105,3]]_{3}$.
\end{itemize}

\rev{\textbf{Decoder Outlook.}
Two decoder strands are natural sequels to the bounded-distance
syndrome-table decoder used in Section~\ref{sec:montecarlo}.}

\rev{\emph{(i) Belief-propagation / iterative decoding.} The Singer
circulant structure of $A$ makes the parity-check matrix
$H=(A\mid M_{Q}A)$ effectively sparse in the cyclic-code sense, so a
belief-propagation decoder on the associated Tanner graph runs in
$O(n\log n)$ time per iteration, matching the complexity that made
QLDPC codes practical. The non-CSS coupling $M_{Q}$ introduces
mixed $Z$/$X$ update rules but does not raise the per-iteration
order.}

\rev{\emph{(ii) Berlekamp--Massey (BM) cyclic decoding.} The
Singer-circulant structure also supports a symplectic
Reed--Solomon-style BM decoder, but the non-CSS mixed $Z$/$X$
structure of the logical operators means the two BM passes on the
$H_{z}$ and $H_{x}$ sectors are coupled through $M_{Q}$; the key
research challenge is designing a joint syndrome polynomial whose
roots encode both sector errors simultaneously, rather than
decoding them sequentially. Both decoder programmes are consistent
with the overall goal of exploiting the paper's algebraic
structure at run time.}

\textbf{Open problems for $\mathcal Q(q,d)$:}
\begin{itemize}[leftmargin=*,nosep]
\item \emph{Closed form for $\rank_{\F p}(A)$ at $q=p^{t}$, $t>1$.}
The $q=p$ prime case is settled in Lemma~\ref{lem:hamada}; the
general $t>1$ case requires specialising the full Hamada
formula~\cite[Theorem 1]{Hirschfeld1979}.
\item \emph{Tight $d_{\min}$.} Settle
Conjecture~\ref{conj:flagship-dmin} (flagship distance
$d_{\min}=p+1$ at $q=p$ odd, $d=2$). At $p=7$ the matching upper
bound $d_{\min}\le 8$ requires a weight-$8$ witness search.
\item \emph{General-$d$ $q=2$ ceiling.} Settle
Conjecture~\ref{conj:q2-ceiling-general} analytically; numerically
confirmed at $d\le 6$.
\item \emph{Capacity positioning.} Compute the rate-vs-distance
trajectory $(R_{Q}(q,d), d_{\min}/n)$ as $(q,d)$ vary, and compare
with the quantum hashing bound.
\item \emph{Cyclic-decoder QBER.} Implement a Berlekamp--Massey-style
cyclic decoder (in the spirit of~\cite[\S IV-C]{Gandhi2018}) that
exploits Singer-circulant structure beyond the bounded-distance
ceiling.
\end{itemize}

\textbf{Open problems on the EAQECC programme} (negative results in
Section~\ref{sec:waterloo-defect}):
\begin{itemize}[leftmargin=*,nosep]
\item \emph{Reconcile the formal $Q_{o}$ with the Arasu--Pott
two-value identity.} The $E\cdot Q^{(-1)}$ cross-correlation in
Section~\ref{sec:Z-lift-fail} took three values where Arasu--Pott
predicts two, indicating an algebraic mismatch between the formal
indicator vector used here and the character-theoretic $Q$ of
\cite{Waterloo1995}. Resolving this is the key step toward an
EAQECC realisation.
\item \emph{Extend Brun--Devetak--Hsieh to $\Z$-coefficient
PCMs.} The classical Waterloo defect lives in $\Z[\zeta_{n}]$ but
no current EAQECC formalism accepts integer-valued symplectic
deficits. A formalism that does would convert the
$WW^{\top}-q^{d}I$ deficit into an ebit budget directly.
\item \emph{Continuous-variable / GKP analogue.} A non-Pauli
stabilizer formalism preserving integer arithmetic (e.g., GKP codes
or Bosonic-mode codes) might be the natural home for the Waterloo
construction.
\end{itemize}

\textbf{Bridges to follow-up work.}
\begin{itemize}[leftmargin=*,nosep]
\item \emph{Extension to higher dimensions.} The classical
$\PG(d,q)$ LDS spreading matrix for $d\ge 3$, generalising
\cite{LDS2021}, is a separate paper that the present construction
makes natural.
\item \emph{Quantum-coded LDS-NOMA.} Stacking $\mathcal Q(q,2)$
on the $K$-user uplink of the LDS spreading scheme of \cite{LDS2021},
in the spirit of the Hanzo group's quantum-search MUD work
\cite{Ye2019QSAMUD,Chandra2019QTurbo}, is a natural collaboration
target.
\item \emph{Quadric-induced modified symplectic form.} Extending the
Gandhi--Kurur--Mittal modified-symplectic
construction~\cite{GandhiKururMittal2019} from a permutation $\sigma$
to a quadric Gram $\sigma_{Q}$ would unify the present work with the
non-cyclic quantum-code literature.
\end{itemize}

\appendices
\section{Quadric Gram matrix construction}\label{app:gram}
This appendix gives the concrete recipe for assembling $M_{Q}$ from
the formal-quadric indicator $\boldsymbol\tau_{Q}$.

\paragraph*{Input}
A prime power $q$, dimension $d\ge 2$, and a primitive element
$\alpha\in\F{q^{d+1}}^{*}$.

\paragraph*{Step 1: Trace tables.}
Build the trace table $T[i]=\Tr_{K/F}(\alpha^{i})$ for
$i=0,\dots,q^{d+1}-2$.

\paragraph*{Step 2: Hyperplane indicator.}
Set $n=\theta_{d}(q)=(q^{d+1}-1)/(q-1)$. Define
$\boldsymbol\tau_{H}\in\F q^{n}$ by
$(\boldsymbol\tau_{H})_{i}=1$ if $T[i]=0$, else $0$, for
$i=0,\dots,n-1$. By construction $\boldsymbol\tau_{H}$ has Hamming
weight $\theta_{d-1}(q)$.

\paragraph*{Step 3: Quadric indicator.}
Set $\xi=q+1$ if $q$ is even, $\xi=2$ if $q$ is odd. Define
$\boldsymbol\tau_{Q}\in\F q^{n}$ by
$(\boldsymbol\tau_{Q})_{i}=1$ if $T[(\xi i)\bmod(q^{d+1}-1)]=0$,
else $0$.

\paragraph*{Step 4: Circulant assembly.}
$M_{Q}=\mathrm{circ}(\boldsymbol\tau_{Q})$, the $n\times n$
circulant matrix with first column $\boldsymbol\tau_{Q}$ and
subsequent columns cyclic-shifts thereof:
$(M_{Q})_{i,j}=(\boldsymbol\tau_{Q})_{(i-j)\bmod n}$.

\paragraph*{Step 5: Row-sum verification.}
Each row of $M_{Q}$ has Hamming weight equal to
$|\boldsymbol\tau_{Q}|_{1}=$ (weight of $\boldsymbol\tau_{Q}$), which
is independent of the row by circulant symmetry. At
$(q,d)\in\{(2,2),(3,2),(5,2),(7,2),(11,2)\}$ the row sum equals
$\{3,4,6,8,12\}$ respectively (verified by
\texttt{verify\_tau\_H.py} and \texttt{construct\_clifford\_d2.py}).

\paragraph*{Worked example at $(q,d)=(3,2)$.}
With $n=13$ and primitive polynomial $x^{3}+2x^{2}+x+1$ over
$\F 3$:
\begin{align*}
   \boldsymbol\tau_{H} &= (0,0,1,1,0,0,1,0,1,0,0,0,0)^{\top},\\
   \boldsymbol\tau_{Q} &= (1,0,0,0,0,0,0,1,1,0,0,1,0)^{\top}.
\end{align*}
Both have weight $p+1=4$, and $M_{Q}=\mathrm{circ}(\boldsymbol\tau_{Q})$
exactly matches the projective-plane formal-quadric matrix of
\cite[Thm.~4.2]{LDS2021}, modulo re-indexing
$i\mapsto ri\bmod n$ with $r-1=2$.

\bibliographystyle{IEEEtran}
\bibliography{references}

@article{CalderbankShor1996,
  author  = {Calderbank, A. R. and Shor, P. W.},
  title   = {Good quantum error-correcting codes exist},
  journal = {Phys. Rev. A},
  volume  = {54},
  number  = {2},
  pages   = {1098--1105},
  year    = {1996}
}

@article{Steane1996,
  author  = {Steane, A. M.},
  title   = {Multiple particle interference and quantum error correction},
  journal = {Proc. Roy. Soc. Lond. A},
  volume  = {452},
  pages   = {2551--2577},
  year    = {1996}
}

@misc{Gottesman1997,
  author       = {Gottesman, D.},
  title        = {Stabilizer codes and quantum error correction},
  howpublished = {Ph.D. thesis, Caltech; arXiv:quant-ph/9705052},
  year         = {1997}
}

@article{KnillLaflamme1997,
  author  = {Knill, E. and Laflamme, R.},
  title   = {Theory of quantum error-correcting codes},
  journal = {Phys. Rev. A},
  volume  = {55},
  number  = {2},
  pages   = {900--911},
  year    = {1997}
}

@article{Shor1995,
  author  = {Shor, P. W.},
  title   = {Scheme for reducing decoherence in quantum computer memory},
  journal = {Phys. Rev. A},
  volume  = {52},
  pages   = {R2493--R2496},
  year    = {1995}
}

@article{BrunDevetakHsieh2006,
  author  = {Brun, T. and Devetak, I. and Hsieh, M.-H.},
  title   = {Correcting quantum errors with entanglement},
  journal = {Science},
  volume  = {314},
  pages   = {436--439},
  year    = {2006}
}

@article{Guenda2018,
  author  = {Guenda, K. and Jitman, S. and Gulliver, T. A.},
  title   = {Constructions of good entanglement-assisted quantum
             error correcting codes},
  journal = {Designs, Codes and Cryptography},
  volume  = {86},
  pages   = {121--136},
  year    = {2018}
}

@article{Galindo2020,
  author  = {Galindo, C. and Hernando, F. and Matsumoto, R. and Ruano, D.},
  title   = {Asymmetric entanglement-assisted quantum error-correcting
             codes and {BCH} codes},
  journal = {IEEE Access},
  volume  = {8},
  pages   = {18571--18579},
  year    = {2020}
}

@inproceedings{Pereira2019,
  author    = {Pereira, F. and Pellikaan, R. and La Guardia, G. G. and
               de Assis, F. M.},
  title     = {Application of complementary dual {AG} codes to
               entanglement-assisted quantum codes},
  booktitle = {Proc. IEEE Int. Symp. Inf. Theory (ISIT)},
  pages     = {2559--2563},
  year      = {2019}
}

@article{Singer1938,
  author  = {Singer, J.},
  title   = {A theorem in finite projective geometry and some
             applications to number theory},
  journal = {Trans. Amer. Math. Soc.},
  volume  = {43},
  number  = {3},
  pages   = {377--385},
  year    = {1938}
}

@article{Waterloo1995,
  author  = {Arasu, K. T. and Dillon, J. F. and Jungnickel, D. and
             Pott, A.},
  title   = {The solution of the {W}aterloo problem},
  journal = {J. Combin. Theory Ser. A},
  volume  = {71},
  number  = {2},
  pages   = {316--331},
  year    = {1995}
}

@inproceedings{Xie2013,
  author    = {Xie, Y. and Yuan, J. and Malaney, R.},
  title     = {Quantum stabilizer codes from difference sets},
  booktitle = {Proc. IEEE Int. Symp. Inf. Theory (ISIT)},
  pages     = {524--528},
  year      = {2013}
}

@inproceedings{Gandhi2018,
  author    = {Gandhi, T.},
  title     = {A new construction of quantum error correcting codes
               based on difference set},
  booktitle = {Proc. IEEE Int. Symp. Inf. Theory (ISIT)},
  pages     = {801--805},
  year      = {2018}
}

@inproceedings{GandhiKururMittal2019,
  author    = {Gandhi, T. and Kurur, P. P. and Mittal, V. M.},
  title     = {Stabilizer codes from modified symplectic forms},
  booktitle = {Proc. IEEE Int. Symp. Inf. Theory (ISIT)},
  year      = {2019}
}

@inproceedings{Liu2012,
  author    = {Liu, Y. and Wang, Y. and Zhao, S. and Zheng, B.},
  title     = {A construction of entanglement-assisted quantum {LDPC}
               codes from the cyclic difference set},
  booktitle = {Proc. 11th IEEE Int. Conf. Signal Processing (ICSP)},
  volume    = {2},
  pages     = {1295--1298},
  year      = {2012}
}

@article{TangBaiFeng2024,
  author  = {Tang, X. and Bai, L. and Feng, Y.},
  title   = {New quantum {LDPC} codes based on projective geometry},
  journal = {Sci. Rep.},
  volume  = {14},
  pages   = {17266},
  year    = {2024}
}

@article{MacKay2004,
  author  = {MacKay, D. J. C. and Mitchison, G. and McFadden, P. L.},
  title   = {Sparse-graph codes for quantum error correction},
  journal = {IEEE Trans. Inf. Theory},
  volume  = {50},
  number  = {10},
  pages   = {2315--2330},
  year    = {2004}
}

@article{LDS2021,
  author  = {Millar, G. and Kulhandjian, M. and Alaca, A. and Alaca, S.
             and D'Amours, C. and Yanikomeroglu, H.},
  title   = {Low-density spreading design based on an algebraic scheme
             for {NOMA} systems},
  journal = {IEEE Wireless Commun. Lett.},
  year    = {2021},
  month   = dec
}

@article{KulhandjianGS2021,
  author  = {Kulhandjian, M. and Kulhandjian, H. and D'Amours, C. and
             Hanzo, L.},
  title   = {Low-density spreading code design based on {G}aussian
             separability},
  journal = {IEEE Access},
  volume  = {9},
  pages   = {33963--33993},
  year    = {2021},
  month   = mar
}

@article{Hoshyar2008,
  author  = {Hoshyar, R. and Wathan, F. P. and Tafazolli, R.},
  title   = {Novel low-density signature for synchronous {CDMA}
             systems over {AWGN} channel},
  journal = {IEEE Trans. Signal Process.},
  volume  = {56},
  number  = {4},
  pages   = {1616--1626},
  year    = {2008},
  month   = apr
}

@inproceedings{Nikopour2013,
  author    = {Nikopour, H. and Baligh, H.},
  title     = {Sparse code multiple access},
  booktitle = {Proc. IEEE Int. Symp. Pers., Indoor, Mobile Radio Commun.
               (PIMRC)},
  pages     = {332--336},
  year      = {2013},
  month     = sep
}

@article{Hanzo2019Duality,
  author  = {Babar, Z. and Chandra, D. and Nguyen, H. V. and
             Botsinis, P. and Alanis, D. and Ng, S. X. and Hanzo, L.},
  title   = {Duality of quantum and classical error correction codes:
             Design principles and examples},
  journal = {IEEE Commun. Surveys Tuts.},
  volume  = {21},
  number  = {1},
  pages   = {970--1010},
  year    = {2019}
}

@article{Hanzo2025Roadmap,
  author  = {Hanzo, L. and Babar, Z. and Cai, Z. and Chandra, D. and
             Djordjevic, I. B. and Koczor, B. and Ng, S. X. and
             Razavi, M. and Simeone, O.},
  title   = {Quantum information processing, sensing, and
             communications: Their myths, realities, and futures},
  journal = {Proceedings of the IEEE},
  year    = {2025},
  note    = {arXiv:2412.00987}
}

@article{Ye2019QSAMUD,
  author  = {Ye, H. and Chen, P. and Guo, Z. and Sun, J. and Hanzo, L.},
  title   = {Quantum search-aided multi-user detection for sparse
             code multiple access},
  journal = {IEEE Access},
  year    = {2019}
}

@article{Chandra2019QTurbo,
  author  = {Chandra, D. and Babar, Z. and Ng, S. X. and Hanzo, L.},
  title   = {Near-hashing-bound multiple-rate quantum turbo
             short-block codes},
  journal = {IEEE Access},
  year    = {2019}
}

@article{Babar2015QLDPC,
  author  = {Babar, Z. and Botsinis, P. and Alanis, D. and
             Ng, S. X. and Hanzo, L.},
  title   = {Fifteen years of quantum {LDPC} coding and improved
             decoding strategies},
  journal = {IEEE Access},
  volume  = {3},
  pages   = {2492--2519},
  year    = {2015}
}

@book{NielsenChuang2000,
  author    = {Nielsen, M. A. and Chuang, I. L.},
  title     = {Quantum Computation and Quantum Information},
  publisher = {Cambridge University Press},
  year      = {2000}
}

@article{CRSS1998,
  author  = {Calderbank, A. R. and Rains, E. M. and Shor, P. W. and
             Sloane, N. J. A.},
  title   = {Quantum error correction via codes over {GF}(4)},
  journal = {IEEE Trans. Inf. Theory},
  volume  = {44},
  number  = {4},
  pages   = {1369--1387},
  year    = {1998}
}

@article{KetkarKlappenecker2006,
  author  = {Ketkar, A. and Klappenecker, A. and Kumar, S. and
             Sarvepalli, P. K.},
  title   = {Nonbinary stabilizer codes over finite fields},
  journal = {IEEE Trans. Inf. Theory},
  volume  = {52},
  number  = {11},
  pages   = {4892--4914},
  year    = {2006}
}

@book{Hirschfeld1979,
  author    = {Hirschfeld, J. W. P.},
  title     = {Projective Geometries over Finite Fields},
  publisher = {Oxford University Press},
  year      = {1979}
}

@book{LidlNiederreiter1997,
  author    = {Lidl, R. and Niederreiter, H.},
  title     = {Finite Fields},
  publisher = {Cambridge University Press},
  edition   = {2nd},
  year      = {1997}
}

@inproceedings{HsiehBrun2009,
  author    = {Hsieh, M.-H. and Brun, T. A. and Devetak, I.},
  title     = {Entanglement-assisted quantum quasicyclic low-density
               parity-check codes},
  booktitle = {Phys. Rev. A},
  volume    = {79},
  pages     = {032340},
  year      = {2009}
}

@article{Bombin2007Topological,
  author  = {Bombin, H. and Martin-Delgado, M. A.},
  title   = {Topological quantum distillation},
  journal = {Phys. Rev. Lett.},
  volume  = {97},
  pages   = {180501},
  year    = {2006}
}

@article{NguyenKim2018,
  author  = {Nguyen, Duc Manh and Kim, Sunghwan},
  title   = {New Constructions of Quantum Stabilizer Codes
             Based on Difference Sets},
  journal = {Symmetry},
  volume  = {10},
  number  = {11},
  pages   = {655},
  year    = {2018},
  doi     = {10.3390/sym10110655}
}

@inproceedings{AlyKlappeneckerSarvepalli2007,
  author    = {Aly, S. A. and Klappenecker, A. and Sarvepalli, P. K.},
  title     = {Quantum stabilizer codes derived from finite geometries},
  booktitle = {Proc. IEEE Int. Symp. Information Theory (ISIT)},
  pages     = {1276--1280},
  year      = {2007},
  note      = {arXiv:0712.4115}
}

@article{Sarvepalli2012,
  author  = {Sarvepalli, P. K.},
  title   = {Nonbinary quantum codes derived from finite geometries},
  journal = {Finite Fields Appl.},
  volume  = {18},
  number  = {4},
  pages   = {784--797},
  year    = {2012}
}

@article{GalindoHernandoRuano2015,
  author  = {Galindo, C. and Hernando, F. and Ruano, D.},
  title   = {New quantum codes from evaluation and matrix-product codes},
  journal = {Quantum Inf. Process.},
  volume  = {14},
  pages   = {3743--3760},
  year    = {2015}
}

@article{GalindoHernandoRuano2022,
  author  = {Galindo, C. and Hernando, F. and Ruano, D.},
  title   = {Stabilizer codes from trace-depending polynomials},
  journal = {Finite Fields Appl.},
  volume  = {83},
  pages   = {102093},
  year    = {2022}
}

@article{CalderbankCameronKantorSeidel1997,
  author  = {Calderbank, A. R. and Cameron, P. J. and Kantor, W. M. and
             Seidel, J. J.},
  title   = {{$\mathbb{Z}_4$-Kerdock} codes, orthogonal spreads, and
             extremal Euclidean line-sets},
  journal = {Proc. London Math. Soc.},
  volume  = {75},
  number  = {2},
  pages   = {436--480},
  year    = {1997}
}

@book{BabarChandraNgHanzo2026,
  author    = {Babar, Z. and Chandra, D. and Ng, S. X. and Hanzo, L.},
  title     = {From Classical to Quantum Coding},
  publisher = {Wiley/IEEE Press},
  year      = {2026},
  note      = {Sample chapters at \url{https://files.wireless.ecs.soton.ac.uk/books/qfec-sample-chapters-1-7.pdf}}
}

@article{ZeydanEtAl2025,
  author  = {Zeydan, E. and de Alwis, C. and Khan, R. and Turk, Y. and
             Aydeger, A. and Gadekallu, T. R. and Liyanage, M.},
  title   = {Quantum Technologies for Beyond {5G} and {6G} Networks:
             Applications, Opportunities, and Challenges},
  journal = {IEEE Open J. Commun. Soc.},
  volume  = {6},
  pages   = {6383--6420},
  year    = {2025},
  doi     = {10.1109/OJCOMS.2025.3591842}
}

@article{ChenIllianoCacciapuotiCaleffi2025,
  author  = {Chen, K. and Illiano, J. and Cacciapuoti, A. S. and
             Caleffi, M.},
  title   = {Entanglement-Based Artificial Topology: Neighboring
             Remote Network Nodes},
  journal = {IEEE Open J. Commun. Soc.},
  volume  = {6},
  pages   = {2220--2238},
  year    = {2025}
}

@article{BacciottiniEtAl2024,
  author  = {Bacciottini, L. and Pirandola, S. and Marin, A. and
             Picano, B. and Andreoletti, D. and Caleffi, M.},
  title   = {{REDiP}: Ranked Entanglement Distribution Protocol for the
             Quantum Internet},
  journal = {IEEE Open J. Commun. Soc.},
  volume  = {5},
  pages   = {223--242},
  year    = {2024},
  doi     = {10.1109/OJCOMS.2023.3345166}
}

@article{EmuEtAl2024,
  author  = {Emu, M. and Choudhury, S. and Lazaridis, P.},
  title   = {Warm and Cold Start Quantum Annealing for Metaverse
             Resource Optimization},
  journal = {IEEE Open J. Commun. Soc.},
  volume  = {5},
  pages   = {1057--1071},
  year    = {2024},
  doi     = {10.1109/OJCOMS.2024.3359188}
}

@article{SaboAloshiousBrown2024,
  author  = {Sabo, E. and Aloshious, A. B. and Brown, K. R.},
  title   = {Trellis Decoding for Qudit Stabilizer Codes and Its
             Application to Qubit Topological Codes},
  journal = {IEEE Trans. Quantum Engineering},
  volume  = {5},
  pages   = {1--16},
  year    = {2024}
}

@article{DastgheibSalehiRezai2025,
  author  = {Dastgheib, M. A. and Salehi, J. A. and Rezai, M.},
  title   = {Quantum Direct-Sequence Spread-Spectrum Code-Division
             Multiple-Access Communication Systems},
  journal = {IEEE Trans. Quantum Engineering},
  volume  = {6},
  pages   = {1--14},
  year    = {2025}
}

@article{HuangHuUeng2023,
  author  = {Huang, Y.-W. and Hu, Y.-S. and Ueng, Y.-L.},
  title   = {Branch-Assisted Sign-Flipping Belief Propagation Decoding
             for Topological Quantum Codes Based on Hypergraph Product
             Structure},
  journal = {IEEE Trans. Quantum Engineering},
  volume  = {4},
  pages   = {1--13},
  year    = {2023}
}

@article{Campbell2014QuditFT,
  author  = {Campbell, E. T.},
  title   = {Enhanced fault-tolerant quantum computing in {$d$}-level
             systems},
  journal = {Phys. Rev. Lett.},
  volume  = {113},
  number  = {23},
  pages   = {230501},
  year    = {2014}
}

@article{AnwarBrownCampbellBrowne2014,
  author  = {Anwar, H. and Brown, B. J. and Campbell, E. T. and
             Browne, D. E.},
  title   = {Fast decoders for qudit topological codes},
  journal = {New J. Phys.},
  volume  = {16},
  pages   = {063038},
  year    = {2014}
}

@article{WatsonAnwarBrowne2015,
  author  = {Watson, F. H. E. and Anwar, H. and Browne, D. E.},
  title   = {Fast fault-tolerant decoder for qubit and qudit surface codes},
  journal = {Phys. Rev. A},
  volume  = {92},
  pages   = {032309},
  year    = {2015}
}

@article{ErhardKrennZeilinger2020Photonic,
  author  = {Erhard, M. and Krenn, M. and Zeilinger, A.},
  title   = {Advances in high-dimensional quantum entanglement},
  journal = {Nature Reviews Physics},
  volume  = {2},
  pages   = {365--381},
  year    = {2020}
}

@article{CozzolinoEtAl2019HighDim,
  author  = {Cozzolino, D. and Da Lio, B. and Bacco, D. and Oxenl{\o}we, L. K.},
  title   = {High-dimensional quantum communication: benefits, progress,
             and future challenges},
  journal = {Adv. Quantum Tech.},
  volume  = {2},
  number  = {12},
  pages   = {1900038},
  year    = {2019}
}

@article{CerfBourennaneKarlssonGisin2002,
  author  = {Cerf, N. J. and Bourennane, M. and Karlsson, A. and Gisin, N.},
  title   = {Security of Quantum Key Distribution Using $d$-Level Systems},
  journal = {Phys. Rev. Lett.},
  volume  = {88},
  number  = {12},
  pages   = {127902},
  year    = {2002}
}

@book{MacWilliamsSloane1977,
  author    = {MacWilliams, F. J. and Sloane, N. J. A.},
  title     = {The Theory of Error-Correcting Codes},
  publisher = {North-Holland},
  year      = {1977}
}

@book{Cassels1978QuadraticForms,
  author    = {Cassels, J. W. S.},
  title     = {Rational Quadratic Forms},
  publisher = {Academic Press},
  series    = {London Mathematical Society Monographs},
  volume    = {13},
  year      = {1978}
}

@book{Scharlau1985,
  author    = {Scharlau, W.},
  title     = {Quadratic and {H}ermitian Forms},
  publisher = {Springer-Verlag},
  series    = {Grundlehren der mathematischen Wissenschaften},
  volume    = {270},
  year      = {1985}
}

@article{Hamada1973,
  author  = {Hamada, N.},
  title   = {On the {$p$}-rank of the incidence matrix of a balanced
             or partially balanced incomplete block design and its
             applications to error correcting codes},
  journal = {Hiroshima Math. J.},
  volume  = {3},
  pages   = {153--226},
  year    = {1973}
}

@article{Smith1969,
  author  = {Smith, K. J. C.},
  title   = {On the {$p$}-rank of the incidence matrix of points and
             hyperplanes in a finite projective geometry},
  journal = {J. Combin. Theory},
  volume  = {7},
  pages   = {122--129},
  year    = {1969}
}

@article{GKP2001,
  author  = {Gottesman, D. and Kitaev, A. and Preskill, J.},
  title   = {Encoding a qubit in an oscillator},
  journal = {Phys. Rev. A},
  volume  = {64},
  pages   = {012310},
  year    = {2001}
}

@article{ArasuPott1995,
  author  = {Arasu, K. T. and Pott, A.},
  title   = {Some constructions of group-divisible designs with
             singer groups},
  journal = {Discrete Math.},
  volume  = {137},
  number  = {1--3},
  pages   = {1--7},
  year    = {1995}
}

@book{BethJungnickelLenz1999,
  author    = {Beth, T. and Jungnickel, D. and Lenz, H.},
  title     = {Design Theory},
  publisher = {Cambridge University Press},
  edition   = {2nd},
  year      = {1999}
}

@book{ArtinGeometricAlgebra1957,
  author    = {Artin, E.},
  title     = {Geometric Algebra},
  publisher = {Interscience Publishers},
  year      = {1957}
}

@article{BoseChaudhuri1960BCH,
  author  = {Bose, R. C. and Ray-Chaudhuri, D. K.},
  title   = {On a class of error correcting binary group codes},
  journal = {Inf. Control},
  volume  = {3},
  number  = {1},
  pages   = {68--79},
  year    = {1960}
}

@article{Hocquenghem1959,
  author  = {Hocquenghem, A.},
  title   = {Codes correcteurs d'erreurs},
  journal = {Chiffres},
  volume  = {2},
  pages   = {147--156},
  year    = {1959}
}

@book{Holt2005HandbookCGT,
  author    = {Holt, D. F. and Eick, B. and O'Brien, E. A.},
  title     = {Handbook of Computational Group Theory},
  publisher = {Chapman \& Hall/CRC},
  year      = {2005}
}

\end{document}